\documentclass{imsart}

\RequirePackage{amsthm,amsmath,amsfonts,amssymb}
\RequirePackage[numbers]{natbib}
\RequirePackage[colorlinks,citecolor=blue,urlcolor=blue]{hyperref}
\RequirePackage{graphicx}

\usepackage[linesnumbered,vlined]{algorithm2e}
\usepackage{dsfont}
\usepackage{caption}
\usepackage{color}
\usepackage{enumitem}
\usepackage{ifthen}
\usepackage{mathtools}  
\usepackage{stmaryrd}
\usepackage{subcaption}
\usepackage{tikz}
\usetikzlibrary{fit,positioning}
\usepackage{xargs}

\startlocaldefs
\numberwithin{equation}{section}

\newtheorem{theorem}{Theorem}[section]
\newtheorem{lemma}[theorem]{Lemma}
\newtheorem{proposition}[theorem]{Proposition}
\newtheorem{corollary}[theorem]{Corollary}
\theoremstyle{remark}

\newtheorem{remark}[theorem]{Remark}
\newtheorem{example}{Example}

\newcommandx{\accprob}[1][1=]{
\ifthenelse{\equal{#1}{}}{\alpha^{\mbox{\footnotesize{\textsf{R}}}}}{\alpha^{\mbox{\footnotesize{\textsf{MH}}}}}
}
\newcommandx{\accprobext}[1][1=]{
\ifthenelse{\equal{#1}{}}{\bar{\alpha}^{\mbox{\footnotesize{\textsf{R}}}}}{\bar{\alpha}^{\mbox{\footnotesize{\textsf{MH}}}}}
}
\newcommandx{\addf}[2][1=]{
\ifthenelse{\equal{#1}{}}{\termletter_{#2}}{\bar{h}_{#2 | #1}}
}
\newcommand{\adjfuncforward}[1]{\vartheta_{#1}}
\newcommand{\bi}[3]{J_{#1}^{(#2, #3)}}
\newcommandx{\bkmod}[2][1=]{ 
\ifthenelse{\equal{#1}{}}
{\kernel{B}_{#2}^\precpar}
{\kernel{B}_{#2}^\precpar}
}
\newcommandx{\bkw}[2][1=]{ 
\ifthenelse{\equal{#1}{}}
{\kernel{B}_{#2}}
{\kernel{B}_{#2}}
}
\newcommand{\bmf}[1]{\mathsf{F}(#1)}
\newcommand{\bmaf}[1]{\mathsf{A}(#1)}

\newcommand{\calF}[2]{\mathcal{F}_{#1}^{#2}}
\newcommand{\calG}[2]{\mathcal{G}_{#1}^{#2}}
\newcommand{\cat}{\mathsf{cat}}
\newcommand{\cbound}{c(\sigma_\pm)}
\newcommand{\cboundtd}{\tilde{c}(\sigma_\pm)}
\newcommand{\cev}[1]{\reflectbox{\ensuremath{\vec{\reflectbox{\ensuremath{#1}}}}}}
\newcommand{\cond}{\mid}
\newcommand{\dbound}{d(\sigma_\pm)}
\newcommand{\dlim}{\overset{\mathcal{D}}{\longrightarrow}}
\newcommand{\eg}{\emph{e.g.}}
\newcommand{\epart}[2]{\ensuremath{\xi_{#1}^{#2}}}
\newcommand{\eqdef}{\coloneqq}
\newcommand{\ewght}[2]{\ensuremath{\omega_{#1}^{#2}}}
\newcommand{\fk}[2]{\mathbf{F}_{#1 | #2}}
\newcommand{\ftd}[1]{\tilde{f}_{#1}}
\newcommand{\hd}[1]{q_{#1}} 
\newcommand{\hk}{\kernel{Q}}
\newcommand{\ie}{\emph{i.e.}}
\newcommand{\im}{\operatorname{i}}
\newcommand{\incrementalvar}{\delta}
\newcommand{\ind}[2]{I_{#1}^{#2}}
\newcommand{\init}{\nu}
\newcommand{\initwgtfunc}{w_{-1}}
\newcommand{\intvect}[2]{\llbracket #1, #2 \rrbracket}
\newcommandx{\K}[1][1=]{\ifthenelse{\equal{#1}{}}{{\kletter}}{{M^{#1}}}}
\newcommand{\kernel}[1]{\mathbf{#1}}
\newcommand{\kissforward}[3][]
{\ifthenelse{\equal{#1}{}}{p_{#2}}
{\ifthenelse{\equal{#1}{fully}}{p^{\star}_{#2}}
{\ifthenelse{\equal{#1}{smooth}}{\tilde{r}_{#2}}{\mathrm{erreur}}}}}
\newcommand{\kletter}{M}
\newcommand{\llh}[1]{L_{#1}}
\newcommand{\lwd}[1]{l_{#1}}
\newcommand{\meas}[1]{\mathsf{M}(#1)}
\newcommand{\md}[1]{g_{#1}}

\newcommand{\mk}{\kernel{G}}
\newcommand{\N}{N}

\newcommand{\noshift}{\shiftsymbol^{\precpar}}

\newcommand{\nset}{\mathbb{N}}
\newcommand{\nsetpos}{\mathbb{N}_{> 0}}
\newcommand{\1}{\mathds{1}} 
\newcommand{\partmixt}{\pi}

\newcommand{\pE}{\mathbb{E}}
\newcommand{\pP}{\mathbb{P}}
\newcommand{\pplim}{\stackrel{\pP}{\longrightarrow}}
\newcommandx{\post}[2][1=]{
\ifthenelse{\equal{#1}{}}
	{\phi_{#2}}
	{\phi_{#2}^\N}
}
\newcommandx{\postmod}[2][1=]{
\ifthenelse{\equal{#1}{}}
	{\phi_{#2}^\precpar}
	{\phi_{#2}^\N}
}
\newcommand{\powerset}[1]{\mathcal{P}(#1)}
\newcommand{\precpar}{\varepsilon}
\newcommand{\precparsp}{\mathcal{E}}
\newcommand{\probmeas}[1]{\mathsf{M}_1(#1)}
\newcommand{\prop}[1]{\mathbf{P}_{#1}}
\newcommand{\propdens}[1]{p_{#1}}
\newcommand{\retrok}{\boldsymbol{\mathcal{L}}}
\newcommand{\retroknorm}{\bar{\boldsymbol{\mathcal{L}}}}
\newcommand{\retrokmod}{\boldsymbol{\mathcal{L}}^\precpar}
\newcommand{\retrokmodnorm}{\bar{\boldsymbol{\mathcal{L}}}^{\precpar}}
\newcommand{\retrokmodmod}{\boldsymbol{\mathcal{R}}^\precpar}
\newcommand{\retrokmodmodnorm}{\bar{\boldsymbol{\mathcal{R}}}^\precpar}
\newcommand{\rmd}{d}
\newcommand{\rset}{\ensuremath{\mathbb{R}}}
\newcommand{\rsetnn}{\rset_{\geq 0}}
\newcommand{\rsetpos}{\rset_{> 0}}
\newcommand{\shiftbwd}{\cev{\shiftsymbol}^{\precpar}}
\newcommand{\shiftfwd}{\vec{\shiftsymbol}^{\, \precpar}}
\newcommand{\shiftsymbol}{\varphi}
\newcommand{\sumwght}[2][]{%
\ifthenelse{\equal{#1}{}}{\ensuremath{\Omega_{#2}}}{\ensuremath{\Omega_{#2}^{(#1)}}}}
\newcommand{\sumwghthat}[2][]{%
\ifthenelse{\equal{#1}{}}{\ensuremath{\widehat{\Omega}_{#2}}}{\ensuremath{\widehat{\Omega}_{#2}^{(#1)}}}}
\newcommand{\tensprod}{\varotimes}
\newcommand{\termletter}{\tilde{h}}
\newcommandx{\tstat}[2][1=]{
\ifthenelse{\equal{#1}{}}
	{\tstatletter_{#2}}
	{\tau_{#2}^{#1}}
}
\newcommand{\termprime}{\Delta_N^\prime}
\newcommand{\termbis}{\Delta_N^{\prime \prime}}

\newcommand{\trmletter}{\boldsymbol{\Pi}}
\newcommand{\trm}[1]{\trmletter_{#1}}
\newcommand{\trmext}[1]{\bar{\trmletter}_{#1}}
\newcommand{\tstatletter}{\kernel{T}}
\newcommandx\tstatmod[2][1=]{
\ifthenelse{\equal{#1}{}}
	{\tstatletter^{\varepsilon}_{#2}}
	{\tau^{\varepsilon}_{#2}^{#1}}
}
\newcommand{\ud}[1]{\uksymbol_{#1}} 
\newcommand{\udlow}{\sigma_-}
\newcommand{\udmod}[1]{\uksymbol^\precpar_{#1}}
\newcommand{\udup}{\sigma_+}
\newcommand{\uk}[1]{\mathbf{L}_{#1}}
\newcommand{\ukest}[2]{\uksymbol_{#1} \langle #2 \rangle}
\newcommand{\ukestvar}[1]{\varsigma_{#1}^2}
\newcommand{\ukdist}[1]{\mathbf{R}_{#1}}
\newcommand{\ukmod}[1]{\mathbf{L}_{#1}^\precpar}
\newcommand{\uksymbol}{\ell}
\newcommand{\wgtfunc}[1]{w_{#1}}
\newcommand{\wgtfuncideal}[1]{\bar{w}_{#1}}
\newcommand{\wgtfuncmod}[1]{w^{\varepsilon}_{#1}}
\newcommand{\xarb}{x^\ast}
\newcommand{\Xfd}{\mathcal{X}}
\newcommand{\Xset}{\mathsf{X}}

\newcommand{\Yfd}{\mathcal{Y}}
\newcommand{\Yset}{\mathsf{Y}}

\newcommand{\Zfd}{\mathcal{Z}}
\newcommand{\zpart}[2]{\zeta_{#1}^{#2}} 
\newcommand{\zset}{\mathbb{Z}}
\newcommand{\Zset}{\mathsf{Z}}

\newcommand{\wgtfuncext}[1]{w_{#1}}
\newcommand{\ewghthat}[2]{\omega_{#1}^{#2}}
\newcommandx\tstathat[2][1=]{
\ifthenelse{\equal{#1}{}}
	{\tstatletter_{#2}}
	{\tau_{#2}^{#1}}
}
\newcommand{\M}{M}

\definecolor{violet}{cmyk}{0.79,0.88,0,0}
\definecolor{lavander}{cmyk}{0,0.48,0,0}
\definecolor{burntblue}{cmyk}{0.86,0.30,0.18,0}
\definecolor{burntorange}{cmyk}{0,0.52,1,0}
\definecolor{burntgreen}{cmyk}{0.62,0.44,0.47,0}
\definecolor{colorproof}{RGB}{80,93,113}
\definecolor{lightr}{RGB}{204,0,0}
\definecolor{palegreen}{cmyk}{0.86,0.30,0.96,0}

\newcounter{hypH}
\newenvironment{hypH}{\refstepcounter{hypH}\begin{itemize}
\item[({\bf H\arabic{hypH}})]}{\end{itemize}}
\newcommandx{\hypref}[2][1=]{
\ifthenelse{\equal{#1}{}}
{\hspace{-1mm}(\textbf{H\ref{#2}})\hspace{-1mm}}
{\hspace{-1mm}(\textbf{H\ref{#1}--\ref{#2}})\hspace{-1mm}}
}

\newcommand{\ourlongtitle}{A pseudo-marginal sequential Monte Carlo online smoothing algorithm}
\newcommand{\ourshorttitle}{A pseudo-marginal SMC online smoothing algorithm}

\endlocaldefs

\begin{document}

\begin{frontmatter}
\title{\ourlongtitle} 
\runtitle{\ourshorttitle} 

\begin{aug}
\author[A]{\fnms{Pierre} \snm{Gloaguen}\ead[label=e1]{pierre.gloaguem@agroparistech.fr}},
\author[B]{\fnms{Sylvain} \snm{Le Corff}\ead[label=e2]{sylvain\_lecorff@telecom-sudparis.eu}}
\and
\author[C]{\fnms{Jimmy} \snm{Olsson}\ead[label=e3]{jimmyol@kth.se}}
\address[A]{AgroParisTech, Paris, France.
\printead{e1}}

\address[B]{T\'el\'ecom SudParis, Paris, France.
\printead{e2}}

\address[C]{KTH Royal Institute of Technology, Stockholm, Sweden. 
\printead{e3}}

\end{aug}

\begin{abstract}
We consider online computation of expectations of additive state functionals under general path probability measures proportional to products of unnormalised transition densities. These transition densities are assumed to be intractable but possible to estimate, with or without bias. Using pseudo-marginalisation techniques we are able to extend the particle-based, rapid incremental smoother (PaRIS) algorithm proposed in [J.~Olsson and J.~Westerborn. Efficient particle-based online smoothing in general hidden Markov models: The {PaRIS} algorithm. {\em Bernoulli}, 23(3):1951--1996, 2017] to this setting. The resulting algorithm, which has a linear complexity in the number of particles and constant memory requirements, applies to a wide range of challenging path-space Monte Carlo problems, including smoothing in partially observed diffusion processes and models with intractable likelihood. The algorithm is furnished with several theoretical results, including a central limit theorem, establishing its convergence and numerical stability. Moreover, under strong mixing assumptions we establish a novel $\mathcal{O}(n \precpar)$ bound on the asymptotic bias of the algorithm, where $n$ is the path length and $\precpar$ controls the bias of the density estimators. 
\end{abstract}


\end{frontmatter}

\section{Introduction}
\label{sec:introduction}
Let $(\Xset_n, \Xfd_n)_{n \in \nset}$ be a sequence of general state spaces and let, for all $n \in \nset$, $\uk{n} : \Xset_n \times \Xfd_{n + 1} \to \rsetnn$ be bounded kernels in the sense that $\sup_{x \in \Xset_n} \uk{n}(x, \Xset_{n + 1}) < \infty$. We will assume a dominated model where each kernel $\uk{n}$ has a kernel density $\ud{n}$ with respect to some $\sigma$-finite reference measure $\mu_{n + 1}$ on $\Xfd_{n + 1}$. Finally, let $\chi$ be some bounded measure on $\Xfd_0$. In the following, we denote state-space product sets and $\sigma$-fields by $\Xset^n \eqdef \Xset_0 \times \cdots \times \Xset_n$ and $\Xfd^n \eqdef \Xfd_0 \tensprod \cdots \tensprod \Xfd_n$, respectively, and consider probability measures  
\begin{equation} \label{eq:def:post}
\post{0:n}(\rmd x_{0:n}) \propto \chi(\rmd x_0) \prod_{m = 0}^{n - 1} \uk{m}(x_m, \rmd x_{m + 1}), \quad n \in \nset, 
\end{equation}
on these product spaces.\footnote{We will always use the standard convention $\prod_{\varnothing} = 1$, implying that $\post{0} \propto \chi$.} Given a sequence $(\addf{n})_{n \in \nset}$ of measurable functions $\addf{n} : \Xset_n \times \Xset_{n + 1} \to \rset$, the aim of the present paper is the online approximation of expectations of \emph{additive functionals}   
\begin{equation} \label{eq:add:func}
    h_n : \Xset^n \ni x_{0:n} \mapsto \sum_{m = 0}^{n - 1} \addf{m}(x_m, x_{m + 1})
\end{equation}
under the distribution flow $(\post{0:n})_{n \in \nset}$ using \emph{sequential Monte Carlo} (SMC) methods. 

The generality of the model \eqref{eq:def:post} is striking. In the special case where each $\uk{n}$ can be decomposed as $\uk{n}(x_n, \rmd x_{n + 1}) = \md{n}(x_n) \, \hk_n(x_n, \rmd x_{n + 1})$ for some Markov kernel $\hk_n$ and some nonnegative potential function $\md{n}$, \eqref{eq:def:post} yields the \emph{Feynman-Kac path models} \cite{delmoral:2004}, which are applied in a large variety of scientific and engineering disciplines, including statistics, physics, biology, and signal processing. In a \emph{hidden Markov model} (HMM) (see, \eg,  \cite{Cappe:2005:IHM:1088883}), a Markov chain $(X_n)_{n \in \nset}$ with kernels $(\hk_n)_{n \in \nset}$ and initial distribution $\chi$ is only partially observed through a sequence $(Y_n)_{n \in \nset}$ of observations being conditionally independent given the Markov states. In that case, $g_n$ plays the role of the likelihood of the state $X_n$ given the corresponding observation $Y_n$, and $\post{0:n}$ describes the \emph{joint-smoothing distribution}, \ie, the joint posterior of the hidden states $X_0, \ldots, X_n$ given corresponding observations (see Example~\ref{ex:state-space:models} for details). We will adopt this terminology throughout the present paper and refer to the distributions defined in \eqref{eq:def:post} as `smoothing distributions'; the problem of computing expectations of functionals of type \eqref{eq:add:func} under these distributions will be referred to as `additive smoothing'. General state-space HMMs are prevalent in time-series and sequential-data analysis and are used extensively in, \eg, movement ecology \cite{michelot2016movehmm}, energy-consumption modeling \cite{candanedo2017methodology}, genomics \cite{yau2011bayesian}, target tracking \cite{sarkka2007rao}, enhancement and segmentation of speech and audio signals \cite{rabiner1989tutorial}; see also \cite{sarkka2013bayesian, zucchini2017hidden} and the numerous references therein. Operating on models of this sort, online additive smoothing is instrumental for, \emph{e.g.},    
\begin{enumerate}
    \item[--] \emph{path reconstruction}, \ie, the estimation of hidden states given observations. Especially in \emph{fixed-point smoothing}, where interest is in computing the expectations of $h(X_m)$ conditionally to $Y_0, \ldots, Y_n$ for some given $m$ and test function $h$ as $n$ tends to infinity, a problem that can be cast into our framework by letting, in \eqref{eq:add:func}, $\addf{m}(x_m, x_{m + 1}) = h(x_m)$ and $\addf{\ell} \equiv 0$ for all $\ell \neq m$. 
    \item[--] \emph{parameter inference}, where additive smoothing is a key ingredient in the computation of log-likelihood gradients (\emph{score functions}) via Fisher's identity or the intermediate quantity of the \emph{expectation-maximisation} (EM) \emph{algorithm}; see, \eg, \cite[Chapter~10]{Cappe:2005:IHM:1088883}. On-the-fly computation becomes especially important in online implementations via, \eg, the \emph{online EM} or \emph{recursive maximum likelihood} approaches \cite{cappe:2009,legland:mevel:1997}.  
 \end{enumerate}

As closed-form solutions to this smoothing problem can be obtained only for linear Gaussian models or models with finite state spaces $(\Xset_n)_{n \in \nset}$, loads of papers have been written over the years with the aim of developing SMC-based approximative solutions. Most of these works assume that each density $\ud{n}$ (or, in the HMM case, the transition density of $\hk_n$ and the likelihood $g_n$) is available in a closed form; however, this is not the case for a large number of interesting models, including most state-space HMMs governed by stochastic differential equations. Still, there are a few exceptions in the literature. In \cite{fearnhead2008particle} (see also \cite{fearnhead:papaspiliopoulos:roberts:stuart:2010}), the authors showed that asymptotically consistent online state estimation in partially observed diffusion processes can be achieved by means of a \emph{random-weight particle filter}, in which unavailable importance weights are replaced by unbiased estimates (produced using so-called \emph{generalized Poisson estimators} \cite{beskos:papaspiliopoulos:roberts:fearnhead:2006}). This approach is closely related to \emph{pseudo-marginal methods} \cite{andrieu:robert:2009}, since the unbiasedness allows the true, intractable target to be embedded into an extended distribution having the target as a marginal; as a consequence, the consistency of the algorithm follows straightforwardly from standard SMC convergence results. A similar pseudo-marginal SMC approach was developed in \cite{mcgree:drovandi:white:pettitt:2016} for random effects models with non-analytic likelihood. In \cite{olsson:strojby:2011}, this technology was cast into the framework of \emph{fixed-lag particle smoothing} of additive state functionals, where the well-known particle-path degeneracy of naive particle smoothers is avoided at the price of a lag-induced bias. Recently, \cite{yonekura:beskos:2020} designed a random-weight version of the forward-only particle smoother proposed in \cite{delmoral:doucet:singh:2010}, whose computational complexity is quadratic in the number $\N$ of particles, yielding a strongly consistent---though computationally demanding---algorithm. Moreover, \cite{gloaguen2018online} extended the random-weight particle filtering approach to the \emph{particle-based, rapid incremental smoother} (PaRIS), proposed in \cite{olsson:westerborn:2014b} as a means for additive smoothing in HMMs, yielding an algorithm with just linear complexity. The complexity of the latter algorithm is appealing; however, the schedule was not furnished with any theoretical results concerning the asymptotic properties and long-term stability of the estimator or the effect of the weight randomisation on the accuracy. In addition, the algorithm is restricted to partially observed diffusions and unbiased weight estimation, calling for strong assumptions on the unobserved process. 
        
In the present paper we further develop the approach in \cite{gloaguen2018online} and extend the PaRIS to online additive smoothing in general models in the form \eqref{eq:def:post} and the scenario where the transition densities $(\ud{n})_{n \in \nset}$ are intractable but can be estimated by means of simulation. These estimates may be unbiased or biased. In its original form, the PaRIS avoids particle-path degeneracy by alternating two sampling operations, one that propagates a sample of forward-filtering particles and another that resamples a set of backward-smoothing statistics, and the proposed method replaces the sampling distributions associated with these operations by suitable pseudo-marginals. This leads to an $\mathcal{O}(\N)$ algorithm that can be applied to a wide range of smoothing problems, including additive smoothing in partially observed diffusion processes and additive \emph{approximate Bayesian computation smoothing} \cite{martin:jasra:singh:whiteley:delmoral:maccoy:2014}. As illustrated by our examples, it covers the random-weight algorithms proposed in \cite{fearnhead2008particle} and \cite{gloaguen2018online} as special cases and provides, as another special case, an extension of the original PaRIS proposed in \cite{olsson:westerborn:2014b} to general path models \eqref{eq:def:post} and auxiliary particle filters. In addition, the proposed method is furnished with a rigorous theoretical analysis, the results of which can be summarised as follows. 

\begin{itemize}
\item We establish exponential concentration and asymptotic normality of the estimators produced by the algorithm. These results extend analogous results established in \cite{olsson:westerborn:2014b} for the original PaRIS (operating on fully dominated HMMs using the bootstrap particle filter), and the additional randomness of the pseudo-marginals can be shown to effect the asymptotic variance through an additional positive term. The fact that our smoothing algorithm, as explained above, involves two separate levels of pseudo-marginalisation makes this extension highly non-trivial. 
\item Under strong mixing assumptions we establish the long-term stochastic stability of our algorithm by showing that its asymptotic variance grows at most linearly in $n$. As explained in \cite[Section~1]{olsson:westerborn:2014b}, this is optimal for a path-space Monte Carlo estimator. As a by-product of this analysis, we obtain a time-uniform bound on the asymptotic variance of the random-weight particle filter. 
\item As mentioned above, we do not require the estimators of $(\uk{n})_{n \in \nset}$ to be unbiased. The bias is assumed to be regulated by some precision parameter $\precpar$ (see \hypref{assum:bias:bound}), and under additional strong mixing assumptions we establish an $\mathcal{O}(n \precpar)$ bound on the asymptotic bias of the final estimator. In addition, we obtain, as a by-product, an $\mathcal{O}(\precpar)$ bound for the random-weight particle filter. These results are the first of its kind.   
\end{itemize}

The paper is structured as follows. In Section~\ref{sec:preliminaries} we cast, under the temporary assumption that each $\ud{n}$ is tractable, the PaRIS into the general model \eqref{eq:def:post} and auxiliary particle filters and define carefully the two forward and backward sampling operations constituting the algorithm. Since this extension is of independent interest, we provide the details. In Section~\ref{sec:pseudo:marginal:PaRIS} we show how pseudo-marginal forward and backward sampling allow the temporary tractability assumption to be abandoned. Section~\ref{sec:theoretical:results} presents all theoretical results and although an extensive numerical study of the proposed scheme is beyond the scope of our paper, we present a minor numerical illustration of the $\mathcal{O}(n \precpar)$ bias bound in Section~\ref{sec:numerical:results}. Sections~\ref{sec:proof:lem:reversibility}--\ref{sec:variance:bounds} contain all proofs.

\section{Preliminaries}
\label{sec:preliminaries}
We first introduce some general notation. 
For any $(m, n) \in \zset$ such that $n \leq m$, we let $\intvect{m}{n}$ denote the set $\{m, \ldots, n\}$. For arbitrary elements $a_\ell$, $\ell \in \intvect{m}{n}$, we denote vectors by $a_{m:n} = (a_m, \ldots, a_n)$. The sets of measures, probability measures, and real-valued bounded measurable functions on some given state-space $(\mathsf{X}, \mathcal{X})$ are denoted by $\meas{\mathcal{X}}$, $\probmeas{\mathcal{X}}$, and $\bmf{\mathcal{X}}$, respectively. For any measure $\mu$ and measurable function $h$ we let $\mu h \eqdef \int h(x) \, \mu(\rmd x)$ denote the Lebesgue integral of $h$ with respect to $\mu$ whenever this is well defined. We will write $\mu^2 f = (\mu f)^2$ (whereas $\mu f^2 = \mu(f^2)$). For any finite set $S$, $\powerset{S}$ denotes the power set of $S$. The following kernel notation will be used repeatedly in the paper. Let $(\mathsf{X}, \mathcal{X})$ and $(\mathsf{Y}, \mathcal{Y})$ be general state spaces and $\kernel{K} : \mathsf{X} \times \mathcal{Y} \to \rsetnn$ some transition kernel. Then $\kernel{K}$ induces two operators, one acting on measurable functions and the other on measures. More precisely, for any $h \in \bmf{\mathcal{X} \tensprod \mathcal{Y}}$ and $\mu \in \meas{\mathcal{X}}$, let 
$$
\kernel{K} h : \mathsf{X} \ni x \mapsto \int h(x, y) \, \kernel{K}(x, \rmd y), \quad \mu \kernel{K} : \mathcal{Y} \ni A \mapsto \int \mu(\rmd x) \, \kernel{K}(x, A)
$$ 
Moreover, let $(\mathsf{Z}, \mathcal{Z})$ be a third state space and $\kernel{K}' : \mathsf{Y} \times \mathcal{Z} \to \rsetnn$ another kernel; then the product of $\kernel{K}$ and $\kernel{K}'$ is the kernel defined by 
$$
\kernel{K} \kernel{K}' : (x, A) \ni \mathsf{X} \times \mathcal{Z} \mapsto \int \kernel{K}(x, \rmd y) \, \kernel{K}'(y, A). 
$$    

\subsection{Model and aim}
\label{sec:model}

With notations as in Section~\ref{sec:introduction}, define, for each $n \in \nset$ and $m \in \intvect{0}{n}$, the kernel 
\begin{equation} \label{eq:def:uk:products}
    \uk{m, n}(x_{0:m}', \rmd x_{0:n + 1}) \eqdef \delta_{x_{0:m}'}(\rmd x_{0:m}) \prod_{\ell = m}^n \uk{\ell}(x_\ell, \rmd x_{\ell + 1}) 
\end{equation}
on $\Xset^n \times \Xfd^{n + 1}$. In addition, let $\uk{n, n - 1} = \operatorname{id}$. Note that $\uk{n, n}$ is different from $\uk{n}$ in the sense that the former is defined on $\Xset^n \times \Xfd^{n + 1}$ whereas the latter is defined on $\Xset_n \times \Xfd_{n + 1}$. We will always assume that for all $n \in \nset$, $\chi \uk{0, n - 1} \1_{\Xset^n} = \chi \uk{0} \cdots \uk{n - 1} \1_{\Xset_n} > 0$. Since each mapping $\uk{m, n} \1_{\Xset^n}$ depends only on the last coordinate $x_m$, a version of this mapping with domain $\Xset_m$ is well defined; we will denote the latter by the same symbol and write $\uk{m, n} \1_{\Xset^n}(x_m)$, $x_m \in \Xset_m$, when needed. Using the previous notations, the path measures \eqref{eq:def:post} can be expressed as 
\begin{equation}
\label{eq:FK:path}
    \post{0:n}(\rmd x_{0:n}) = \frac{\chi \uk{0, n - 1}(\rmd x_{0:n})}{\chi \uk{0, n - 1} \1_{\Xset^n}}, \quad n \in \nset. 
\end{equation} 
For each $n \in \nset$, let  
$ 
\post{n} : A \ni \Xfd_n \mapsto \post{0:n}(\Xset^{n - 1} \times A)
$
denote the marginal of $\post{0:n}$ with respect to the last component. Note that the path- and marginal-measure flows can be expressed recursively as  
\begin{equation} \label{eq:recursion:FK:path}
    \post{0:n + 1}(\rmd x_{0:n + 1}) = \frac{\post{0:n} \uk{n, n}(\rmd x_{0:n + 1})}{\post{0:n} \uk{n, n} \1_{\Xset^{n + 1}}} = \frac{\post{0:n} \uk{n, n}(\rmd x_{0:n + 1})}{\post{n} \uk{n} \1_{\Xset_{n + 1}}}, \quad n \in \nset, 
\end{equation}
and 
\begin{equation} \label{eq:recursion:FK:marg}
    \post{n + 1}(\rmd x_{n + 1}) = \frac{\post{n} \uk{n}(\rmd x_{n + 1})}{\post{n} \uk{n} \1_{\Xset_{n + 1}}}, \quad n \in \nset, 
\end{equation}
respectively. Given some sequence $(\addf{n})_{n \in \nset}$ of functions $\addf{n} : \Xset_n \times \Xset_{n + 1} \to \rset$, our aim is, as declared in Section~\ref{sec:introduction}, the online approximation of $(\post{0:n} h_n)_{n \in \nset}$, where $h_n$ defined by \eqref{eq:add:func}. 

\begin{remark} \label{we:vs:delmoral}
    Note that our framework is equivalent with the Feynman-Kac models considered in \cite[Section~1.3]{delmoral:2004}, where it is assumed that each kernel $\uk{n}$ can be decomposed into a Markov transition kernel $\kernel{M}_n$ on $\Xset_n \times \Xfd_{n + 1}$ and a potential function $\md{n} \in \bmf{\Xfd_n}$ according to $\uk{n}(x_n, \rmd x_{n + 1}) = \md{n}(x_n) \, \kernel{M}(x_n, \rmd x_{n + 1})$. Indeed, as soon as $\uk{n}$ is bounded, such a decomposition is always possible by letting $\kernel{M}(x_n, \rmd x_{n + 1}) \eqdef \uk{n}(x_n, \rmd x_{n + 1}) / \uk{n}(x_n, \Xset_{n + 1})$ and $\md{n}(x_n) \eqdef \uk{n}(x_n, \Xset_{n + 1})$. However, in our case this potential function is, on the contrary to what is assumed in \cite{delmoral:2004}, generally intractable, since $\uk{n}(x_n, \Xset_{n + 1})$ is typically unknown for the models that we will consider. Moreover, as noted in \cite[Section~1.3]{delmoral:2004}, the path model $(\post{0:n})_{n \in \nset}$ and the marginal model $(\post{n})_{n \in \nset}$ have the same mathematical structure in the sense that the path model can be formulated as a marginal model evolving on the spaces $(\Xset_n', \Xfd_n')_{n \in \nset}$, where $\Xset_n' \eqdef \Xset^n$ and $\Xfd_n' \eqdef \Xfd^n$, according to the initial distribution $\chi \eqdef \chi'$ and the transition kernels $(\uk{n}')_{n \in \nset}$, where $\uk{n}' \eqdef \uk{n, n}$. Still, the kernels $(\uk{n}')_{n \in \nset}$ involve transitions according Dirac measures, which makes the model formed by $\chi'$ and $(\uk{n}')_{n \in \nset}$ ill-suited for naive particle approximation; see Section~\ref{sec:SMC} for further discussion. 
\end{remark}

\begin{example}[state-space models]
\label{ex:state-space:models}
Let $(\Xset, \Xfd)$ and $(\Yset, \Yfd)$ be general state spaces and $(\hk_n)_{n \in \nset}$ and $(\mk_n)_{n \in \nset}$ sequences of Markov kernels on $\Xset \times \Xfd$ and $\Xset^2 \times \Yfd$, respectively. In addition, let $\chi$ be some probability measure on $\Xfd$. Consider a fully dominated model where all $\mk_n$ and $\hk_n$ have transition densities $\md{n}$ and $\hd{n}$ with respect to some reference measures $\nu$ and $\mu$ on $\Yfd$ and $\Xfd$, respectively. Let $\{X_0, (X_n, Y_n) : n \in \nsetpos\}$ be the canonical Markov chain induced by the initial distribution $\chi$ and the Markov 
kernel $\hk_n(x_n, \rmd x_{n + 1}) \, \mk_n(x_n, x_{n + 1}, \rmd y_{n + 1})$ (which has no dependence on the $y_n$ variable; the same dynamics hence applies to the first transition $X_0 \rightsquigarrow (X_1, Y_1)$) and denote by $\pP_\chi$ its law with corresponding expectation $\pE_\chi$. In this model, we assume that the \emph{state process} $(X_n)_{n \in \nset}$ is latent and only partially observed thought the \emph{observation process} $(Y_n)_{n \in \nsetpos}$. It can be shown that (i) the state process is itself a Markov chain with initial distribution $\chi$ and transition kernels $(\hk_n)_{n \in \nset}$ and (ii) conditionally to the state process, the observations are independent and such that the marginal distribution of $Y_n$ is given by $\mk_{n - 1}(X_{n - 1}, X_n, \cdot)$ for all $n$. In the case where $\mk_{n - 1}$ does not depend on $x_{n - 1}$, the model is a fully adapted general state-space HMM; see \cite[Section~2.2]{Cappe:2005:IHM:1088883}. In this setting, the joint-smoothing distribution at time $n \in \nset$ is, for a given record $y_{1:n} \in \Yset^n$ of observations, defined as the probability measure  
\begin{equation}
\label{eq:smooth}
    \post{0:n} \langle y_{1:n} \rangle (\rmd x_{0:n}) \eqdef \llh{n}^{-1}(y_{1:n}) \chi(\rmd x_0) \prod_{m = 0}^{n - 1} \hk_m(x_m, \rmd x_{m + 1}) \, \md{m}(x_m, x_{m + 1}, y_{m + 1}), 
\end{equation}
on $\Xfd^n$, where   
\begin{equation}
    \label{eq:likelihood}
    \llh{n}(y_{1:n}) \eqdef \idotsint \chi(\rmd x_0) \prod_{m = 0}^{n - 1} \hk_m(x_m, \rmd x_{m + 1}) \, \md{m}(x_m, x_{m + 1}; y_{m + 1})
\end{equation}
is the observed data likelihood. Along the lines of \cite[Proposition~3.1.4]{Cappe:2005:IHM:1088883} one may show that $\post{0:n} \langle Y_{1:n} \rangle$ is, $\pP_\chi$-a.s., the conditional distribution of $X_{0:n}$ given $Y_{1:n}$. The marginal  $\post{n} \langle y_{1:n} \rangle$ of the joint smoothing distribution with respect to its last component $x_n$ is referred to as the \emph{filtering distribution at time $n$}. Consequently, by defining kernel densities $\ud{n}(x_n, x_{n + 1}) \eqdef \md{n}(x_n, x_{n + 1}, y_{n + 1}) \hd{n}(x_n, x_{n + 1})$ for all $n \in \nset$ (while keeping dependence on observations implicit) and letting $\uk{n}$ be the induced transition kernels, the joint-smoothing distributions may be expressed in the form \eqref{eq:FK:path}.
\end{example}

The measures \eqref{eq:FK:path} are generally intractable in two ways. First, in many applications the transition densities $(\ud{n})_{n \in \nset}$ cannot be evaluated pointwise. Returning to Example~\ref{ex:state-space:models} and the context of smoothing in state-space models, this is typically the case when the dynamics of the latent process is governed by a stochastic differential equation (see Example~\ref{eq:durham:gallant} below). Second, even in the case where these transition densities are evaluable, the normalising constant in \eqref{eq:FK:path} is generally intractable. Thus, in order to solve the smoothing problem in full generality, one needs to be able to handle this double intractability, which is the goal of the algorithm that we will develop next. We will proceed in two steps. In the next section, Section~\ref{sec:PaRIS}, we will solve the additive smoothing problem under the temporary assumption that the densities $(\ud{n})_{n \in \nset}$ are tractable; the resulting extension of the PaRIS proposed in  \cite{olsson:westerborn:2014b} to general models in the form \eqref{eq:def:post} and auxiliary particle filters is of independent interest. Then, in Section~\ref{sec:pseudo:marginal:PaRIS}, we abandon the temporary assumption of tractability and assume that the user has only possibly biased proxies of these densities at hand. 

\subsection{The PaRIS}
\label{sec:PaRIS}


\subsubsection{Auxiliary particle filters}
\label{sec:SMC}

Assume for a moment that each transition density $\ud{n}$ is available in a closed form. Then standard SMC methods (see \cite{chopin:papaspiliopoulos:2020} for a recent introduction) can be used to approximate the distribution flows $(\post{0:n})_{n \in \nset}$ and $(\post{n})_{n \in \nset}$ using Monte Carlo samples generated recursively by means of sequential importance sampling and resampling operations. In order to set notations, let us recall the most general class of such algorithms, the so-called auxiliary particle filters \cite{pitt:shephard:1999}. In the light of Remark~\ref{we:vs:delmoral} it is enough to consider particle approximation of the marginals $(\post{n})_{n \in \nset}$. We proceed recursively and assume that we, at time $n \in \nset$, have at hand a sample $(\epart{n}{i}, \ewght{n}{i})_{i = 1}^\N$ of $\Xset_n$-valued \emph{particles} (the $\epart{n}{i}$) and associated nonnegative importance weights (the $\ewght{n}{i}$) such that the self-normalised estimator $\post[\N]{n} h \eqdef \sumwght{n}^{-1} \sum_{i = 1}^\N \ewght{n}{i} h(\epart{n}{i})$, with $\sumwght{n} \eqdef \sum_{i = 1}^\N \ewght{n}{i}$, approximates $\post{n} h$ for every $\post{n}$-integrable function $h$. Then plugging $\post[\N]{n}$ into the recursion \eqref{eq:recursion:FK:marg} yields the approximation $\sum_{i = 1}^\N \pi_n(i, \rmd x)$ of $\post{n + 1}$, where 
\begin{equation} \label{eq:def:partmixt}
\partmixt_n(i, \rmd x) \propto \ewght{n}{i} \uk{n}(\epart{n}{i}, \rmd x)
\end{equation}
is a mixture distribution on $\powerset{\intvect{1}{\N}} \tensprod \Xfd_{n + 1}$. In order to form new particles approximating $\post{n + 1}$, we may draw, using importance sampling, pairs $(\ind{n + 1}{i}, \epart{n + 1}{i})_{i = 1}^\N$ of indices and particles from $\partmixt_n$ and discard the former. For this purpose, we introduce some instrumental mixture distribution 
\begin{equation} \label{eq:cond:instrumental:mixture}
\rho_n(i, \rmd x) \propto \ewght{n}{i} \adjfuncforward{n}(\epart{n}{i}) \, \prop{n}(\epart{n}{i}, \rmd x)
\end{equation}
on the same space, where $\adjfuncforward{n}$ is a real-valued positive \emph{adjustment-weight function} on $\Xset_n$ and $\prop{n}$ is a proposal kernel on $\Xset_n \times \Xfd_{n + 1}$ such that $\uk{n}(x, \cdot) \ll \prop{n}(x, \cdot)$ for all $x \in \Xset_n$. We will always assume that $\prop{n}$ has a transition density $\propdens{n}$ with respect to $\mu_{n + 1}$. A draw $(\ind{n + 1}{i}, \epart{n + 1}{i})$ from $\rho_n$ is easily generated by first drawing $\ind{n + 1}{i}$ from the categorical distribution induced by the adjusted importance weights and then drawing $\epart{n + 1}{i}$ by moving randomly the selected ancestor $\epart{n}{{\ind{n + 1}{i}}}$ according to the proposal kernel. These steps are often referred to as \emph{selection} and \emph{mutation}, respectively. 
Finally, each draw $\epart{n + 1}{i}$ is assigned the updated importance weight proportional to $\rmd \partmixt_n / \rmd \rho_n (\ind{n + 1}{i}, \epart{n + 1}{i})$, and the estimator $\post[\N]{n + 1} h = \sumwght{n + 1}^{-1} \sum_{i = 1}^\N \ewght{n + 1}{i} h(\epart{n + 1}{i})$ approximates $\post{n + 1} h$ for every $\post{n + 1}$-integrable $h$. The full update, which we will refer to as \emph{forward sampling} and express in a short form as,  
\begin{equation} \label{eq:forward:sampling}
    (\epart{n + 1}{i}, \ewght{n + 1}{i})_{i = 1}^\N \sim \mathsf{FS}((\epart{n}{i}, \ewght{n}{i})_{i = 1}^\N),   
\end{equation}
is summarised in Algorithm~\ref{alg:ideal:SMC}.\footnote{Mathematically, the forward sampling operation defines a Markov transition kernel, which motivates the use of the symbol $\sim$ in \eqref{eq:forward:sampling}.} 

\begin{algorithm}[h] 
    \KwData{$(\epart{n}{i}, \ewght{n}{i})_{i = 1}^\N$}
    \KwResult{$(\epart{n + 1}{i}, \ewght{n + 1}{i})_{i = 1}^\N$}
    \For {$i = 1 \to \N$}{
        draw $\ind{n + 1}{i} \sim \cat(\{ \adjfuncforward{n}(\epart{n}{\ell}) \ewght{n}{\ell} \}_{\ell = 1}^\N)$\;
        draw $\epart{n + 1}{i} \sim \prop{n}(\epart{n}{{\ind{n + 1}{i}}}, \cdot)$\;
        set $\ewght{n + 1}{i} \gets \frac{\ud{n}(\epart{n}{{\ind{n + 1}{i}}}, \epart{n + 1}{i})}{\adjfuncforward{n}(\epart{n}{{\ind{n + 1}{i}}}) \propdens{n}(\epart{n}{{\ind{n + 1}{i}}}, \epart{n + 1}{i})}$\;
}
\caption{Forward sampling, \textsf{FS}} \label{alg:ideal:SMC}
\end{algorithm}

With this terminology, the auxiliary particle filter consists of iterated forward sampling operations, and we will assume that the process is initialised by sampling independent particles $(\epart{0}{i})_{i = 1}^\N$ from some proposal $\init$ on $(\Xset_0, \Xfd_0)$ such that $\chi \ll \init$ and letting $\ewght{0}{i} \eqdef \rmd \chi / \rmd \init(\epart{0}{i})$ for all $i$. 

Note that we may, in the light of Remark~\ref{we:vs:delmoral},  obtain particle approximations $(\epart{0:n}{i}, \ewght{n}{i})_{i = 1}^\N$, $n \in \nset$, of the smoothing distribution flow $(\post{0:n})_{n \in \nset}$ by applying the previous sampling scheme to the model formed by $\chi'$ and $(\uk{n}')_{n \in \nset}$. From an algorithmic point of view, it is easy to see that the only change needed is to insert, just after Line 3, the command $\epart{0:n + 1}{i} \gets (\epart{0:n}{\ind{n + 1}{i}}, \epart{n + 1}{i})$ storing the particle paths (in particular, the weight-updating step on Line~4 remains the same). Still, it is well known that repeated selection operations lead to coalescing paths $(\epart{0:n}{i})_{i = 1}^\N$, and as a consequence the variance of this naive smoothing estimator increases rapidly with $n$; indeed, in the case of additive state functionals, the growth in variance is typically quadratic in $n$ (see \cite{olsson:westerborn:2014b} for a discussion), which is unreasonable from a computational point of view. We will thus rely on more advanced, stochastically stable smoothing technology avoiding the particle-path degeneracy problem by taking advantage of the time-uniform convergence of the marginal samples $(\epart{n}{i})_{i = 1}^\N$. This will be discussed in the next section. 

Finally, we note that the re-weighting operation on Line~4 in Algorithm~\ref{alg:ideal:SMC} requires the transition density $\ud{n}$ to be tractable, which is not the case in general. We will return to the general case in Section~\ref{sec:pseudo:marginal:PaRIS}. 


\subsubsection{Backward sampling} 
\label{sec:BS}

The following quantities will play a key role in the following. For each $m \in \nset$, define the  \emph{backward Markov kernel} 
\begin{equation} \label{eq:def:backward:kernel}
    \bkw{m}(x_{m + 1}, \rmd x_m) \eqdef \frac{\post{m}(\rmd x_m) \, \ud{m}(x_m, x_{m + 1})}{\post{m}[\ud{m}(\cdot, x_{m + 1})]}
\end{equation}
on $\Xset_{m + 1} \times \Xfd_m$. In addition, for each $n \in \nsetpos$, let the Markov kernel   
\begin{equation} \label{eq:def:tstat}
\tstat{n}(x_n, \rmd x_{0:n - 1}) \eqdef \prod_{m = 0}^{n - 1} \bkw{m}(x_{m + 1}, \rmd x_m)
\end{equation}
on $\Xset_n \times \Xfd^{n - 1}$ denote the joint law of the backward Markov chain induced by the kernels \eqref{eq:def:backward:kernel} when initialised at $x_n \in \Xset_n$. An important class of sequential Monte Carlo joint-smoothing methods \cite{doucet2000sequential,godsill:doucet:west:2004} is based on the following result. 
\begin{lemma} 
\label{lem:reversibility}
\ 

\begin{itemize} 
\item[(i)]
For all $n \in \nset$ and $h \in \bmf{\Xfd_n \tensprod \Xfd_{n + 1}}$, 
\begin{multline} \label{eq:reversibility}
\iint \post{n}(\rmd x_n) \, \uk{n}(x_n, \rmd x_{n + 1}) \, h(x_n, x_{n + 1}) \\
= \iint \post{n} \uk{n}(\rmd x_{n + 1}) \, \bkw{n}(x_{n + 1}, \rmd x_n) \, h(x_n, x_{n + 1}). 
\end{multline}
\item[(ii)]
For all $n \in \nsetpos$ and $h \in \bmf{\Xfd^n}$, 
$
\post{0:n} h = \post{n} \tstat{n} h. 
$
\end{itemize}
\end{lemma}

In the case of state-space models, the identity (ii) above is a well-known result typically referred to as the \emph{backward decomposition} of the joint-smoothing distribution; still, as far as known to the authors, it has never been established in the general setting considered in the present paper, and a proof of Lemma~\ref{lem:reversibility} is hence given in Section~\ref{sec:proof:lem:reversibility} for completeness. Importantly, as noted in \cite{cappe:2009}, the functions $(\tstat{n} h_n)_{n \in \nsetpos}$ can be expressed recursively through  
\begin{equation} \label{eq:forward:smoothing}
\tstat{n + 1} h_{n + 1}(x_{n + 1}) = \int (\tstat{n} h_n(x_n) + \addf{n}(x_n, x_{n + 1})) \, \bkw{n}(x_{n + 1}, \rmd x_n), \quad n \in \nset, 
\end{equation}
with, by convention, $\tstat{0} h_0 \equiv 0$. Here the backward kernel $\bkw{n}$ depends on the marginal $\post{n}$; thus, the recursion is driven by the marginal flow $(\post{n})_{n \in \nset}$, which may again be expressed recursively through \eqref{eq:recursion:FK:marg}. However, as these marginals are, as already mentioned, generally intractable, exact computations need typically to be replaced by approximations. The authors of \cite{delmoral:doucet:singh:2010} propose to approximate the values of each statistic $\tstat{n} h_n$ at a random, discrete support formed by particles. More precisely, assume again that the transition density $\ud{n}$ is tractable and, by induction, that we at time step $n$ have at hand a given particle sample $(\epart{n}{i}, \ewght{n}{i})_{i = 1}^\N$ and a set of statistics $(\tstat[i]{n})_{i = 1}^\N$ such that $\tstat[i]{n}$ is an approximation of $\tstat{n} h_n(\epart{n}{i})$. Then, in order to propagate the statistics $(\tstat[i]{n})_{i = 1}^\N$ forward, one updates, in a first substep, the particle sample $(\epart{n}{i}, \ewght{n}{i})_{i = 1}^\N$ recursively by forward sampling (Algorithm~\ref{alg:ideal:SMC}). After forward sampling, one replaces, in the definition \eqref{eq:def:backward:kernel} of $\bkw{n}$, $\post{n}$ by the corresponding particle approximation, yielding the updates 
\begin{equation} \label{eq:tstat:FFBSm:update}
\tstat[i]{n + 1} = \sum_{j = 1}^\N \trm{n}(i, j) (\tstat[j]{n} + \addf{n}(\epart{n}{j}, \epart{n + 1}{i})), \quad i \in \intvect{1}{\N}, 
\end{equation} 
where we have defined the transition kernel  
\begin{equation} \label{eq:def:trm}
\trm{n}(i, j) \eqdef \frac{\ewght{n}{j} \ud{n}(\epart{n}{j}, \epart{n + 1}{i})}{\sum_{j' = 1}^\N \ewght{n}{j'} \ud{n}(\epart{n}{j'}, \epart{n + 1}{i})}
\end{equation}
on $\intvect{1}{\N} \times \powerset{\intvect{1}{\N}}$. Since computing each $\tstat[i]{n}$ according to \eqref{eq:tstat:FFBSm:update} has a linear computational complexity in the number $\N$ of particles, the overall complexity this approach is \emph{quadratic} in $\N$. In order to deal with this significant computational burden, the authors of \cite{olsson:westerborn:2014b} suggest replacing summation by additional Monte Carlo simulation. More precisely, by sampling, for each $i$, $\K \in \nsetpos$ independent indices $(J^{(i, j)})_{j = 1}^\M$ from $\trm{n}(i, \cdot)$ and replacing \eqref{eq:tstat:FFBSm:update} by  
\begin{equation} \label{eq:PaRIS:update}
\tstat[i]{n + 1} = \frac{1}{\K} \sum_{j = 1}^{\K} \left( \tstat[\bi{n + 1}{i}{j}]{n} + \addf{n}(\epart{n}{\bi{n + 1}{i}{j}}, \epart{n + 1}{i}) \right), \quad i \in \intvect{1}{\N}, 
\end{equation}
the computational complexity can, as we shall soon see, be reduced significantly. At each iteration, the self-normalised estimator $\sumwght{n}^{-1} \sum_{i = 1}^\N \ewght{n}{i} \tstat[i]{n}$ serves as an estimator of the quantity $\post{n} \tstat{n} h_n = \post{0:n} h_n$ of interest. This second operation, which we will refer to as \emph{backward sampling}, 
$$
    (\tstat[i]{n + 1})_{i = 1}^\N \sim \mathsf{BS}((\epart{n}{i}, \tstat[i]{n}, \ewght{n}{i})_{i = 1}^\N, (\epart{n + 1}{i})_{i = 1}^\N),  
$$
is summarised in Algorithm~\ref{alg:ideal:BS}. 

\begin{algorithm}[h] 
    \KwData{$(\epart{n}{i}, \tstat[i]{n}, \ewght{n}{i})_{i = 1}^\N$, $(\epart{n + 1}{i})_{i = 1}^\N$}
    \KwResult{$(\tstat[i]{n + 1})_{i = 1}^\N$}
    \For{$i = 1 \to \N$}{
    \For{$j = 1 \to \K$}{
    draw $\bi{n + 1}{i}{j} \sim \trm{n}(i, \cdot)$\;
    }
    set $\tstat[i]{n + 1} \gets \frac{1}{\K} \sum_{j = 1}^{\K} \left( \tstat[\bi{n + 1}{i}{j}]{n} + \addf{n}(\epart{n}{\bi{n + 1}{i}{j}}, \epart{n + 1}{i}) \right)$\;
}
\caption{Backward sampling, \textsf{BS}} \label{alg:ideal:BS}
\end{algorithm}

Let us examine closer the sampling step on Line~3 in Algorithm~\ref{alg:ideal:BS}. 
In order to keep the algorithmic complexity at a reasonable level, the computation of the normalising constant of $\trm{n}(i, \cdot)$, which consists of $\N$ terms, should be avoided; otherwise, the overall complexity remains quadratic in $\N$. This is possible using, \eg,   
\begin{itemize}
\item[--] \emph{rejection sampling}. This approach relies on the mild assumption that there exists some measurable function $c$ on $\Xset_{n + 1}$ such that $\ud{n}(x_n, x_{n + 1}) \leq c(x_{n + 1})$ for all $x_{n: n + 1} \in \Xset_n \times \Xset_{n + 1}$. Then, following \cite{douc:garivier:moulines:olsson:2010}, $\trm{n}(i, \cdot)$ can be sampled from by generating a candidate $J^\ast$ from $\cat(\{ \ewght{n}{\ell}\}_{\ell = 1}^\N)$ and accepting the same with probability 
\begin{equation} \label{eq:std:acc:prob:backward:sampling}
\accprob \eqdef \frac{\ud{n}(\epart{n}{J^\ast}, \epart{n + 1}{i})}{c(\epart{n + 1}{i})}. 
\end{equation}
The procedure is repeated until acceptance, conditionally to which $J^\ast$ is distributed according to $\trm{n}(i, \cdot)$. Since the $\cat(\{ \ewght{n}{\ell}\}_{\ell = 1}^\N)$ distribution is independent of $i$, this circumvents the need to compute a normalising sum for every $i$. The approach may reduce significantly the computational complexity; indeed, as shown in \cite[Proposition~2]{douc:garivier:moulines:olsson:2010}, the expected overall complexity of the algorithm is \emph{linear} in $\N$ under certain assumptions.   
\item[--] \emph{MCMC methods}. Another possibility is to generate the variables $(\bi{n + 1}{i}{j})_{j = 1}^{\K}$ using the Metropolis-Hastings algorithm. For this purpose, let $\rho$ be some proposal transition density on $\intvect{1}{\N}^2$. Then proceeding recursively, given $\bi{n + 1}{i}{j} = J$, a candidate $J^\ast$ for $\bi{n + 1}{i}{j + 1}$ is sampled from the density $\rho(J, \cdot)$ and accepted with probability 
$$
\accprob[MH] \eqdef 1 \wedge\frac{\ewght{n}{J^\ast} \ud{n}(\epart{n}{J^\ast}, \epart{n + 1}{i}) \rho(J^\ast, J)}{\ewght{n}{J} \ud{n}(\epart{n}{J}, \epart{n + 1}{i}) \rho(J, J^\ast)}. 
$$ 
If rejection, then $\bi{n + 1}{i}{j + 1} = J$. It is close at hand is to let $\rho$ take the form of an independent proposal given by the $\cat(\{ \ewght{n}{i} \}_{i = 1}^\N)$ distribution; in that case $\accprob[MH]$ simplifies to  
\begin{equation} \label{eq:std:MH:prob:backward:sampling}
\accprob[MH] = 1 \wedge \frac{\ud{n}(\epart{n}{J^\ast}, \epart{n + 1}{i})}{\ud{n}(\epart{n}{J}, \epart{n + 1}{i})}. 
\end{equation}
With this approach, the variables $(\bi{n + 1}{i}{j})_{j = 1}^{\K}$ are conditionally dependent; this can be counteracted by including only an $m$-skeleton of this sequence in the update \eqref{eq:PaRIS:update}. An important advantage of this approach over rejection sampling is that it does not require $\ud{n}$ to be dominated.  
\end{itemize}

Finally, combining the forward and backward sampling operations in accordance with Algorithm~\ref{alg:ideal:PaRIS} yields a generalisation of the PaRIS proposed in \cite{olsson:westerborn:2014b} to a general framework comprising Feynman-Kac models and auxiliary particle filters. 

\begin{algorithm}[h] 
    \KwData{$(\epart{n}{i}, \tstat[i]{n}, \ewght{n}{i})_{i = 1}^\N$}
    \KwResult{$(\epart{n + 1}{i}, \tstat[i]{n + 1}, \ewght{n + 1}{i})_{i = 1}^\N$}
    run $(\epart{n + 1}{i}, \ewght{n + 1}{i})_{i = 1}^\N \sim \mathsf{FS}((\epart{n}{i}, \ewght{n}{i})_{i = 1}^\N)$\;
    run $(\tstat[i]{n + 1})_{i = 1}^\N \sim \mathsf{BS}((\epart{n}{i}, \tstat[i]{n}, \ewght{n}{i})_{i = 1}^\N, (\epart{n + 1}{i})_{i = 1}^\N)$\;
    \medskip
\caption{Full PaRIS update.} \label{alg:ideal:PaRIS}
\end{algorithm}

Algorithm~\ref{alg:ideal:PaRIS} is initialised by drawing $(\epart{0}{i})_{i = 1}^\N \sim \init^{\varotimes \N}$ and letting $\ewght{0}{i} = \rmd \chi / \rmd \init(\epart{0}{i})$ and $\tstat[i]{n} = 0$.  

In this scheme, the sample size $\K$ of the backward sampling operation is an algorithmic parameter that has to be set by the user. As shown in Section~\ref{sec:theoretical:results}, the produced estimators are, for all $n \in \nset$, consistent and asymptotically normal for any fixed $\K$ larger than or equal to one. In addition, for any fixed $\K \geq 2$ the algorithm is stochastically stable with an $\mathcal{O}(n)$ variance, which is optimal; see \cite[Section~1]{olsson:westerborn:2014b} for a discussion. These results form a nontrivial extension of similar results obtained by \cite{olsson:westerborn:2014b} in the simpler setting of state-space models and bootstrap particle filters. 

Finally, we remind the reader that we have here considered the idealised situation where the unnormalised transition densities $(\ud{n})_{n \in \nset}$ can be evaluated pointwise, which will generally not be the case for the applications we will consider. Thus, in the next section we will approach the more general case where these transition densities are intractable but may be estimated, and we will show how consistent, asymptotically normal, and stochastically stable estimators can be produced also in such a scenario by pseudo-marginalising separately the forward and backward sampling operations.

\section{Pseudo-marginal PaRIS algorithms}
\label{sec:pseudo:marginal:PaRIS}
\subsection{Pseudo marginalisation in Monte Carlo methods}
\label{sec:pseudo:marginalisation}


Pseudo-marginalisation was originally proposed in \cite{beaumont:2003} in the framework of MCMC methods, and in \cite{andrieu:robert:2009} the method was developed further and provided with a solid theoretical basis. In the following we recapitulate briefly the main idea behind this approach. Consider the problem of sampling from some target distribution $\pi$ defined on some measurable space $(\Xset, \Xfd)$ and having a density with respect to some reference measure $\mu$. This density is assumed to be proportional to some \emph{intractable} nonnegative measurable function $\ell$ on $\Xset$, \ie, $\pi(\rmd x) = \lambda(\rmd x) / \lambda \1_\Xset$, where $\lambda(\rmd x) \eqdef \ell(x) \, \mu(\rmd x)$ is finite. While the target density is intractable we assume that there exist some additional state space $(\Zset, \Zfd)$, a Markov kernel $\mathbf{R}$ on $\Xset \times \Zfd$, and some nonnegative measurable function $\Xset \times \Zset \ni (x, z) \mapsto \ell \langle z \rangle (x)$ known up to a constant of proportionality and such that for all $x \in \Xset$, 
\begin{equation} \label{eq:pseudo:marginalisation}
\int \ell \langle z \rangle(x) \, \mathbf{R}(x, \rmd z) = \ell(x). 
\end{equation}
Thus, a pointwise estimate of $\ell(x)$ can be obtained by generating $\zeta$ from $\mathbf{R}(x, \rmd z)$ and computing the statistic $\ell \langle \zeta \rangle(x)$, the \emph{pseudo marginal}. In Monte Carlo methods, the measure to replace, when necessary, the true marginal $\ell$ by its pseudo marginal is referred to as \emph{pseudo marginalisation}. Interestingly, even though pseudo marginalisation is based on the plug-in principle, it preserves typically the consistency of an algorithm. In order to see the this, let $\bar{\mathsf{X}} \eqdef \mathsf{X} \times \mathsf{Z}$ and $\bar{\mathcal{X}} \eqdef \mathcal{X} \tensprod \mathcal{Z}$; then one may define an extended target distribution $\bar{\pi}(\rmd \bar{x}) \eqdef \bar{\lambda}(\rmd \bar{x}) / \bar{\lambda} \1_{\bar{\mathsf{X}}} = \bar{\lambda}(\rmd \bar{x}) / \lambda \1_{\mathsf{X}}$ on $(\bar{\mathsf{X}}, \bar{\mathcal{X}})$, where 
\begin{equation} \label{eq:extended:target}
\bar{\lambda}(\rmd \bar{x}) \eqdef \ell \langle z \rangle(x) \, \kernel{R}(x, \rmd z) \, \mu(\rmd x)
\end{equation}
(with $\bar{x} = (x, z)$). By \eqref{eq:pseudo:marginalisation}, $\pi$ is the marginal of $\bar{\pi}$ with respect to the $x$ component. This means that we may produce a random sample $(\xi^i)_{i = 1}^N$ in $\Xset$ targeting $\pi$ by generating a sample $(\xi^i, \zeta^i)_{i = 1}^N$ targeting $\bar{\pi}$ and simply discarding the $\Zset$-valued variables $(\zeta^i)_{i = 1}^N$. Let $\rho$ be a Markov transition density on $\Xset$ with respect to the reference measure $\mu$. Then following \cite{andrieu:robert:2009}, a Markov chain $(\xi_m, \zeta_m)_{m \in \nset}$ targeting $\bar{\pi}$ can be produced on the basis of the Metropolis-Hastings algorithm as follows. Given a state $(\xi_m, \zeta_m)$, a candidate $(\xi^\ast, \zeta^\ast)$ for the next state is generated by drawing $\xi^\ast \sim \rho(x) \, \mu(\rmd x)$ and $\zeta^\ast \sim \kernel{R}(\xi^\ast, \rmd z)$ and accepting the same with probability 
$$
\alpha \eqdef 1 \wedge \frac{\ell \langle \zeta^\ast \rangle(\xi^\ast) \rho(\xi^\ast, \xi_m)}{\ell \langle \zeta_m \rangle (\xi_m) \rho(\xi_m, \xi^\ast)},  
$$
which is tractable. Note that $\alpha$ is indeed a pseudo-marginal version of the exact acceptance probability $1 \wedge \ell (\xi^\ast) \rho(\xi^\ast, \xi_m) / \ell (\xi_m) \rho(\xi_m, \xi^\ast)$ obtained if $\ell$ would be known. Note that the auxiliary variable $\zeta$ enters $\alpha$ only through the estimates $\ell \langle \zeta^\ast \rangle (\xi^\ast)$ and $\ell \langle \zeta_m \rangle (\xi_m)$, and since the latter has already been computed at the previous iteration there is no need of recomputing this quantity. In the \emph{Monte-Carlo-within-Metropolis algorithm} (see again \cite{andrieu:robert:2009}) $\ell \langle z \rangle(x)$ is a pointwise importance sampling estimate of $\ell(x)$ based on a Monte Carlo sample $\zeta$ generated from $\mathbf{R}(x, \cdot)$. Alternatively, the extended distribution $\bar{\pi}$ can be sampled using rejection sampling or importance sampling, leading to consistent pseudo-marginal formulations of these algorithms as well. 

In the present paper we will generalise the pseudo-marginal approach towards biased estimation by allowing the function 
\begin{equation} \label{eq:biased:pseudo:marginalisation}
    \ell^\precpar(x) \eqdef \int \ell \langle z \rangle(x) \, \mathbf{R}(x, \rmd z)
\end{equation}
on $\Xset$ to be different from $\ell$. Here $\varepsilon \geq 0$ is some accuracy parameter describing the distance between $\ell^\precpar$ and $\ell$. Such biased estimates appear naturally when the law $\pi$ is, \eg, governed by a diffusion process and the density $\ell$ is approximated on the basis of different discretisation schemes; see the next section. In that case, $\precpar$ plays the role of the discretisation step size. In \eqref{eq:biased:pseudo:marginalisation}, also the estimator $\ell \langle z \rangle(x)$ and the kernel $\mathbf{R}(x, \rmd z)$ may depend on $\precpar$, even though this is suppressed in the notation. By introducing the possibly unnormalised measure $\lambda^\varepsilon(\rmd x) \eqdef \ell^\precpar(\rmd x) \, \mu(\rmd x)$ on $\Xfd$ we may define the \emph{skew} target probability measure
$$
    \pi^\varepsilon(\rmd x) \eqdef \frac{\lambda^\precpar(\rmd x)}{\lambda^\precpar \1_\Xset}
$$
on $\Xfd$. In the biased case, generating a sample $(\xi^i, \zeta^i, \omega^i)_{i = 1}^N$ targeting \eqref{eq:extended:target} will, by \eqref{eq:biased:pseudo:marginalisation}, provide a sample $(\xi^i, \omega^i)_{i = 1}^N$ targeting $\pi^\varepsilon$ as a by-product. Thus, of outmost importance is to obtain control over the bias between $\pi$ and $\pi^\precpar$, which is possible under the assumption that there exists a constant $c \geq 0$ such that for all $h \in \bmf{\Xfd}$ and $\varepsilon$,   
\begin{equation} \label{eq:lipschitz:simple:case}
    \left| \lambda^\varepsilon h - \lambda h \right| \leq c \varepsilon \| h \|_\infty. 
\end{equation}
For instance, in the diffusion process case mentioned above, a condition of type \eqref{eq:lipschitz:simple:case} holds, as we will see in Section~\ref{sec:theoretical:results}, typically true in the case where the density is approximated using the \emph{Durham-Gallant estimator} \cite{durham:gallant:2002}. Using that for all $h \in \bmf{\Xfd}$,  
$$
    \pi^\precpar h - \pi h = \pi^\precpar h \left( 1 - \frac{\lambda^\varepsilon \1_\Xset}{\lambda \1_\Xset} \right) + \frac{\lambda^\precpar h - \lambda h}{\lambda \1_\Xset}, 
$$
we straightforwardly obtain the bound  
$$
    \left| \pi^\precpar h - \pi h \right| \leq \precpar \frac{2 c}{\lambda \1_\Xset} \| h \|_\infty
$$
on the systematic error induced by the skew model. Note that the unbiased case \eqref{eq:pseudo:marginalisation} corresponds to letting $\precpar = 0$ in assumption~\eqref{eq:lipschitz:simple:case}. In the next section we will present a solution to the main problem addressed in Section~\ref{sec:model} exploring a pseudo-marginalised version of the PaRIS discussed in Section~\ref{sec:PaRIS}. 

\subsection{Pseudo-marginal PaRIS}

The algorithm that we will propose relies on the following assumption. 
 
\begin{hypH}
\label{assum:biased:estimate}
Let $(\Zset_n, \Zfd_n)_{n \in \nsetpos}$ be a sequence of general state spaces. For each $n \in \nset$ there exist a Markov kernel $\ukdist{n}$ on $\Xset_n \times \Xset_{n + 1} \times \Zfd_{n + 1}$ and a positive measurable function $\ukest{n}{z}(x_n, x_{n + 1})$ on $\Xset_n \times \Xset_{n + 1} \times \Zset_{n + 1}$ such that for every $x_{n:n + 1} \in \Xset_n \times \Xset_{n + 1}$, drawing $\zeta \sim \ukdist{n}(x_{n:n + 1}, \rmd z)$ and computing $\ukest{n}{\zeta}(x_n, x_{n + 1})$ yields an estimate of $\ud{n}(x_n, x_{n + 1})$.  
\end{hypH}

Under \hypref{assum:biased:estimate}, we denote, for every $n \in \nset$ and $x_{n:n + 1} \in \Xset_n \times \Xset_{n + 1}$, by   
\begin{equation} \label{eq:def:udmod}
\udmod{n}(x_n, x_{n + 1}) \eqdef \int \ukdist{n}(x_{n:n + 1}, \rmd z) \, \ukest{n}{z}(x_n, x_{n + 1})
\end{equation}
the expectation of the statistic $\ukest{n}{\zeta}(x_n, x_{n + 1})$. Here $\precpar$ is an accuracy parameter belonging to some parameter space $\precparsp \subset \rset$ and controlling the bias of the estimated model with respect to the true model; this will be discussed in depth in Section~\ref{sec:lipschitz:continuity}. For each $n \in \nset$ we define the unnormalised kernel 
\begin{equation} \label{eq:def:ukmod}
    \ukmod{n}(x_n, \rmd x_{n + 1}) = \udmod{n}(x_n, x_{n + 1}) \, \mu_{n + 1}(\rmd x_{n + 1}) 
\end{equation}
on $\Xset_n \times \Xfd_{n + 1}$. 
 

\subsubsection{Pseudo-marginal forward sampling}
\label{eq:sec:pm:forward:sampling}

In the case where each $\ud{n}$ is intractable, so is the mixture distribution $\partmixt_n$ defined in \eqref{eq:def:partmixt}. Under \hypref{assum:biased:estimate}, we may, as in Section~\ref{sec:pseudo:marginalisation}, aim at consistent pseudo-marginalisation of the forward-sampling operation by applying self-normalised importance sampling to the extended mixture  
$$
\bar{\pi}_n(i, \rmd x, \rmd z)
\propto \ewght{n}{i} \ukest{n}{z}(\epart{n}{i}, x) \, \mu_{n + 1}(\rmd x) \, \ukdist{n}(\epart{n}{i}, x, \rmd z) 
$$
on $\bar{\Xfd}_{n + 1} \eqdef \powerset{\intvect{1}{\N}} \tensprod \Xfd_{n + 1} \tensprod \Zfd_{n + 1}$ using the instrumental distribution  
$$
\bar{\rho}_n(i, \rmd x, \rmd z) \propto \ewght{n}{i} \adjfuncforward{n}(\epart{n}{i}) \, \prop{n}(\epart{n}{i}, \rmd x) \, \ukdist{n}(\epart{n}{i}, x, \rmd z)
$$
on the same space. Note that the marginal of $\bar{\pi}_n$ with respect to $(i, x)$ is the distribution proportional to $\ewght{n}{i} \ukmod{n}(\epart{n}{i}, \rmd x)$, whose distance to the target $\pi_n$ of interest is controlled by the precision parameter $\precpar$. Here the adjustment multiplier $\adjfuncforward{n}$ and the proposal kernel $\prop{n}$ of the instrumental distribution are as in Section~\ref{sec:SMC}. Each draw from $\bar{\rho}_n$ is assigned an importance weight given by the (tractable) Radon--Nikodym derivative of $\bar{\pi}_n$ to $\bar{\rho}_n$. It is easily seen that this sampling operation, which is detailed in Algorithm~\ref{alg:pm:SMC}, corresponds to replacing the intractable transition density $\ud{n}$ on Line~4 in Algorithm~\ref{alg:ideal:SMC} by an estimate provided by \hypref{assum:biased:estimate}; we will hence refer to Algorithm~\ref{alg:pm:SMC} as \emph{pseudo-marginal forward sampling} and express it concisely as 
$$
    (\epart{n + 1}{i}, \ewght{n + 1}{i})_{i = 1}^\N \sim \mathsf{pmFS}((\epart{n}{i}, \ewght{n}{i})_{i = 1}^\N). 
$$

\begin{algorithm}[h] 
    \KwData{$(\epart{n}{i}, \ewght{n}{i})_{i = 1}^\N$}
    \KwResult{$(\epart{n + 1}{i}, \ewght{n + 1}{i})_{i = 1}^\N$}
    \For {$i = 1 \to \N$}{
        draw $\ind{n + 1}{i} \sim \cat(\{ \adjfuncforward{n}(\epart{n}{\ell}) \ewght{n}{\ell} \}_{\ell = 1}^\N)$\;
        draw $\epart{n + 1}{i} \sim \prop{n}(\epart{n}{{\ind{n + 1}{i}}}, \cdot)$\;
        draw $\zpart{n + 1}{i} \sim \ukdist{n}(\epart{n}{{\ind{n + 1}{i}}}, \epart{n + 1}{i}, \cdot)$\;
        set $\displaystyle 
        \ewght{n + 1}{i} \gets \frac{\ukest{n}{\zpart{n + 1}{i}}(\epart{n}{{\ind{n + 1}{i}}}, \epart{n + 1}{i})}{\adjfuncforward{n}(\epart{n}{{\ind{n + 1}{i}}}) \propdens{n}(\epart{n}{{\ind{n + 1}{i}}}, \epart{n + 1}{i})}$\;
}
\caption{Pseudo-marginal forward sampling, \textsf{pmFS}.} \label{alg:pm:SMC}
\end{algorithm}

Iterating recursively, after initialisation as in Section~\ref{sec:SMC}, pseudo-marginal forward sampling yields a generalisation of the random-weight particle filter proposed in \cite{fearnhead2008particle} in the context of partially observed diffusion processes. 

\subsubsection{Pseudo-marginal backward sampling}
\label{eq:sec:backward:sampling:pseudo:marg}

Let us turn our focus to backward sampling. As intractability of $\ud{n}$ implies intractability of the kernel $\trm{n}$ (defined in \eqref{eq:def:trm}), we aim at further pseudo marginalisation by embedding $\trm{n}$ into the extended kernel  
$$
\trmext{n}(i, j, \rmd z) \propto \ewght{n}{j} \ukest{n}{z}(\epart{n}{j}, \epart{n + 1}{i}) \, \ukdist{n}(\epart{n}{j}, \epart{n + 1}{i}, \rmd z)
$$
on $\intvect{1}{\N} \times \powerset{\intvect{1}{\N}} \tensprod \Zfd_{n + 1}$. For every $i$, the marginal of $\trmext{n}(i, \cdot)$ with respect to the $j$ component is, by \eqref{eq:def:udmod}, proportional to $\ewght{n}{j} \udmod{n}(\epart{n}{j}, \epart{n + 1}{i})$, a distribution that we expect to be close to $\trm{n}(i, \cdot)$ for small $\precpar$. The intractable sampling step on Line~3 in Algorithm~\ref{alg:ideal:BS} can therefore be replaced by sampling from $\trmext{n}(i, \cdot)$, after which the auxiliary variables are discarded. The latter sampling operation will be examined in detail in the next section. This approach, which we express concisely as 
$$
    (\tstat[i]{n + 1})_{i = 1}^\N \sim \mathsf{pmBS}((\epart{n}{i}, \tstat[i]{n}, \ewght{n}{i})_{i = 1}^\N, (\epart{n + 1}{i})_{i = 1}^\N), 
$$
is summarised in Algorithm~\ref{alg:pm:backward:sampling}.  

\begin{algorithm}[h] 
    \KwData{$(\epart{n}{i}, \tstat[i]{n}, \ewght{n}{i})_{i = 1}^\N$, $(\epart{n + 1}{i})_{i = 1}^\N$}
    \KwResult{$(\tstat[i]{n + 1})_{i = 1}^\N$}
    \For{$i = 1 \to \N$}{
    \For{$j = 1 \to \K$}{
    draw $( \bi{n + 1}{i}{j}, \zpart{n + 1}{(i, j)}) \sim \trmext{n}(i, \cdot)$\;
    }
    set $\tstat[i]{n + 1} \gets \frac{1}{\K} \sum_{j = 1}^{\K} \left( \tstat[\bi{n + 1}{i}{j}]{n} + \addf{n}(\epart{n}{\bi{n + 1}{i}{j}}, \epart{n + 1}{i}) \right)$\;
}
\caption{Pseudo-marginal backward sampling, \textsf{pmBS}.} \label{alg:pm:backward:sampling}
\end{algorithm}

It remains to discuss how to sample from the extended distribution $\trmext{n}(i, \cdot)$. 
In the following we propose two possible approaches, which can be viewed as pseudo-marginal versions of the techniques discussed in Section~\ref{sec:BS}.  

\subsubsection*{Rejection sampling from $\trmext{n}$} 

Assume that there exists some measurable nonnegative function $c$ on $\Xset_{n + 1}$ such that for all $(x_{n:n + 1}, z) \in \Xset_n \times \Xset_{n + 1} \times \Zset_{n + 1}$,
$
\ukest{n}{z}(x_n, x_{n + 1}) \leq c(x_{n + 1}).  
$
Since this condition allows the Radon--Nikodym derivative of $\trmext{n}(i, \cdot)$ with respect to the probability measure 
\begin{equation} \label{eq:proposal:pm:rejection}
\rho_n^i(j, \rmd z) \propto \ewght{n}{j} \kernel{R}(\epart{n}{j}, \epart{n + 1}{i}, \rmd z)
\end{equation} 
on $\powerset{\intvect{1}{\N}} \tensprod \Zfd_{n + 1}$ to be bounded uniformly as 
$$
\frac{\rmd \trmext{n}(i, \cdot)}{\rmd \rho_n^i}(j, z) \leq \frac{c(\epart{n + 1}{i}) \sumwght{n}}{\sum_{i' = 1}^\N \ewght{n}{i'} \udmod{n}(\epart{n}{i'}, \epart{n + 1}{i})},  
$$ 
we may sample from the target $\trmext{n}(i, \cdot)$ using rejection sampling. Thus, the following procedure is iterated until acceptance: simulate a candidate $(J^\ast, \zeta^\ast)$ from $\rho_n^i$ by drawing $J^\ast \sim \cat(\{ \ewght{n}{\ell} \}_{\ell = 1}^\N)$ and $\zeta^\ast \sim \kernel{R}(\epart{n}{J^\ast}, \epart{n + 1}{i}, \rmd z)$; then accept the same with (tractable) probability 
$$
\accprobext \eqdef \frac{\ukest{n}{\zeta^\ast}(\epart{n}{J^\ast}, \epart{n + 1}{i})}{c(\epart{n + 1}{i})}. 
$$ 
Then conditionally to acceptance, the candidate has distribution $\trmext{n}(i, \cdot)$.  
Notably, the probability $\accprobext$ is obtained by simply plugging a transition density estimate provided by \hypref{assum:biased:estimate} into the probability $\accprob$ (see \eqref{eq:std:acc:prob:backward:sampling}) corresponding to the case where $\ud{n}$ is known. Moreover, since the proposal density is independent of $i$, the expected complexity of this sampling schedule is linear in the number of particles (we refer again to \cite{douc:garivier:moulines:olsson:2010}).  

\subsubsection*{MCMC sampling from $\trmext{n}$}

In some cases, bounding the estimator $\ukest{n}{z}(x_n, x_{n + 1})$ uniformly in $z$ and $x_n$ is not possible. Still, we may sample from $\trmext{n}(i, \cdot)$ using the Metropolis-Hastings algorithm with $\rho_n$ (in \eqref{eq:proposal:pm:rejection}) as independent proposal. In this case, $(\bi{n + 1}{i}{j}, \zpart{n + 1}{(i, j)})_{j = 1}^{\K}$ is a Markov chain generated recursively by the following mechanism. Given a state $\bi{n + 1}{i}{j} = J$ and $\zpart{n + 1}{(i, j)} = \zeta$, a candidate $(J^\ast, \zeta^\ast)$ for the next state is drawn from $\rho_n$ (as described above) and accepted with probability 
$$
\accprobext[MH] \eqdef 1 \wedge \frac{\ukest{n}{\zeta^\ast}(\epart{n}{J^\ast}, \epart{n + 1}{i})}{\ukest{n}{\zeta}(\epart{n}{J}, \epart{n + 1}{i})}. 
$$  
In the case of rejection, the next state is assigned the previous state. The resulting Markov chain has $\trmext{n}(i, \cdot)$ as stationary distribution and similar to the case of rejection sampling, the acceptance probability $\accprobext[MH]$ can be viewed as a plug-in estimate of the corresponding probability $\accprob[MH]$ (see \eqref{eq:std:MH:prob:backward:sampling}). 

\subsubsection{Pseudo-marginal PaRIS: full update} 

Combining the pseudo-marginal forward and backward sampling operations yields a pseudo-marginal PaRIS update described in the following algorithm, which is the main contribution of this section.  

\begin{algorithm}[h] 
    \KwData{$(\epart{n}{i}, \tstat[i]{n}, \ewght{n}{i})_{i = 1}^\N$}
    \KwResult{$(\epart{n + 1}{i}, \tstat[i]{n + 1}, \ewght{n + 1}{i})_{i = 1}^\N$}
    run $(\epart{n + 1}{i}, \ewght{n + 1}{i})_{i = 1}^\N \sim \mathsf{pmFS}((\epart{n}{i}, \ewght{n}{i})_{i = 1}^\N)$\;
    run $(\tstat[i]{n + 1})_{i = 1}^\N \sim \mathsf{pmBS}((\epart{n}{i}, \tstat[i]{n}, \ewght{n}{i})_{i = 1}^\N, (\epart{n + 1}{i})_{i = 1}^\N)$\; 
    \medskip
\caption{Full pseudo-marginal PaRIS update.} \label{alg:pm:PaRIS}
\end{algorithm}

Algorithm~\ref{alg:pm:PaRIS} is initialised by drawing $(\epart{0}{i})_{i = 1}^\N \sim \chi^{\varotimes \N}$ and letting $\ewght{0}{i} = \rmd \chi / \rmd \init(\epart{0}{i})$ and $\tstat[i]{n} = 0$.

We now illustrate \hypref{assum:biased:estimate} by a few examples. 

\begin{example}[Durham--Gallant estimator]
\label{eq:durham:gallant}
As an illustration, we return to the state-space model framework discussed in Example~\ref{ex:state-space:models}. Let $\Xset \eqdef \rset^{d_x}$ and $\Yset \eqdef \rset^{d_y}$ be equipped with their respective Borel $\sigma$-fields $\Xfd$ and $\Yfd$, and let $(X_t)_{t > 0}$ be some diffusion process on $\Xset$ driven by the homogeneous stochastic differential equation
\begin{equation} \label{eq:SDE}
\rmd X_t = \mu(X_t) \, \rmd t + \sigma(X_t) \, \rmd W_t, \quad t > 0, 
\end{equation}
where $X_0 = 0$, $(W_t)_{t > 0}$ is $d_x$-dimensional Brownian motion, $b : \Xset \to \Xset$ and $\sigma : \Xset \to \Xset^2$ are twice differentiable with bounded first and second order derivatives. In addition, the matrix $\sigma \sigma^\intercal$ is uniformly non-degenerate. Let $(\mathcal{F}_t)_{t > 0}$ be the natural filtration generated by the process $(X_t)_{t > 0}$. The state sequence $(X_t)_{t > 0}$ is latent but partially observed at discrete time points $(t_n)_{n \in \nsetpos}$ which are assumed to be equally spaced for simplicity, \ie, $t_n = t_1 + \delta (n - 1)$ for all $n$ and some $\delta > 0$. Abusing notations, we denote $X_n \eqdef X_{t_n}$ and let $q_\delta$ be the transition density of $(X_n)_{n \in \nsetpos}$. Denote by $\kernel{Q}_\delta$ the transition kernel induced by $q_\delta$. In general, $q_\delta$ is intractable, which makes the problem of computing online, for a given data stream $(y_n)_{n \in \nsetpos}$ in $\Yset$, the sequence of joint-smoothing distributions \eqref{eq:smooth} in models of this sort very challenging. Still, using the Euler scheme, one may, for small $\delta$, approximate $q_\delta$ by 
$$
\bar{q}_\delta(x_n, x_{n + 1}) \eqdef \phi(x_{n + 1}; x_n + \delta \mu(x_n), \delta \sigma^2(x_n)), 
$$ 
where $\phi(\cdot; m, s^2)$ is the density of the Gaussian distribution with mean $m$ and variance $s^2$. Let $\bar{\kernel{Q}}_\delta$ be the transition kernel induced by $\bar{q}_\delta$. Since the approximation $\bar{q}_\delta$ is poor for $\delta$ not small enough, we may instead, as suggested in \cite{durham:gallant:2002}, pick some finer step size $\precpar \in \precparsp_\delta \eqdef \{ \delta / n : n \in \nsetpos \}$ and estimate the density $q_\delta(x_n, x_{n + 1})$ by $q_\delta \langle \zeta \rangle (x_n, x_{n + 1})$, where $\zeta = (\zeta^i)_{i = 1}^L$ are independent draws from some proposal $r(x_n, x_{n + 1}, z) \, \rmd z$ on $\Xset^2 \times \Xfd^{\delta / \precpar - 1}$ and  
$$
q_\delta \langle z \rangle (x_n, x_{n + 1}) \eqdef \frac{1}{L} \sum_{i = 1}^L \frac{\prod_{k = 1}^{\delta / \precpar} \bar{q}_\precpar(z_{k - 1}^i, z_k^i)}{r(x_n, x_{n + 1}, z^i)},  
$$
with $z^i = (z_1^i, \ldots, z_{\delta / \precpar - 1}^i)$ and, by convention, $z_0^i = x_n$ and $z_{\delta / \precpar}^i = x_{n + 1}$. In \cite{durham:gallant:2002}, $r(x_n, x_{n + 1}, z) \, \rmd z$ is the distribution of a discretised (possibly modified) \emph{Brownian bridge}, \ie, Brownian motion started at $x_n$ and conditioned to terminate at $x_{n + 1}$. 

Let $Y_n$ denote the $\Yset$-valued observation at time $t_n$. We will consider two different models for the observation process $(Y_n)_{n \in \nset}$.    

\subsubsection*{Case~1}
First, assume that for all $n \in \nsetpos$, 
$$
Y_n \mid \mathcal{F}_{t_n} \sim \md{n - 1}(X_{n - 1}, X_n, y_n) \, \rmd y_n, 
$$
where $\md{n - 1}$ is some tractable transition density with respect to Lebesgue measure. 
In this case, \hypref{assum:biased:estimate} holds with the estimator 
$$
\ukest{n}{z}(x_n, x_{n + 1}) = q_\delta \langle z \rangle (x_n, x_{n + 1}) \md{n}(x_n, x_{n + 1}, y_{n + 1})
$$ 
and the instrumental kernel 
$$\ukdist{n}(x_n, x_{n + 1}, \rmd z) = \prod_{i = 1}^L r(x_n, x_{n + 1}, z^i) \, \rmd z^i.
$$
Finally, we note that 
\begin{equation} \label{eq:ukmod:durham:gallant}
\ukmod{n}(x_n, \rmd x_{n + 1}) = \bar{\kernel{Q}}_\precpar^{\delta / \precpar}(x_n, \rmd x_{n + 1}) \, \md{n}(x_n, x_{n + 1}, y_{n + 1}), 
\end{equation}
which is generally intractable.  

\subsubsection*{Case 2} Alternatively, we may assume that $(Y_n)_{n \in \nsetpos}$ are discrete observations of the solution to the stochastic differential equation 
$$
\rmd Y_t = \tilde{\mu}(X_t, Y_t) \, \rmd t + \tilde{\sigma}(X_t, Y_t) \, \rmd \tilde{W}_t, \quad t > 0, 
$$
where $Y = 0$, $(\tilde{W}_t)_{t > 0}$ is $d_y$-dimensional Brownian motion independent of $(W_t)_{t > 0}$ and $\tilde{\mu} : \Xset \times \Yset \to \Yset$ and $\tilde{\sigma} : \Xset \times \Yset \to \Yset^2$ are known functions which are twice differentiable with bounded first and second order derivatives. In addition, the matrix $\tilde{\sigma} \tilde{\sigma}^\intercal$ is uniformly non-degenerate. Denote by $p_\delta$ the transition density of $(X_t, Y_t)_{t > 0}$. In this case, the joint-smoothing distributions can, for a given data stream $(y_n)_{n \in \nsetpos}$, be expressed as path measures \eqref{eq:FK:path} induced by $\ud{n}(x_n, x_{n + 1}) = p_\delta(x_n, y_n, x_{n + 1}, y_{n + 1})$, $n \in \nset$. Since also $p_\delta$ is generally intractable we subject the bivariate process to Euler discretisation, yielding the approximation 
\begin{multline}
\bar{p}_\delta(x_n, y_n, x_{n + 1}, y_{n + 1}) \\
\eqdef \phi \left( x_{n + 1}, y_{n + 1}; (x_n + \delta \mu(x_n), y_n + \delta \tilde{\mu}(x_n, y_n))^\intercal, \delta \, \mbox{diag}(\sigma^2(x_n), \tilde{\sigma}^2(x_n, y_n)) \right) \label{eq:td:bivariate:process}
\end{multline}
of $p_\delta$. Denote by $\bar{\kernel{P}}_\delta$ the Markov kernel induced by $\bar{p}_\delta$. In the case of sparse observations we may improve the approximation $\bar{p}_\delta$ by picking again some finer step size $\precpar \in \mathcal{E}_\delta$ and computing (by swapping, in~\eqref{eq:ukmod:durham:gallant}, $\bar{q}_\delta$ for $\bar{p}_\delta$ and letting $r(x_n, y_n, x_{n + 1}, y_{n + 1}, z) \, \rmd z$ be the distribution of a discretised, $\rset^{d_x + d_y}$-valued Brownian bridge started in $(x_n, y_n)$ and conditioned to terminate in $(x_{n + 1}, y_{n + 1})$) the Durham--Gallant estimator $p_\delta \langle z \rangle$. In this case, $\ukest{n}{z}(x_n, x_{n + 1}) = p_\delta \langle z \rangle (x_n, y_n, x_{n + 1}, y_{n + 1})$, which yields
\begin{equation} \label{eq:ukmod:durham:gallant:case:2}
\ukmod{n}(x_n, \rmd x_{n + 1}) = \bar{p}_\precpar^{\delta / \precpar}(x_n, y_n, x_{n + 1}, y_{n + 1}) \, \rmd x_{n + 1}, 
\end{equation}
where $\bar{p}_\precpar^{\delta / \precpar}$ denotes the transition density of (the $\delta / \precpar$-skeleton) $\bar{\kernel{P}}_\precpar^{\delta / \precpar}$.  
\end{example}

\begin{example}[the exact algorithm] \label{ex:exact:algorithm}
We consider again the partially observed diffusion process model in Example~\ref{eq:durham:gallant}. In the special case where the diffusion process governed by \eqref{eq:SDE} can be transformed into one with a constant diffusion term through the \emph{Lamperti transformation}, it was shown in \cite{beskos:papaspiliopoulos:roberts:fearnhead:2006,fearnhead2008particle} how unbiased estimation of $q_\delta$ can be achieved using generalised Poisson estimators. In our setup, this simply yields $\udmod{n} = \ud{n}$ for all $n$. We refer to the mentioned papers for details. 
\end{example}

\begin{example} 
[approximate Bayesian computation smoothing]
Consider smoothing in a fully dominated general state-space HMM for which the state likelihood functions $(\md{n}(\cdot, y_n))_{n \in \nset}$ are intractable (or expensive to evaluate) for any given sequence $(y_n)_{n \in \nset}$ of observations in $\rset^{d_y}$. In the case where it is possible (or faster) to sample observation emissions according to the kernels $(\mk_n)_{n \in \nset}$ one may then 
take an \emph{approximate-Bayesian-computation} (ABC) approach (see {\eg}  \cite{marin:pudlo:robert:ryder:2012}), and replace any value $\md{n}(x_n, y_n)$ by a point estimate $\kappa_\precpar(\zeta_n - y_n)$, where $\zeta_n \sim \mk_n(x_n, \cdot)$ and $\kappa_\varepsilon$ is a $d_y$-dimensional kernel density scaled by some bandwidth $\precpar > 0$. In \cite{martin:jasra:singh:whiteley:delmoral:maccoy:2014}, the authors apply the forward-only smoothing approach of \cite{delmoral:doucet:singh:2010} to this approximate model, yielding a particle-based ABC smoothing algorithm. Also this framework is covered by \hypref{assum:biased:estimate}, by letting $\ukest{n}{z}(x_n, x_{n + 1}) = \hd{n}(x_n, x_{n + 1}) \kappa_\precpar(z - y_{n + 1})$ and $\ukdist{n}(x_n, x_{n + 1}, \rmd z) = \mk_n(x_{n + 1}, \rmd z)$. In this case, 
\begin{equation} \label{eq:ukmod:ABC:smoothing}
\ukmod{n}(x_n, \rmd x_{n + 1}) = \hk_n(x_n, \rmd x_{n + 1}) \int \kappa_\precpar(z - y_{n + 1}) \, \mk_n(x_{n + 1}, \rmd z). 
\end{equation}
\end{example}

\section{Theoretical results}
\label{sec:theoretical:results}
\subsection{Convergence of pseudo-marginal PaRIS estimates}
\label{sec:convergence:results}

\subsubsection{Convergence of Algorithm~\ref{alg:pm:PaRIS}}
\label{sec:convergence:pm:PaRIS}

In \cite{olsson:westerborn:2014b}, the authors established strong consistency and asymptotic normality of the PaRIS in the framework of fully dominated general state-space HMMs and the bootstrap particle filter, \ie, in the simple case where $\adjfuncforward{n} \equiv 1$ and $\kissforward{n}{n} \equiv \hd{n}$ for all $n$. In the following we will extend these results to the considerably more general setting comprising models \eqref{eq:def:post} and the pseudo-marginal PaRIS in Algorithm~\ref{alg:pm:PaRIS}. More precisely, we will show that each weighted sample $(\epart{n}{i}, \tstat[i]{n}, \ewght{n}{i})_{i = 1}^N$, $n \in \nset$, produced by Algorithm~\ref{alg:pm:PaRIS} satisfies exponential concentration (Theorem~\ref{cor:hoeffding:tau:marginal}) and asymptotic normality (Theorem~\ref{cor:clt:pseudo:marginal:paris}) with respect to the expected additive functional $h_n$ under the \emph{skew} path model 
$$
\postmod{0:n}(\rmd x_{0:n}) \propto \chi(\rmd x_0) \prod_{m = 0}^{n - 1} \ukmod{m}(x_m, \rmd x_{m + 1}), \quad n \in \nset. 
$$
Even though these results are established along the lines of the corresponding proofs in \cite{olsson:westerborn:2014b}, it is the matter of non-trivial adaptations. As explained in Section~\ref{sec:introduction}, a random-weight particle filter (iterated pseudo-marginal forward sampling) can be viewed as a standard particle filter evolving on an extended state space comprising also the states of the auxiliary variables, and hence its convergence follows from standard SMC convergence results \cite{fearnhead2008particle}. On the contrary, since Algorithm~\ref{alg:pm:PaRIS} involves \emph{two} levels of pseudo marginalisation (with respect to both the forward-sampling and the backward-sampling operations), it cannot be described equivalently as a special instance of the original PaRIS (even in its general form given by Algorithm~\ref{alg:ideal:PaRIS}) operating on an extended space. Thus, there is no free lunch when it concerns the theoretical analysis of this scheme. Furthermore, in the case of fully dominated HMMs and when forward sampling is guided by the bootstrap filter, which was the setting in \cite{olsson:westerborn:2014b}, the conditional distribution of the particles given their ancestors (\ie, the marginal of $\rho_n$ in \eqref{eq:cond:instrumental:mixture} with respect to $x$) coincides, at any time point $n$, with the denominator of the backward kernel. Mathematically, this enables a cancellation that facilitates the analysis significantly. However, this simplification is not possible once the particle dynamics is guided by general proposal kernels, which is necessarily the case in Algorithm~\ref{alg:pm:PaRIS}. This complicates the analysis; see Remark~\ref{rem:non-trivial:extension} for further details.

To be able to describe our results in full detail, we also need to introduce, for every $n \in \nset$, the skew backward Markov kernels 
$$
    \bkmod{n}(x_{n + 1}, \rmd x_n) \eqdef \frac{\postmod{n}(\rmd x_n) \, \udmod{n}(x_n, x_{n + 1})}{\postmod{n}[\udmod{n}(\cdot, x_{n + 1})]}
$$
on $\Xset_{n + 1} \times \Xfd_n$ as well as the joint law 
\begin{equation} \label{eq:skew:backward:law}
\tstatmod{n}(x_n, \rmd x_{0:n - 1}) \eqdef \prod_{m = 0}^{n - 1} \bkmod{m}(x_{m + 1}, \rmd x_m)
\end{equation}
on $\Xset_n \times \Xfd^{n - 1}$. 

The analysis will be carried through under the following assumption.
\begin{hypH}
\label{assum:bound:filter:pseudomarginal}
For all $n \in \nset$, the functions $\adjfuncforward{n}$ and  
\begin{align*}
&\wgtfuncext{n} : \Xset_n \times \Xset_{n + 1} \times \Zset_{n + 1} \ni (x_n, x_{n + 1}, z) \mapsto \frac{\ukest{n}{z}(x_n, x_{n + 1})}{\adjfuncforward{n}(x) \kissforward{n}{n}(x_n, x_{n + 1})}, \\
&\wgtfuncmod{n} : \Xset_n \times \Xset_{n + 1} \ni (x_n, x_{n + 1}) \mapsto \frac{\udmod{n}(x_n, x_{n + 1})}{\adjfuncforward{n}(x_n) \kissforward{n}{n}(x_n, x_{n + 1})}
\end{align*}
are bounded. So is also $\initwgtfunc : \Xset_0 \ni x_0 \mapsto \rmd \chi / \rmd \init(x_0)$. 
\end{hypH}


\subsubsection*{Hoeffding inequality}

Our first theoretical result is the following Hoeffding-type concentration inequality, which also plays a critical role for the derivation of the central limit theorem (CLT) in Theorem~\ref{cor:clt:pseudo:marginal:paris}. For every $n \in \nset$, let $\bmaf{\Xfd^n}$ denote the set of additive functionals \eqref{eq:add:func} such that $\addf{m} \in \bmf{\Xfd_m \tensprod \Xfd_{m + 1}}$ for all $m \in \intvect{0}{n - 1}$.   
\begin{theorem}
\label{cor:hoeffding:tau:marginal}
Assume \hypref[assum:biased:estimate]{assum:bound:filter:pseudomarginal}. Then for every $n \in \nset$, $\precpar \in \precparsp$, $h_n \in \bmaf{\Xfd^n}$, and $\K \in \nsetpos$ there exists $(c_n, d_n) \in \rsetpos^2$ such that for all $\N \in \nsetpos$ and $\epsilon > 0$,
$$
\pP\left(\left| \sum_{i = 1}^\N \frac{\ewght{n}{i}}{\sumwght{n}} \tstat[i]{n} - \postmod{0:n} h_n \right| \geq \epsilon \right) \leq c_n \exp \left( - d_n \N \epsilon^2 \right).
$$
\end{theorem}

The proof of Proposition~\ref{cor:hoeffding:tau:marginal}, which is an adaptation of the proof of \cite[Theorem~1]{olsson:westerborn:2014b}, is presented in Section~\ref{sec:proof:prop:hoeffding:tau:marginal}. Since we proceed by induction and the objective functions $(h_n)_{n \in \nset}$ are additive, it is, following \cite{olsson:westerborn:2014b}, necessary to establish the result for estimators in the form $\sum_{i = 1}^\N \ewght{n}{i} \{ f_n(\epart{n}{i}) \tstat[i]{n} + \ftd{n}(\epart{n}{i}) \} / \sumwght{n}$, where $(f_n, \ftd{n}) \in \bmf{\Xfd_n}^2$. This is done in Proposition~\ref{prop:hoeffding:tau:marginal}, and Theorem \ref{cor:hoeffding:tau:marginal} follows as a corollary of that result. 

The previous bound describes the pointwise convergence of the estimator delivered by Algorithm~\ref{alg:pm:PaRIS} as $\N$ grows for $n$ fixed, and here no attempt has been made to obtain a bound that is uniform in $n$. As we shall see shortly, in Theorem~\ref{thm:variance:bound}, the numerical stability of the algorithm can instead be established by bounding the asymptotic variance of the estimator. Finally, note that the previous bound implies that the estimator $\sum_{i = 1}^\N \ewght{n}{i} \tstat[i]{n} / \sumwght{n}$ tends $\pP$-a.s. to $\postmod{0:n} h_n$ as $\N$ tends to infinity. 


\subsubsection*{Central limit theorem}

We now focus on the asymptotic properties of Algorithm~\ref{alg:pm:PaRIS} and furnish, in Theorem~\ref{cor:clt:pseudo:marginal:paris} below, this scheme with a CLT. On the basis of this result, the stochastic stability of the algorithm is expressed by establishing, in Theorem~\ref{thm:variance:bound}, an $\mathcal{O}(n)$ bound on the asymptotic variances. In order to be able to state these results accurately, we need some additional notation. First, let 
 for $(x_n, x_{n + 1}) \in \Xset_n \times \Xset_{n + 1}$, 
\begin{multline} \label{eq:def:ukestvar}
\ukestvar{n}(x_n, x_{n + 1}) \\
\eqdef \frac{1}{\wgtfuncmod{n}(x_n, x_{n + 1})} \int \{ \wgtfunc{n}(x_n, x_{n + 1}, z) - \wgtfuncmod{n}(x_n, x_{n + 1}) \}^2 \, \ukdist{n}(x_n, x_{n + 1}, \rmd z)
\end{multline}
denote the conditional relative variance of the (random) weight $\ewght{n + 1}{i}$ given $\epart{n}{\ind{n}{i}} = x_n$ and $\epart{n + 1}{i}= x_{n + 1}$. In addition, for each $n \in \nset$ and $m \in \intvect{0}{n}$, define the kernel  
\begin{equation} \label{eq:def:retrokmodmod}
    \retrokmodmod_{m, n}(x_m', \rmd x_{0:n}) \eqdef \delta_{x_m'}(\rmd x_m) \,  
    \tstatmod{m}(x_m, \rmd x_{0:m - 1})
    \prod_{\ell = m}^{n - 1} \ukmod{\ell}(x_\ell, \rmd x_{\ell + 1})
\end{equation}
on $\Xset_m \times \Xfd^n$ as well as the centered version 
\begin{equation} \label{eq:def:retrokmodmodnorm}
\retrokmodmodnorm_{m, n} h (x_m) \eqdef  \retrokmodmod_{m, n}(h - \postmod{0:n} h)(x_m) 
\end{equation}
on the same space. 

The following is the main result of this section. 
\begin{theorem} \label{cor:clt:pseudo:marginal:paris}
Assume~\hypref[assum:biased:estimate]{assum:bound:filter:pseudomarginal}. Then for all $n \in \nset$, $\precpar \in \precparsp$, $\K \in \nsetpos$, and $h_n \in \bmaf{\Xfd^n}$, 
$$
 \sqrt{\N} \left( \sum_{i = 1}^\N \frac{\ewght{n}{i}}{\sumwght{n}} \tstat[i]{n} - \postmod{0:n} h_n  \right) 
  \dlim \sigma_n(h_n) Z, 
$$
where $Z$ is standard normally distributed and 
\begin{equation} \label{eq:as:var:decomp}
\sigma^2_n (h_n) \eqdef \frac{\chi(\initwgtfunc \retrokmodmodnorm_{0, n} h_n)^2}{(\chi \ukmod{0, n - 1} \1_{\Xset^n})^2} + \sigma^2_n \langle (\wgtfuncmod{\ell})_{\ell = 0}^{n - 1} \rangle (h_n) + \sigma^2_n \langle (\ukestvar{\ell})_{\ell = 0}^{n - 1} \rangle (h_n) 
\end{equation}
and
\begin{multline} \label{eq:non-recursive:as:var}
\sigma^2_n \langle (\varphi_\ell)_{\ell = 0}^{n - 1} \rangle (h_n) 
\eqdef \sum_{m = 0}^{n - 1} \frac{\postmod{m} \adjfuncforward{m} \postmod{m} \ukmod{m}(\varphi_m [\retrokmodmodnorm_{m + 1, n} h_n ]^2)}{(\postmod{m} \ukmod{m, n - 1} \1_{\Xset^n})^2}\\
+ \sum_{m = 0}^{n - 1} \sum_{\ell = 0}^m \frac{\postmod{m} \adjfuncforward{m} \postmod{\ell} \ukmod{\ell} \{\bkmod{\ell}(\tstatmod{\ell} h_{\ell} + \addf{\ell} - \tstatmod{\ell + 1} h_{\ell +1})^2 \ukmod{\ell + 1, m}( \bkmod{m} \varphi_m [\ukmod{m + 1, n - 1} \1_{\Xset^n}]^2
)\}}{\K^{m - \ell + 1} (\postmod{\ell} \ukmod{\ell, m - 1} \1_{\Xset^m})(\postmod{m} \ukmod{m, n - 1} \1_{\Xset^n})^2}
\end{multline}
for any sequence $(\varphi_\ell)_{\ell \in \nset}$ of measurable functions $\varphi_\ell : \Xset_\ell \times \Xset_{\ell + 1} \to \rsetnn$. 
\end{theorem}

\begin{remark}
The first term of \eqref{eq:as:var:decomp} corresponds to the contribution of the initialisation step to the asymptotic variance. This term is incorrectly missing in the asymptotic-variance expressions given in \cite[Theorem~3 and Corollary~5]{olsson:westerborn:2017}. The last term corresponds to the additional variance induced by the estimation of $(\ud{n})_{n \in \nset}$ regulated by \hypref{assum:biased:estimate}. Finally, as we shall see in Section~\ref{sec:implied:convergence:results}, the first two terms correspond jointly to the variance of the ideal PaRIS in Algorithm~\ref{alg:ideal:PaRIS}, for which $(\ud{n})_{n \in \nset}$ are assumed to be known and tractable. 
\end{remark}

Theorem~\ref{cor:clt:pseudo:marginal:paris} is established in Section~\ref{sec:proof:prop:clt:pseudo:marginal:paris}. The structure of the proof is adopted from \cite{olsson:westerborn:2017} (however, we remark again that it is, as explained above, the matter of a non-trivial extension), and the CLT in Theorem~\ref{cor:clt:pseudo:marginal:paris} is obtained directly as a corollary of a more general CLT for estimators of form $\sum_{i = 1}^\N \ewght{n}{i} \{ f_n(\epart{n}{i}) \tstat[i]{n} + \ftd{n}(\epart{n}{i}) \} / \sumwght{n}$, where $(f_n, \ftd{n}) \in \bmf{\Xfd_n}^2$; see Theorem~\ref{prop:clt:pseudo:marginal:paris}.  

\subsubsection{Convergence of Algorithm~\ref{alg:pm:SMC}} 
\label{sec:implied:convergence:results}

Our analysis provides, as a by-product, also a CLT for the random-weight particle filter obtained by iterating the forward-sampling operation in Algorithm~\ref{alg:pm:SMC}; indeed, this result, which is similar to \cite[Theorem~3]{fearnhead2008particle}, follows immediately by letting $f_n \equiv 0$ in Theorem~\ref{prop:clt:pseudo:marginal:paris}, and we state it below for completeness. 

\begin{proposition} \label{prop:clt:pseudo:marginal:filter}
Assume~\hypref[assum:biased:estimate]{assum:bound:filter:pseudomarginal} and let $(\epart{n}{i}, \ewght{n}{i})_{i = 1}^\N$, $n \in \nset$, be generated by Algorithm~\ref{alg:pm:SMC}. Then for all $n \in \nset$, $\precpar \in \precparsp$, and $h \in \bmf{\Xfd_n}$, 
$$
 \sqrt{\N} \left( \sum_{i = 1}^\N \frac{\ewght{n}{i}}{\sumwght{n}} h(\epart{n}{i}) - \postmod{n} h \right) 
  \dlim \tilde{\sigma}_n(h) Z, 
$$
where $Z$ is standard normally distributed and
\begin{equation} \label{eq:as:var:pm:filter}
\tilde{\sigma}^2_n(h) \eqdef \frac{\chi\{ \initwgtfunc \ukmod{0} \cdots \ukmod{n - 1}(h - \postmod{n} h) \}^2}{(\chi \ukmod{0} \cdots \ukmod{n - 1} \1_{\Xset_n})^2} + \tilde{\sigma}^2_n \langle (\wgtfuncmod{\ell})_{\ell = 0}^{n - 1} \rangle (h) + \tilde{\sigma}^2_n \langle (\ukestvar{\ell})_{\ell = 0}^{n - 1} \rangle (h) 
\end{equation}
and
\begin{equation} \label{eq:non-recursive:filter:as:var}
\tilde{\sigma}^2_n \langle (\varphi_\ell)_{\ell = 0}^{n - 1} \rangle (h) 
\eqdef \sum_{m = 0}^{n - 1} \frac{\postmod{m} \adjfuncforward{m} \postmod{m} \ukmod{m}(\varphi_m [\ukmod{m + 1} \cdots \ukmod{n - 1} (h - \postmod{n} h)
]^2)}{(\postmod{m} \ukmod{m} \cdots \ukmod{n - 1} \1_{\Xset_n})^2} 
\end{equation}
for any sequence $(\varphi_\ell)_{\ell \in \nset}$ of measurable functions $\varphi_\ell : \Xset_\ell \times \Xset_{\ell + 1} \to \rsetnn$. 
\end{proposition}

\subsubsection{Convergence of Algorithm~\ref{alg:ideal:PaRIS}}

Importantly, Theorem~\ref{cor:clt:pseudo:marginal:paris} provides, as another by-product, also a CLT for the ideal PaRIS in Algorithm~\ref{alg:ideal:PaRIS}. Indeed, in the case where every $\ud{n}$ is tractable, we may set 
\begin{equation} \label{eq:ideal:case}
\ukest{n}{z}(x_n, x_{n + 1}) = \ud{n}(x_n, x_{n + 1})
\end{equation} 
for all $z$ and define $\ukdist{n}$ arbitrarily. This implies that the relative variance \eqref{eq:def:ukestvar} is identically zero, which eliminates the last term of \eqref{eq:as:var:decomp}. Thus, in this case the asymptotic variance is given by the first two terms of \eqref{eq:as:var:decomp}, but now induced by the original dynamics $(\uk{n})_{n \in \nset}$ (as \eqref{eq:ideal:case} implies that also $\udmod{n} = \ud{n}$). This result is formulated in the following corollary, where we have defined, for each $n \in \nset$ and $m \in \intvect{0}{n}$, the kernel  
\begin{equation} \label{eq:def:retrok}
    \retrok_{m, n}(x_m', \rmd x_{0:n}) \eqdef \delta_{x_m'}(\rmd x_m) \,  
    \tstat{n}(x_m, \rmd x_{0:m - 1})
    \prod_{\ell = m}^{n - 1} \uk{\ell}(x_\ell, \rmd x_{\ell + 1})
\end{equation}
on $\Xset_m \times \Xfd^n$ as well as the centered version 
$$
\retroknorm_{m, n} h (x_m) \eqdef  \retrok_{m, n}(h - \post{0:n} h)(x_m) 
$$
on the same space. 

\begin{corollary} \label{cor:clt:ideal:paris}
For each $n \in \nset$, assume that the functions $\adjfuncforward{n}$ and  
$$
\wgtfuncideal{n} : \Xset_n \times \Xset_{n + 1} \ni (x_n, x_{n + 1}) \mapsto \frac{\ud{n}(x_n, x_{n + 1})}{\adjfuncforward{n}(x) \kissforward{n}{n}(x_n, x_{n + 1})} 
$$
are bounded and let $(\epart{n}{i}, \tstat[i]{n}, \ewght{n}{i})_{i = 1}^\N$, $n \in \nset$, be generated by Algorithm~\ref{alg:ideal:PaRIS}. Then for all $n \in \nset$, $\precpar \in \precparsp$, $\K \in \nsetpos$, and $h_n \in \bmaf{\Xfd^n}$, 
$$
\sqrt{\N} \left( \sum_{i = 1}^\N \frac{\ewght{n}{i}}{\sumwght{n}} \tstat[i]{n} - \post{0:n} h_n  \right) \dlim \bar{\sigma}_n (h_n) Z, 
$$
where $Z$ is standard normally distributed and
\begin{multline} \label{eq:as:var:ideal:PaRIS}
\bar{\sigma}^2_n (h_n) 
\eqdef \frac{\chi(\initwgtfunc \retroknorm_{0, n} h_n)^2}{(\chi \ukmod{0, n - 1} \1_{\Xset^n})^2} + \sum_{m = 0}^{n - 1} \frac{\post{m} \adjfuncforward{m} \post{m} \uk{m}(\wgtfuncideal{m} [\retroknorm_{m + 1, n} h_n ]^2)}{(\post{m} \uk{m, n - 1} \1_{\Xset_n})^2} \\
+ \sum_{m = 0}^{n - 1} \sum_{\ell = 0}^m \frac{\post{m} \adjfuncforward{m} \post{\ell} \uk{\ell} \{\bkw{\ell}(\tstat{\ell} h_{\ell} + \addf{\ell} - \tstat{\ell + 1} h_{\ell +1})^2 \uk{\ell + 1, m}( \bkw{m} \wgtfuncideal{m} [\uk{m + 1, n - 1} \1_{\Xset^n}]^2
)\}}{\K^{m - \ell + 1} (\post{\ell} \uk{\ell, m - 1} \1_{\Xset_m})(\post{m} \uk{m, n - 1} \1_{\Xset_n})^2}. 
\end{multline}
\end{corollary}

Corollary~\ref{cor:clt:ideal:paris} extends \cite[Corollary~5]{olsson:westerborn:2017} to general models \eqref{eq:def:post} and auxiliary-particle-filter-guided forward sampling. 

\subsection{Long-term stochastic stability}

In this section we establish the long-term stochastic stability of Algorithm~\ref{alg:pm:PaRIS} by providing an $\mathcal{O}(n)$ bound on the sequence $(\sigma^2_n (h_n))_{n \in \nset}$ for $\K \geq 2$. Using $\K \geq 2$ is critical, since, as noted in \cite{olsson:westerborn:2014b}, using $\K = 1$ in the PaRIS leads to a path-degeneracy phenomenon similar to that of the naive smoother described in Section~\ref{sec:SMC}; we refer to \cite[Section~3.1]{olsson:westerborn:2014b} for a detailed discussion. Still, the variance bounds that we will present are of order $n(1 + 1/(\K - 1))$, which means that large $\K$ do not serve to reduce the variance significantly. This is in good agreement with simulations, where $\K = 2$ leads generally to good results. In addition, we shall see that our analysis, which is carried through in detail in Section~\ref{sec:variance:bounds}, yields, as by-products, a similar bound for the ideal PaRIS in Algorithm~\ref{alg:ideal:PaRIS} as well as a time-uniform bound on the sequence $(\tilde{\sigma}^2(h))_{n \in \nset}$ of asymptotic variances of the random-weight-particle-filter estimators generated by Algorithm~\ref{alg:pm:SMC}. As far as we know, the latter result is the first of its kind. The analysis will be carried through under the following strong-mixing assumption, which is now classical (see, \eg, \cite[Chapter~4]{delmoral:2004} and \cite[Section~4.3]{Cappe:2005:IHM:1088883}) and typically requires the state spaces $(\Xset_n)_{n \in \nset}$ to be compact sets. 
\begin{hypH}
\label{assum:strong:mixing}
There exist constants $0 < \udlow < \udup < \infty$ such that for all $m \in \nset$ and $(x_m, x_{m + 1}) \in \Xset_m \times \Xset_{m + 1}$, 
$$
    \udlow \leq \ud{m}(x_m, x_{m + 1}) \leq \udup
$$ 
and 
$$
    \udlow \leq \inf_{\precpar \in \precparsp} \udmod{m}(x_m, x_{m + 1}), \quad \sup_{\precpar \in \precparsp} \udmod{m}(x_m, x_{m + 1}) \leq \udup. 
$$ 
\end{hypH}
Note that under \hypref{assum:strong:mixing}, each reference measure $\mu_m$ is finite; we may hence, without loss of generality, assume that each $\mu_m$ is a probability measure. Under \hypref{assum:strong:mixing}, define the mixing rate 
\begin{equation} \label{eq:def:rho}
    \rho \eqdef 1 - \frac{\udlow}{\udup}
\end{equation}
as well as the constants 
$$
c(\sigma_\pm) \eqdef \frac{1}{\rho^2 (1 - \rho)^5 \udlow}, \quad d(\sigma_\pm) \eqdef \frac{1}{(1- \rho)^4 \udlow^2} \left( 2 + \frac{1}{1 - \rho} \right)^2, 
$$
and 
$$
\cboundtd \eqdef \frac{1}{(1 - \rho)^3 \udlow} \left( \frac{1}{\rho^2(1 - \rho^2)} + 1 \right). 
$$

Having introduced these quantities, we are ready to present the main result of this section. 

\begin{theorem} \label{thm:variance:bound}
Assume \hypref[assum:biased:estimate]{assum:strong:mixing}. 
Then for every $\precpar \in \precparsp$, $\K \geq 2$, $h_n \in \bmaf{\Xfd^n}$, and sequence $(\varphi_\ell)_{\ell \in \nset}$ of bounded measurable functions $\varphi_\ell : \Xset_\ell \times \Xset_{\ell + 1} \to \rsetnn$, 
\begin{multline*}
\limsup_{n \to \infty} \frac{1}{n} \sigma^2_n \langle (\varphi_\ell)_{\ell = 0}^{n - 1} \rangle (h_n) \\
\leq \left( \cbound + \frac{1}{\K - 1} \dbound \right) 
\sup_{\ell \in \nset} \| \addf{\ell} \|_\infty^2 \sup_{\ell \in \nset} \| \adjfuncforward{\ell} \|_\infty \sup_{\ell \in \nset} \| \varphi_\ell \|_\infty.  
\end{multline*}
\end{theorem}

Using Theorem~\ref{thm:variance:bound}, $\mathcal{O}(n)$ bounds on the asymptotic variances of Algorithm~\ref{alg:pm:PaRIS} and Algorithm~\ref{alg:ideal:PaRIS} are readily obtained. 

\begin{corollary} \label{cor:variance:bound}
Assume \hypref[assum:biased:estimate]{assum:strong:mixing}. Then for every $\precpar \in \precparsp$, $\K \geq 2$, and $h_n \in \bmaf{\Xfd^n}$, 
\begin{multline} \label{eq:variance:bound:full:monty}
\limsup_{n \to \infty} \frac{1}{n} \sigma^2_n(h_n) \\ 
\leq \left( \cbound + \frac{1}{\K - 1} \dbound \right) \sup_{\ell \in \nset} \| \addf{\ell} \|_\infty^2 \sup_{\ell \in \nset} \| \adjfuncforward{\ell} \|_\infty  \left( \sup_{\ell \in \nset} \| \wgtfuncmod{\ell} \|_\infty + \sup_{\ell \in \nset} \| \ukestvar{\ell} \|_\infty \right) \\
+ \frac{1}{\rho^2(1 - \rho)^4 (\chi \1_{\Xset_0})^2} \sup_{\ell \in \nset} \| \addf{\ell} \|_\infty^2. 
\end{multline}
\end{corollary}

\begin{corollary} \label{cor:variance:bound:ideal:PaRIS}
Assume \hypref[assum:biased:estimate]{assum:strong:mixing}. Then for every $\K \geq 2$ and $h_n \in \bmaf{\Xfd^n}$, 
\begin{multline*}
\limsup_{n \to \infty} \frac{1}{n} \bar{\sigma}^2_n(h_n) \leq  \left( \cbound + \frac{1}{\K - 1} \dbound \right)
\sup_{\ell \in \nset} \| \addf{\ell} \|_\infty^2 \sup_{\ell \in \nset} \| \adjfuncforward{\ell} \|_\infty \sup_{\ell \in \nset} \| \wgtfuncideal{\ell} \|_\infty \\
+ \frac{1}{\rho^2(1 - \rho)^4 (\chi \1_{\Xset_0})^2} \sup_{\ell \in \nset} \| \addf{\ell} \|_\infty^2,  
\end{multline*}
where $(\bar{\sigma}^2_n)_{n \in \nset}$ is given by \eqref{eq:as:var:ideal:PaRIS}. 
\end{corollary}

Finally, we provide, for completeness, a time-uniform bound on the asymptotic variances of the random-weight particle filter corresponding to repeated forward sampling (Algorithm~\ref{alg:pm:SMC}). This bound is obtained more or less for free while establishing Theorem~\ref{thm:variance:bound} (see Section~\ref{sec:variance:bounds} for details).  

\begin{proposition}  \label{prop:variance:bound:filter}
Assume \hypref[assum:biased:estimate]{assum:bound:filter:pseudomarginal} and \hypref{assum:strong:mixing}. Then for every $\precpar \in \precparsp$ and $h \in \bmf{\Xfd_n}$,  
\begin{multline*}
\tilde{\sigma}^2_n(h) 
\leq \cboundtd  \|h \|_\infty^2 \sup_{\ell \in \nset} \| \adjfuncforward{\ell} \|_\infty  \left( \sup_{\ell \in \nset} \| \wgtfuncmod{\ell} \|_\infty + \sup_{\ell \in \nset} \| \ukestvar{\ell} \|_\infty \right) \\
+ \| h \|_\infty^2 \frac{\rho^{2n}}{\rho^4(1 - \rho)^2 (\chi \1_{\Xset_0})^2},   
\end{multline*}
where $(\tilde{\sigma}^2_n)_{n \in \nset}$ are given by \eqref{eq:as:var:pm:filter}. 
\end{proposition}

\subsection{Bounds on asymptotic bias}
\label{sec:lipschitz:continuity}

In the previous section we saw that asymptotically, as $\N$ tends to infinity, the estimator produced by $n$ iterations of Algorithm~\ref{alg:pm:PaRIS} converges to the `skew' expectation $\postmod{0:n} h_n$. In this part we will study the discrepancy between $\postmod{0:n} h_n$ and $\post{0:n} h_n$ and establish an $\mathcal{O}(n \precpar)$ bound on the same. The analysis will be performed under the assumption that the precision parameter $\precpar$ controls, uniformly in $n$, the bias of the estimators provided by \hypref{assum:biased:estimate} in the following sense.    

\begin{hypH}
\label{assum:bias:bound}
There exists a constant $c > 0$ such that for all $n \in \mathbb{N}$, $\varepsilon \in \precparsp$, $h \in \bmf{\Xfd_n \tensprod \Xfd_{n + 1}}$, and $x \in \Xset_n$,   
$$
|\ukmod{n} h (x) - \uk{n} h (x)| \leq c \varepsilon \| h \|_\infty. 
$$
\end{hypH}

\begin{example}[Durham--Gallant estimator, cont.]
We check \hypref{assum:bias:bound} for the Durham--Gallant estimator in Example~\ref{eq:durham:gallant}. 

\subsubsection*{Case~1}
Assume that the emission densities of the model are uniformly bounded, \ie, there exists $\sigma_+ \in \rsetpos$ such that $\| \md{n} \|_\infty \leq \sigma_+$ for all $n \in \nset$. In the case of the Euler scheme and under the given assumptions on equation \eqref{eq:SDE}, it can be shown that there exist $c_\delta > 0$ and $d_\delta > 0$ such that for all $\precpar \in \precparsp_\delta$ and $(x_n, x_{n + 1}) \in \Xset^2$, 
\begin{equation} \label{eq:Euler:bound}
\left| q_\delta(x_n, x_{n + 1}) - \int \bar{\kernel{Q}}^{\delta / \precpar - 1}_\precpar(x_n, \rmd x) \, \bar{q}_\precpar(x, x_{n + 1}) \right|\leq c_\delta \frac{\precpar}{\delta} \exp \left( - d_\delta \| x_{n + 1} - x_n \|^2\right);  
\end{equation}
see \cite{bally:talay:1996} (see also \cite{delmoral:jacod:2001} for an application to SMC methods). Thus, using \eqref{eq:ukmod:durham:gallant}, for all $x_n$ and $h \in \bmf{\Xfd^{\tensprod 2}}$,
\begin{align*}
\lefteqn{\left|\ukmod{n} h (x_n) - \uk{n} h (x_n) \right|}\hspace{10mm} \\ 
&\leq c_\delta \frac{\precpar}{\delta} \int h(x_n, x_{n + 1}) g(x_n, x_{n + 1}, y_{n + 1}) \exp \left( - d_\delta \| x_{n + 1} - x_n \|^2\right) \, \rmd x_{n + 1} \\
&\leq c_\delta \frac{\precpar}{\delta} \left( \frac{\pi}{d_\delta} \right)^{d_x / 2} \sigma_+ \|h \|_\infty.  
\end{align*}

\subsubsection*{Case~2}

The second case can be treated straightforwardly by combining the bound \eqref{eq:Euler:bound}, applied to the bivariate process $(X_t, Y_t)_{t > 0}$, with \eqref{eq:ukmod:durham:gallant:case:2}. This provides the existence of constants $\tilde{c}_\delta > 0$ and $\tilde{d}_\delta > 0$ such that for all $\precpar \in \precparsp_\delta$ and $(x_n, y_n, y_{n + 1}) \in \Xset \times \Yset^2$ and $h \in \bmf{\Xfd^{\tensprod 2}}$,  
\begin{align*}
\lefteqn{\left|\ukmod{n} h (x_n) - \uk{n} h (x_n) \right|}\\
&\leq \tilde{c}_\delta \frac{\precpar}{\delta} \exp \left( - \tilde{d}_\delta \| y_{n + 1} -y_n \|^2\right) \int h(x_n, x_{n + 1}) \exp \left( - \tilde{d}_\delta \| x_{n + 1} - x_n \|^2\right) \, \rmd x_{n + 1} \\
&\leq \tilde{c}_\delta \frac{\precpar}{\delta} \left( \frac{\pi}{\tilde{d}_\delta} \right)^{d_x / 2} \|h \|_\infty.  
\end{align*}

\end{example}

\begin{example}[the exact algorithm, cont.] \label{ex:exact:algorithm:cont}
In this case, the estimator is unbiased; thus, \hypref{assum:bias:bound} holds true for $\precpar = 0$. 
\end{example}

\begin{example}[ABC smoothing, cont.] In \cite{martin:jasra:singh:whiteley:delmoral:maccoy:2014}, the authors carry through their theoretical analysis under the assumption that each emission density $\md{n}$ is Lipschitz in the sense that there exists some constant $L \in \rsetpos$ such that   
\begin{equation} \label{eq:lipschitz:md}
\sup_{x_n \in \Xset_n} \left| \md{n}(x_n, y_n) - \md{n}(x_n, y'_n) \right| \leq L \| y_n - y_n' \|_1
\end{equation}
for all $y_n$ and $y'_n$ in $\rset^{d_y}$. For the purpose of illustration, assume that $\kappa_\precpar$ is a zero-mean multivariate normal distribution with covariance matrix $\precpar^2 \mathbf{I}_{d_y}$ for $\precpar > 0$. It is then easily shown that the condition \eqref{eq:lipschitz:md} implies \hypref{assum:bias:bound}; indeed, in this case, for all $x_n \in \Xset_n$ and $h \in \bmf{\Xfd_n \tensprod \Xfd_{n + 1}}$, using \eqref{eq:ukmod:ABC:smoothing}, 
$$
\left| \ukmod{n} h (x_n) - \mathbf{L}_n h (x_n) \right| \leq L  \|h \|_\infty \int \kappa_\precpar(z - y_{n + 1}) \| z - y_{n + 1} \|_1 \, \rmd z \leq \precpar d_y L \|h \|_\infty. 
$$
\end{example}

Under \hypref{assum:strong:mixing} and \hypref{assum:bias:bound} we may establish the next theorem, whose proof is postponed to Section~\ref{sec:proofs}.  

\begin{theorem} \label{thm:bias:bound}
    Assume \hypref{assum:biased:estimate} and \hypref[assum:strong:mixing]{assum:bias:bound}. Then for all $n \in \nset$, $\precpar \in \precparsp$, and $h_n \in \bmaf{\Xfd^n}$, 
    \begin{align*}
        \precpar^{-1} \big| \postmod{0:n} h_n -  \post{0:n} h_n \big| &\leq 2 c \frac{\udup}{\udlow^2} \sum_{k = 0}^{n - 1} \| \addf{k} \|_\infty \left( \sum_{m = 1}^{n - 1} \rho^{|k - m| - 1} + 1 \right) \\
        &\leq 2 c n \frac{\udup}{\udlow^2} \left( 1 + \frac{1}{\rho} + \frac{2}{1 - \rho} \right) \sup_{k \in \intvect{0}{n - 1}} \| \addf{k} \|_\infty,   
    \end{align*}
where $c$ is the constant in \hypref{assum:bias:bound}. 
\end{theorem}
In the case where $\sup_{k \in \nset} \| \addf{k} \|_\infty < \infty$, the bound provided by Theorem \ref{thm:bias:bound} is $\mathcal{O}(n)$. Moreover, by letting $\addf{k} \equiv 0$, for $k \in \intvect{0}{n - 2}$ and $\addf{n - 1}(x_{n - 1}, x_n) = h(x_n)$ for some given objective function $h \in \bmf{\Xfd_n}$, Theorem \ref{thm:bias:bound} provides, as a by-product, the following uniform error bound for the marginals (a result referred to as the \emph{filter sensitivity} in the case of parametric state-space models).  

\begin{corollary} \label{cor:filter:sensitivity}
 Assume \hypref{assum:biased:estimate} and \hypref[assum:strong:mixing]{assum:bias:bound}. Then for all $n \in \nset$, $\precpar \in \precparsp$, and $h \in \bmaf{\Xfd^n}$,
        $$
        \precpar^{-1} \big| \postmod{n} h -  \post{n} h \big| \leq 2 c \frac{\udup}{\udlow^2}  \| h \|_\infty \left( 1 + \frac{1}{\rho (1 - \rho)} \right).   
    $$
\end{corollary}

\begin{remark}
Consider now a parameterised version of the model, where the transition densities $(\ud{n ; \theta})_{k \in \nset}$ are indexed by some parameter $\theta$ belonging to some parameter space $\Theta$ being a subset of $\rset^d$. Assume further that all $(\ud{n ; \theta})_{n \in \nset}$ are differentiable with respect to $\theta$ and such that for all $n \in \nset$, $x_n \in \Xset$, and $h \in \bmf{\Xfd_{n + 1}}$,  
\begin{multline*}
\nabla_\theta \int \ud{n ; \theta}(x_n, x_{n + 1}) h(x_{n + 1}) \, \mu_{n + 1}(\rmd x_{n + 1}) \\ 
= \int \nabla_\theta \ud{n ; \theta}(x_n, x_{n + 1}) h(x_{n + 1}) \, \mu_{n + 1}(\rmd x_{n + 1})
\end{multline*}
and 
$$
\sup_{\theta \in \Theta} \int |\nabla_\theta \ud{n ; \theta}(x_n, x_{n + 1})| \, \mu_{n + 1}(\rmd x_{n + 1}) \leq c < \infty 
$$
for some positive constant $c$, implying \hypref{assum:bias:bound} in the sense that for all $n \in \nset$, $h \in \bmf{\Xfd_{n + 1}}$, and $x_n \in \Xset$, 
$$
| \uk{n; \theta}h(x_n) - \uk{n; \theta'}h(x_n) | \leq c \| \theta - \theta' \| \| h \|_\infty
$$
(\emph{i.e.}, $\precpar = \| \theta - \theta' \|$ in this case). Finally, assume also that family satisfies \hypref{assum:strong:mixing} uniformly over the parameter space in the sense that for all $n \in \nset$, $(x_n, x_{n + 1}) \in \Xset^2$, and $\theta \in \Theta$, 
$
\udlow \leq \ud{n; \theta}(x_n, x_{n + 1}) \leq \udup. 
$  
Then Theorem~\ref{thm:bias:bound} provides a positive constant $d$ such that for all $n \in \nset$,  
$$
\big| \post{0:n ; \theta} h_n -  \post{0:n ; \theta'} h_n \big|  
\leq d n \| \theta - \theta' \| \sup_{k \in \nset} \| \addf{k} \|_\infty. 
$$
This extends previous results on the uniform continuity of the filter distribution (see, \eg, \cite{papavasiliou:2006,legland:oudjane:2004}) to smoothing of additive state functionals.  
\end{remark}

\section{A numerical example}
\label{sec:numerical:results}

An exhaustive numerical analysis of the methods and results presented above is beyond the scope of the present paper, and is left as future research. Some numerical illustrations of Algorithm~\ref{alg:pm:PaRIS} in the special case of unbiased pseudo-marginalisation ($\precpar = 0$) via the exact algorithm (see Example~\ref{ex:exact:algorithm} and Example~\ref{ex:exact:algorithm:cont}) are provided in \cite{gloaguen2018online}. In this section, we focus on biased estimation ($\precpar > 0$) and illustrate the Lipschitz continuity established by Theorem \ref{thm:bias:bound} on the basis of a simple model that allows, as a comparison, a fully analytical solution to the additive smoothing problem.
 
In the following we will consider an instance of Example~\ref{eq:durham:gallant}, Case~1, where the latent diffusion $(X_t)_{t > 0}$ is an \emph{Ornstein--Uhlenbeck process} \cite{uhlenbeck1930theory} on $\rset$ parameterised by
$$
\mu(x) = - (x - \theta), \quad \sigma(x) \equiv 1,
$$
where $\theta \in \mathbb{R}$ is a parameter.
This process is assumed to be initialised according to the standard normal distribution. Furthermore, conditionally to $(X_t)_{t > 0}$, observations $(Y_n)_{n \in \nset}$ are generated as
$$
Y_n = (1 - \precpar) \varphi(X_n) + \eta_n, \quad n \in \nsetpos,
$$
where $t_n = \delta n$ for some given observation interval $\delta > 0$, $X_n = X_{t_n}$, $(\eta_n)_{n \in \nsetpos}$ are mutually independent and standard normally distributed noise variables, $\precpar \in \precparsp \eqdef [0, 1]$ is a parameter, and $\varphi$ is a bounded measurable function on $\rset$. We let $\hk_\delta$ denote the (Gaussian) Markov transition kernel of the time-discrete chain $(X_n)_{n \in \nset}$. 

In this toy example, our aim is to illustrate Theorem~\ref{thm:bias:bound} by viewing models with $\precpar > 0$ as `skew' versions of a `true' model with $\precpar = 0$. Given simulated data $(y_n)_{n \in \nsetpos}$ from the true model, smoothed additive expectations under skew models with different $\precpar > 0$ will be compared to the same expectations under the true model. Note that the model specified above satisfies condition \hypref{assum:bias:bound}; indeed, denote by $\md{}^\precpar(x_n, y_n)$ the emission density of $Y_n$ given $X_n$, which is Gaussian with mean $(1 - \precpar) \varphi(x_n)$ and unit variance, and by $\md{}(x_n, y_n)$ the same density for the true model. Then by the mean-value theorem, for all $\precpar \in \precparsp$ and $(x_n, y_n) \in \rset^2$, 
\begin{align*}
\left \vert \md{}^\precpar(x_n, y_n) - \md{}(x_n, y_n) \right \vert &\leq \precpar \sup_{\precpar \in \precparsp} \left \vert\frac{\partial}{\partial \precpar} \md{}^\precpar(x_n, y_n) \right\vert \\
&\leq \varepsilon \left( \vert \varphi(x_n) y_n - \varphi^2(x_n) \vert + \varphi^2(x_n) \right), 
\end{align*}
implying immediately that for all bounded measurable real-valued functions $h$,  
\begin{align*}
\left \vert \ukmod{n} h(x_n) - \uk{n} h(x_n) \right \vert \leq \precpar c(y_{n + 1}) \| h\|_\infty, 
\end{align*}
where $\ukmod{n}(x_n, \rmd x_{n + 1}) \eqdef \hk_\delta(x_n, \rmd x_{n + 1}) \md{}^\precpar(x_{n + 1}, y_{n + 1})$ and $c(y_{n + 1}) \eqdef \| \varphi \|_\infty |y_{n + 1}| + 2 \| \varphi \|_\infty^2$. Thus, \hypref{assum:bias:bound} holds true under the mild assumption that $\sup_{n \in \nsetpos} |y_n| < \infty$. 

In this context, we conducted numerical experiments with $\theta = 5$ and $\varphi(x) = \min(\max(x, -10^5), 10^5)$. With this parametrisation, the model is, in practice, linear and Gaussian; thus, for linear additive state functionals in the form $h_n(x_{0:n}) = \sum_{k = 0}^n x_k$, a very good approximation of the exact solution $(\postmod{0:n} h_n)_{n \in \nset}$ to the optimal smoothing problem can, for the skewed models $(\precpar > 0)$ as well as the true model $(\precpar = 0)$, be obtained using Kalman recursions \cite{rauch1965maximum}. Figure~\ref{fig:varying:precpar} displays the discrepancy between $\postmod{0:n} h_n$ and $\post{0:n} h_n$ for varying $\precpar \in \{0, 0.05, 0.1, \ldots, 0.5 \}$ and a fixed $n = 50$ (the red line). Clearly, the bias increases---in perfect agreement with Theorem~\ref{thm:bias:bound}---at a rate that is at most linear in $\precpar$. In addition, Figure~\ref{fig:varying:precpar} shows similar biases (turquoise markers) obtained using the ideal PaRIS in Algorithm~\ref{alg:ideal:PaRIS} with $\N = 200$ particles and $\K = 2$ backward samples. In this algorithm, the particles were propagated using the \textit{optimal importance function} \cite{doucet2000sequential}, which can be computed explicitly in the linear Gaussian case. Here the PaRIS was re-run $60$ times for each value of $\precpar$. Note that the variance of the PaRIS estimates increases somewhat with $\precpar$, reflecting the fact that the proposal becomes less and less compatible with the data as the model gets increasingly skew. 

In order to illustrate further the $\mathcal{O}(n \precpar)$ bound provided by Theorem~\ref{thm:bias:bound} as well as the stochastic stability of Algorithm~\ref{alg:ideal:PaRIS} established in Corollary~\ref{cor:variance:bound:ideal:PaRIS}, Kalman smoothing was conducted for $n \in \intvect{1}{50}$ on the basis of a skew model with fixed $\varepsilon = 0.1$. Figure~\ref{fig:varying:n} shows, as expected, a linear increase of the bias with $n$ (red line). In addition, displaying also the errors (green lines) of $60$ independent ideal PaRIS replicates (obtained under the skew model with the same algorithmic parametrisation as previously), the same plot confirms the stochastic stability of the ideal PaRIS algorithm, as the variance does not grow faster than linearly with $n$. Finally, for completeness Figure~\ref{fig:varying:n} reports similar errors (blue lines) when the ideal PaRIS evolves under the dynamics of the true model, and in this case the bias is negligible as expected (whereas the variance is still increasing linearly). 

\begin{figure}
    \centering
    \begin{subfigure}[b]{0.48\textwidth}
        \includegraphics[width=\textwidth]{
        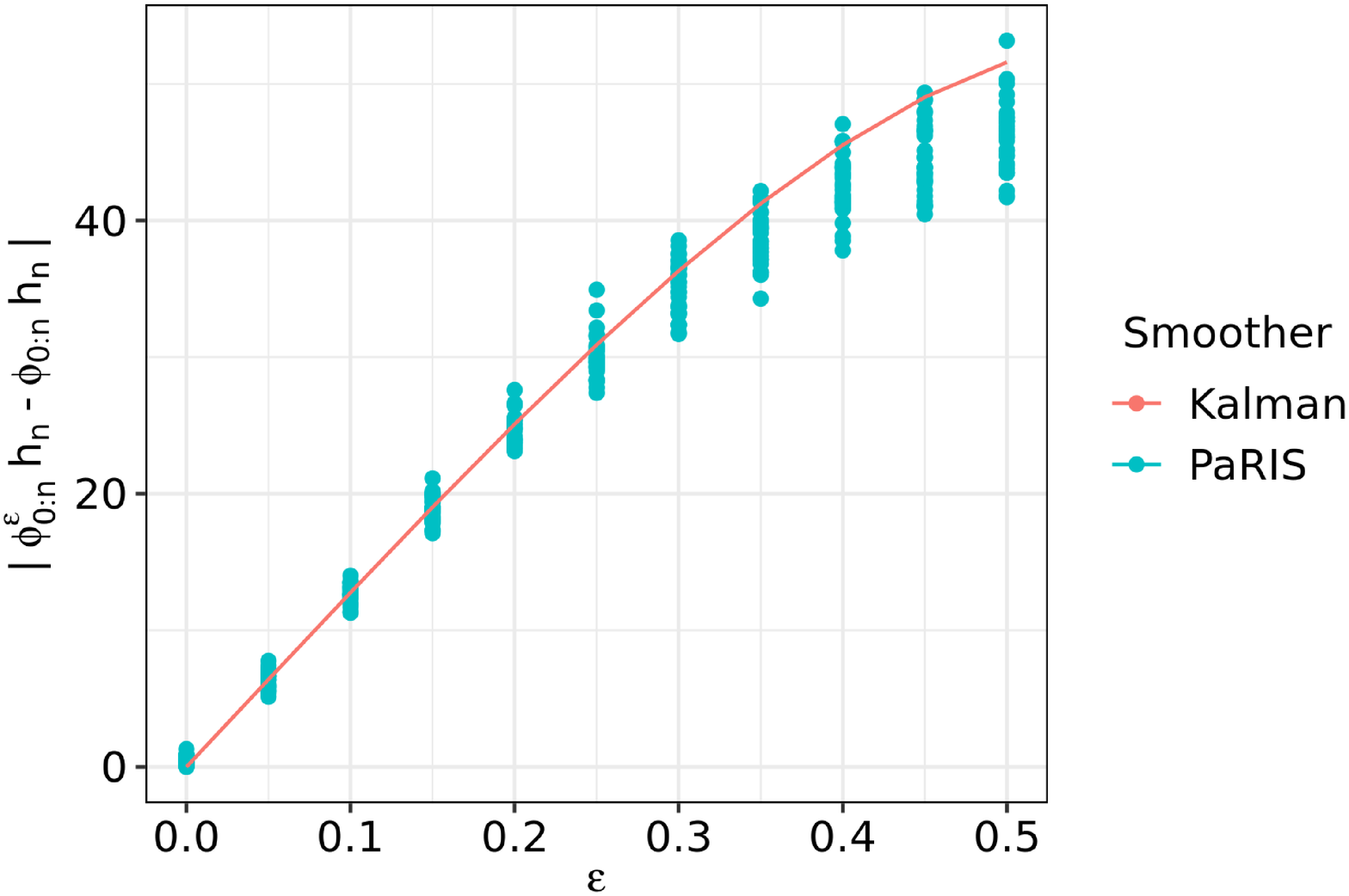}
        \caption{}
        \label{fig:varying:precpar}
    \end{subfigure}
    ~ 
    \begin{subfigure}[b]{0.48\textwidth}
        \includegraphics[width=\textwidth]{
        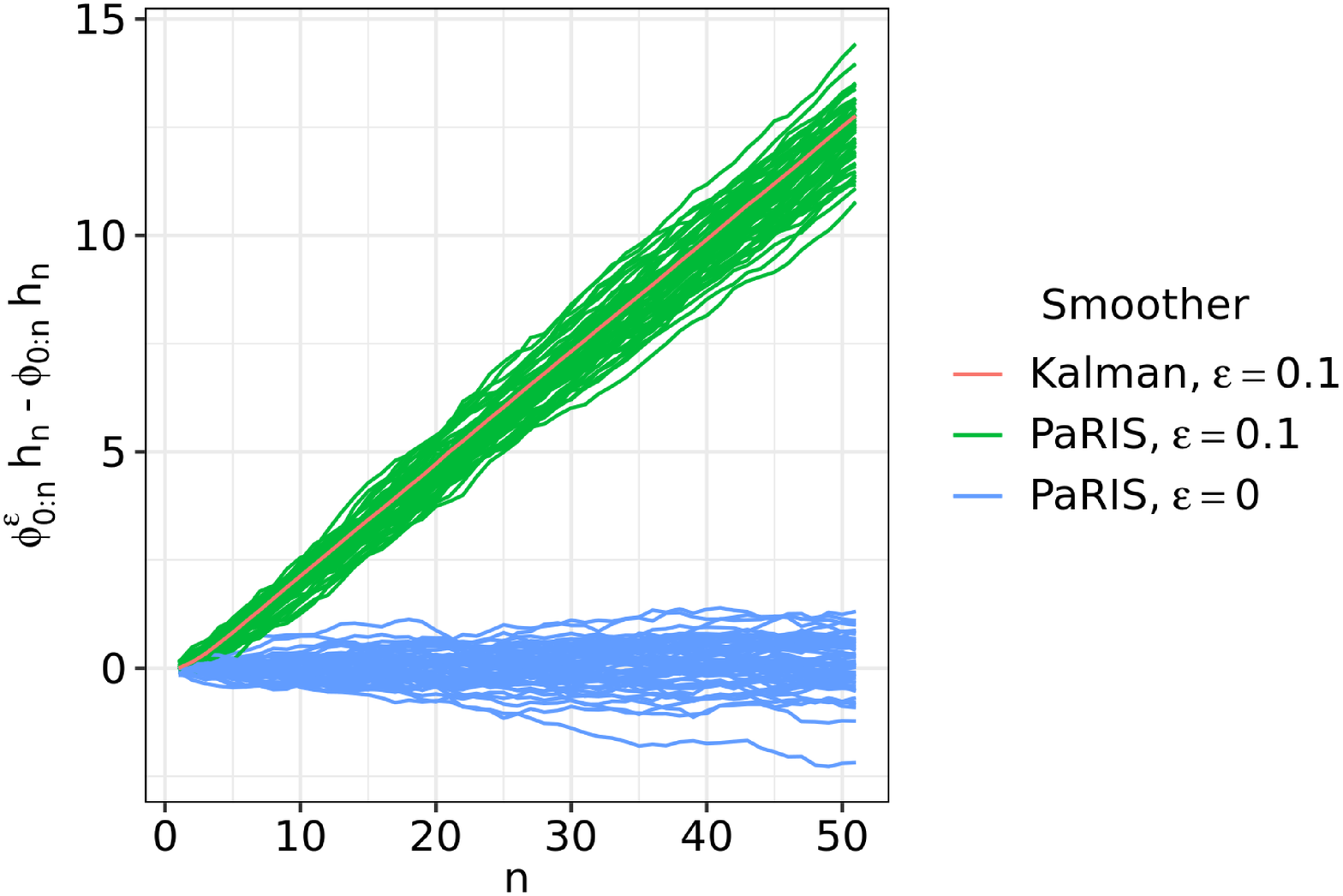}
        \caption{}
        \label{fig:varying:n}
    \end{subfigure}
    \caption{(a) The red line is the deviation of $\postmod{0:n} h_n$ from $\post{0:n} h_n$ for varying $\precpar \in \{0, 0.05, 0.1, \ldots, 0.5 \}$ and fixed $n = 50$, computed by means of Kalman smoothing. Turquoise markers are similar errors obtained on the basis of $60$ replicates of the ideal PaRIS in Algorithm~\ref{alg:ideal:PaRIS} with $\N = 200$ particles and $\K = 2$ backward samples. (b) The red line is the same deviation for increasing $n \in \intvect{1}{50}$ and fixed $\precpar = 0.1$. Then green and blue lines correspond to errors of $60$ independent ideal PaRIS replicates obtained under the skew and true models, respectively, for $\N = 200$ and $\K = 2$.} \label{fig}
\end{figure}




\section*{Acknowledgements}

The research of J.~Olsson is supported by the Swedish Research Council, Grant~2018-05230.  

\bibliographystyle{plain}
\bibliography{glo-2020}


\appendix


\section{Proof of Lemma~\ref{lem:reversibility}}
\label{sec:proof:lem:reversibility}

To show (i), write, using the definition \eqref{eq:def:backward:kernel}, the right-hand side as 
\begin{align*}
\lefteqn{\iint \post{n} \uk{n}(\rmd x_{n + 1}) \, \bkw{n}(x_{n + 1}, \rmd x_n) \, h(x_n, x_{n + 1})} \hspace{10mm} \\
&= \iint \post{n}[\ud{n}(\cdot, x_{n + 1})] \, \mu_{n + 1}(\rmd x_{n + 1}) \frac{\post{n}(\rmd x_n) [\ud{n}(\cdot, x_{n + 1})]}{\post{n}[\ud{n}(\cdot, x_{n + 1})]} \, h(x_n, x_{n + 1}) \\
&= \iint \post{n}(\rmd x_n) \, \uk{n}(x_n, \rmd x_{n + 1}) \, h(x_n, x_{n + 1}), 
\end{align*}
which was to be established. 

The statement (ii) is shown by induction. Thus, assume that the claim holds true for $n$, pick arbitrarily $h \in \bmf{\Xfd^{n + 1}}$, and write, using the induction hypothesis, 
\begin{align*}
\post{0:n} \uk{0, n} h &= \iint \post{n}(\rmd x_n) \, \tstat{n}(x_n, \rmd x_{0:n - 1}) \, \uk{n}(x_n, \rmd x_{n + 1}) \, h(x_{0:n + 1}) \\
&= \iint \post{n}(\rmd x_n) \, \uk{n}(x_n, \rmd x_{n + 1}) \, \bar{h}(x_n, x_{n + 1}), 
\end{align*}
where 
$$
\bar{h}(x_n, x_{n + 1}) \eqdef \int \tstat{n}(x_n, \rmd x_{0:n - 1}) h(x_{0:n + 1}). 
$$
Now, since $\bkw{n} \bar{h}(x_{n + 1}) = \tstat{n + 1} h(x_{n + 1})$, (i) yields
\begin{align*} 
\post{0:n} \uk{0, n} h 
&= \iint \post{n} \uk{n}(\rmd x_{n + 1}) \, \bkw{n}(x_{n + 1}, \rmd x_n) \, \bar{h}(x_n, x_{n + 1}) \\
&= \int \post{n} \uk{n}(\rmd x_{n + 1}) \, \tstat{n + 1} h(x_{n + 1}). 
\end{align*}
Thus, we obtain, by applying \eqref{eq:recursion:FK:path} and \eqref{eq:recursion:FK:marg},  
$$
\post{0:n + 1} h = \frac{\post{0:n} \uk{0, n} h}{\post{0:n} \uk{0, n} \1_{\Xset^n}} 
= \int \frac{\post{n} \uk{n}(\rmd x_{n + 1})}{\post{n} \uk{n} \1_{\Xset_{n + 1}}} \tstat{n + 1} h(x_{n + 1}) = \post{n + 1} \tstat{n + 1} h. 
$$
Finally, we note that the base case $n = 1$ follows straightforwardly by combining (i) and \eqref{eq:recursion:FK:marg}.



\section{Proof of Theorem~\ref{cor:hoeffding:tau:marginal}}
\label{sec:proof:prop:hoeffding:tau:marginal}

As explained in Section~\ref{sec:convergence:pm:PaRIS}, 
Theorem~\ref{cor:hoeffding:tau:marginal} follows as an immediate corollary of the following more general result, to whose proof we devote this section. 

\begin{proposition}
\label{prop:hoeffding:tau:marginal}
Assume \hypref[assum:biased:estimate]{assum:bound:filter:pseudomarginal}. Then for all $n \in \nset$, $h_n \in \bmaf{\Xfd^n}$, $(f_n, \ftd{n}) \in \bmf{\Xfd_n}^2$, and $\K \in \nsetpos$ there exist $(c_n, d_n) \in \rsetpos^2$ such that for all $\N \in \nsetpos$ and $\epsilon > 0$,

\begin{align*}
&\mbox{(i)} \quad \pP\left(\left| \frac{1}{\N} \sum_{i = 1}^\N \ewght{n}{i} \{\tstat[i]{n} f_n(\epart{n}{i}) + \ftd{n}(\epart{n}{i})\} - \frac{\postmod{n - 1} \ukmod{n - 1} ( \tstatmod{n} h_n f_n + \tilde{f}_n )}{\postmod{n - 1} \adjfuncforward{n - 1}} \right| \geq \epsilon \right) \\ 
&\hspace{90mm} \leq c_n \exp \left( - d_n \N \epsilon^2 \right), \\
&\mbox{(ii)} \quad \pP\left(\left| \sum_{i = 1}^\N \frac{\ewght{n}{i}}{\sumwght{n}} \{\tstat[i]{n} f_n(\epart{n}{i}) + \tilde{f}_n(\epart{n}{i})\} - \postmod{n} ( \tstatmod{n} h_n f_n + \tilde{f}_n ) \right| \geq \epsilon \right)\\ 
&\hspace{90mm} \leq c_n \exp \left( - d_n \N \epsilon^2 \right).
\end{align*}
\end{proposition}

Note that letting $f_n \equiv 1$ and $\ftd{n} \equiv 0$ in (ii) yields immediately Theorem~\ref{cor:hoeffding:tau:marginal}. 

We preface the proof of Proposition~\ref{prop:hoeffding:tau:marginal} by a technical lemma, generalising \cite[Lemma~12]{olsson:westerborn:2014b} and \cite[Lemma~2]{gloaguen2018online}, which will be instrumental in the following. For all $n \in \nset$, define the following $\sigma$-fields:
\begin{align*}
\calF{n}{\N} &\eqdef \sigma \{ (\epart{m}{i}, \ewght{m}{i}, \tstat[i]{m}) : i \in \intvect{1}{\N}, m \in \intvect{0}{n} \}, \\
\calG{n + 1}{\N} &\eqdef \sigma \{(\epart{n + 1}{i}, \ind{n + 1}{i}) : i \in \intvect{1}{\N} \} \vee \calF{n}{\N}. 
\end{align*}

\begin{lemma}
\label{lem:expectation:incremental:marginal}
Assume \hypref[assum:biased:estimate]{assum:bound:filter:pseudomarginal}. For all $n \in \nset$, $(f_{n + 1}, \ftd{n + 1}) \in \bmf{\Xfd_{n + 1}}^2$, and $(\N, \K) \in \nsetpos^2$, the random variables $\ewght{n + 1}{i}\{\tstat[i]{n + 1} f_{n + 1}(\epart{n + 1}{i}) + \ftd{n + 1}(\epart{n + 1}{i})\}$, $i \in \intvect{1}{\N}$, are independent and identically distributed (i.i.d.) conditionally to $\calF{n}{\N}$. In addition, for all $i$,  
\begin{multline*}
\pE \left[\ewght{n + 1}{i}\{\tstat[i]{n + 1}f_{n + 1}(\epart{n + 1}{i}) + \ftd{n + 1}(\epart{n + 1}{i})\} \cond  \calF{n}{\N} \right] \\
 = (\post[\N]{n} \adjfuncforward{n})^{-1} 
 \sum_{j = 1}^\N \frac{\ewght{n}{j}}{\sumwght{n}}
\{ \tstat[j]{n} \ukmod{n} f_{n + 1}(\epart{n}{j}) + \ukmod{n}(\addf{n} f_{n + 1} + \ftd{n + 1})(\epart{n}{j})\}.
\end{multline*}
\end{lemma}


\begin{proof}
Independence and equal distribution is immediate by construction of Algorithm~\ref{alg:pm:PaRIS}. Moreover, note that
\begin{equation} \label{eq:expected:weight:given:G}
\pE \left[\ewght{n + 1}{1} \cond \calG{n + 1}{\N} \right] = \int \wgtfunc{n}(\epart{n}{\ind{n + 1}{1}}, \epart{n + 1}{1}, z) \, \ukdist{n}(\epart{n}{\ind{n + 1}{1}}, \epart{n + 1}{1}, \rmd z) = \wgtfuncmod{n}(\epart{n}{\ind{n + 1}{1}}, \epart{n + 1}{1}), 
\end{equation}
where $\wgtfunc{n}$ and $\wgtfuncmod{n}$ are defined in \hypref{assum:bound:filter:pseudomarginal}(ii), and 
$$
\pE \left[\tstat[1]{n + 1} \cond \calG{n + 1}{\N} \right] = \sum_{i = 1}^\N \frac{\ewght{n}{i} \udmod{n} (\epart{n}{i}, \epart{n + 1}{1})}{\sum_{i' = 1}^{\N} \ewght{n}{i'} \udmod{n}(\epart{n}{i'}, \epart{n + 1}{1})} \left(\tstat[i]{n} + \addf{n}(\epart{n}{i}, \epart{n + 1}{1}) \right). 
$$
Then, since $\tstat[1]{n + 1}$ and  $\ewghthat{n + 1}{1}$ are conditionally independent given $\calG{n + 1}{\N}$,
\begin{align*}
\lefteqn{\pE \left[\ewghthat{n + 1}{1} \{ \tstat[1]{n + 1} f_{n + 1}(\epart{n + 1}{1}) + \ftd{n + 1}(\epart{n + 1}{1}) \} \cond \calF{n}{\N}\right]} 
\nonumber \\
&= \pE \left[\pE \left[ \ewghthat{n + 1}{1} \cond \calG{n + 1}{\N} \right] \pE \left[\tstat[1]{n + 1} \cond \calG{n + 1}{N}\right] f_{n + 1}(\epart{n + 1}{1}) \cond \calF{n}{\N} \right] \nonumber \\
&\hspace{4cm} + \pE \left[ \pE \left[ \ewghthat{n + 1}{1} \cond \calG{n + 1}{\N} \right] \ftd{n + 1}(\epart{n + 1}{1}) \cond \calF{n}{N} \right] \nonumber \\
&= \pE \left[ \wgtfuncmod{n}(\epart{n}{\ind{n + 1}{1}}, \epart{n + 1}{1}) \sum_{i = 1}^\N \frac{\ewght{n}{i} \udmod{n} (\epart{n}{i}, \epart{n + 1}{1})}{\sum_{i' = 1}^{\N} \ewght{n}{i'} \udmod{n} (\epart{n}{i'}, \epart{n + 1}{1})} \left(\tstat[i]{n} + \addf{n}(\epart{n}{i}, \epart{n + 1}{1}) \right) f_{n + 1}(\epart{n + 1}{1}) \cond \calF{n}{\N} \right] \nonumber \\
&\hspace{4cm} + \pE \left[\wgtfuncmod{n}(\epart{n}{\ind{n + 1}{1}}, \epart{n + 1}{1}) \ftd{n + 1}(\epart{n + 1}{1}) \cond \calF{n}{\N} \right]. \label{eq:exp:decomposition}
\end{align*}
Now, we may complete the proof by noting that, by \eqref{eq:cond:instrumental:mixture},  
\begin{align*} 
\lefteqn{\pE \left[ \wgtfuncmod{n}(\epart{n}{\ind{n + 1}{1}}, \epart{n + 1}{1}) \sum_{i = 1}^\N \frac{\ewghthat{n}{i} \udmod{n} (\epart{n}{i}, \epart{n + 1}{1})}{\sum_{i' = 1}^{\N} \ewghthat{n}{i'} \udmod{n} (\epart{n}{i'}, \epart{n + 1}{1})} \left(\tstat[i]{n} + \addf{n}(\epart{n}{i}, \epart{n + 1}{1}) \right) f_{n + 1}(\epart{n + 1}{1}) \cond \calF{n}{\N} \right]} \\
&= (\post[\N]{n} \adjfuncforward{n})^{-1} \sum_{j = 1}^\N \frac{\ewght{n}{j}}{\sumwght{n}} \int \udmod{n}(\epart{n}{j}, x) \mu(\rmd x) \sum_{i = 1}^\N \frac{\ewghthat{n}{i} \udmod{n} (\epart{n}{i}, x)}{\sum_{i' = 1}^\N \ewghthat{n}{i'} \udmod{n} (\epart{n}{i'}, x)} \left(\tstat[i]{n} + \addf{n}(\epart{n}{i}, x) \right) f_{n + 1}(x) \\
&= (\post[\N]{n} \adjfuncforward{n})^{-1} \sum_{i = 1}^\N \frac{\ewght{n}{i}}{\sumwght{n}} \{ \tstat[i]{n} \ukmod{n} f_{n + 1}(\epart{n}{i}) + \ukmod{n}(\addf{n} f_{n + 1})(\epart{n}{i}) \}
\end{align*}
and, similarly, 
\begin{equation*}
 \pE \left[\wgtfuncmod{n}(\epart{n}{\ind{n + 1}{1}}, \epart{n + 1}{1}) \ftd{n + 1}(\epart{n + 1}{1}) \cond \calF{n}{\N} \right] = (\post[\N]{n} \adjfuncforward{n})^{-1} \sum_{i = 1}^\N \frac{\ewght{n}{i}}{\sumwght{n}} \ukmod{n} \ftd{n + 1}(\epart{n}{i}). 
\end{equation*}
\end{proof}

\begin{proof}[Proof of Proposition~\ref{prop:hoeffding:tau:marginal}]
We establish (i) by induction over $n$. Write, using Lemma~\ref{lem:expectation:incremental:marginal},
\begin{multline} \label{eq:hoeffding:decomp}
\frac{1}{\N} \sum_{i = 1}^\N \ewght{n + 1}{i} \{\tstat[i]{n + 1} f_{n + 1}(\epart{n + 1}{i}) + \ftd{n + 1}(\epart{n + 1}{i})\} - \frac{\postmod{n} \ukmod{n} \{ \tstatmod{n + 1} h_{n + 1} f_{n + 1} + \ftd{n + 1} \}}{\postmod{n} \adjfuncforward{n}} \\
= \frac{1}{\N} \sum_{i = 1}^\N \ewght{n + 1}{i} \{\tstat[i]{n + 1} f_{n + 1}(\epart{n + 1}{i}) + \ftd{n + 1}(\epart{n + 1}{i})\} \hspace{40mm} \\
- \pE \left[ \ewght{n + 1}{1} \{\tstat[1]{n + 1} f_{n + 1}(\epart{n + 1}{1}) + \ftd{n + 1}(\epart{n + 1}{1})\} \cond \calF{n}{\N} \right] \\
+ \sum_{j = 1}^\N 
\frac{\ewght{n}{j} \adjfuncforward{n}(\epart{n}{j})}{\sum_{j' = 1}^\N \ewght{n}{j'} \adjfuncforward{n}(\epart{n}{j'})} \{ \tstat[j]{n} \prop{n} (\wgtfuncmod{n} f_{n + 1})(\epart{n}{j}) + \prop{n}(\wgtfuncmod{n} \addf{n} f_{n + 1} + \wgtfuncmod{n} \ftd{n + 1})(\epart{n}{j})\} \\
- \frac{\postmod{n - 1} \ukmod{n - 1}\{\tstatmod{n} \ukmod{n} f_{n + 1} + \ukmod{n} (\addf{n} f_{n + 1} + \ftd{n + 1})\}}{\postmod{n - 1} \ukmod{n - 1} \adjfuncforward{n}}. 
\end{multline}
Since for all $i$, 
$$
|\ewght{n + 1}{i} \{\tstat[i]{n + 1} f_{n + 1}(\epart{n + 1}{i}) + \ftd{n + 1}(\epart{n + 1}{i})\}| \leq \| \wgtfunc{n} \|_\infty (\| h_{n + 1} \|_\infty \| f_{n + 1} \|_\infty + \| \ftd{n + 1} \|_\infty), 
$$
the Hoeffding inequality for conditional expectations provides constants $(d, \tilde{d}) \in \rsetpos^2$ such that for all $\epsilon > 0$, 
\begin{multline*}
\pP \left( \left| \frac{1}{\N} \sum_{i = 1}^\N \ewght{n + 1}{i} \{\tstat[i]{n + 1} f_{n + 1}(\epart{n + 1}{i}) + \ftd{n + 1}(\epart{n + 1}{i})\} \right. \right. \\
\left. \left.  - \pE \left[ \ewght{n + 1}{1} \{\tstat[1]{n + 1} f_{n + 1}(\epart{n + 1}{1}) + \ftd{n + 1}(\epart{n + 1}{1})\} \cond \calF{n}{\N} \right] \vphantom{\sum_{i = 1}^\N} \right| \geq \epsilon \right) \leq d \exp \left( - \tilde{d} N \epsilon^2 \right).  
\end{multline*}
In addition, since $\adjfuncforward{n}$, $\adjfuncforward{n} \prop{n} (\wgtfuncmod{n} f_{n + 1})$, and $\adjfuncforward{n} \prop{n}(\wgtfuncmod{n} \addf{n} f_{n + 1} + \wgtfuncmod{n} \ftd{n + 1})$ all belong to $\bmf{\Xfd_n}$, there exist, by the induction hypothesis and \cite[Lemma~4]{douc:garivier:moulines:olsson:2010}, $(d', \tilde{d}') \in \rsetpos^2$ such that for all $\epsilon > 0$, 
\begin{multline*}
\pP \left( \left|  \sum_{j = 1}^\N 
\frac{\ewght{n}{j} \adjfuncforward{n}(\epart{n}{j})}{\sum_{j' = 1}^\N \ewght{n}{j'} \adjfuncforward{n}(\epart{n}{j'})} \{ \tstat[j]{n} \prop{n} (\wgtfuncmod{n} f_{n + 1})(\epart{n}{j}) + \prop{n}(\wgtfuncmod{n} \addf{n} f_{n + 1} + \wgtfuncmod{n} \ftd{n + 1})(\epart{n}{j})\} \right. \right. \\
\left. \left. - \frac{\postmod{n - 1} \ukmod{n - 1}\{\tstatmod{n} \ukmod{n} f_{n + 1} + \ukmod{n} (\addf{n} f_{n + 1} + \ftd{n + 1})\}}{\postmod{n - 1} \ukmod{n - 1} \adjfuncforward{n}} \right| \geq \epsilon \right) \leq d' \exp \left( - \tilde{d}' N \epsilon^2 \right).  
\end{multline*}
Combining the previous two inequalities completes the proof of (i) at time step $n + 1$. The base case $n = 1$ is established using again the decomposition \eqref{eq:hoeffding:decomp}, the standard Hoeffding inequality for independent and identically distributed variables, \cite[Lemma~4]{douc:garivier:moulines:olsson:2010}, and the fact that $h_0 \equiv 0$.   

Final, (ii) follows immediately from (i) and \cite[Lemma~4]{douc:garivier:moulines:olsson:2010}.
\end{proof} 


\section{Proof of Theorem~\ref{cor:clt:pseudo:marginal:paris}}
\label{sec:proof:prop:clt:pseudo:marginal:paris}

We now turn our focus to the proof of Theorem~\ref{cor:clt:pseudo:marginal:paris}, which, just like Theorem~\ref{cor:hoeffding:tau:marginal}, will be established via the following more general result. 

\begin{theorem}
\label{prop:clt:pseudo:marginal:paris}
Assume~\hypref[assum:biased:estimate]{assum:bound:filter:pseudomarginal}. Then for all $n \in \nset$, $\K \in \nsetpos$, and $(f_n, \ftd{n}) \in \bmf{\Xfd_n}^2$,
$$
 \sqrt{\N} \left( \sum_{i = 1}^\N \frac{\ewght{n}{i}}{\sumwght{n}} \{\tstat[i]{n} f_n(\epart{n}{i}) + \ftd{n}(\epart{n}{i}) \} - \postmod{n}(\tstatmod{n} h_n f_n + \ftd{n}) \right) 
  \dlim \sigma_n(h_n, f_n, \ftd{n}) Z,
$$
where $Z$ is standard normally distributed and 
\begin{multline} \label{eq:non-recursive:as:var:affine}
\sigma^2_n (h_n, f_n, \ftd{n}) \eqdef \frac{\chi \{ \initwgtfunc
\retrokmodmodnorm_{0, n} (h_n f_n + \ftd{n}) \}^2}{(\chi \ukmod{0, n - 1} \1_{\Xset^n})^2} \\
+ \sigma^2_n \langle (\wgtfuncmod{\ell})_{\ell = 0}^{n - 1} \rangle (h_n, f_n, \ftd{n}) 
+ \sigma^2_n \langle (\ukestvar{\ell})_{\ell = 0}^{n - 1} \rangle (h_n, f_n, \ftd{n}) 
\end{multline}
and
\begin{multline} \label{eq:partial:variance}
\sigma^2_n \langle (\varphi_\ell)_{\ell = 0}^{n - 1} \rangle (h_n, f_n, \ftd{n}) 
\eqdef \sum_{m = 0}^{n - 1} \frac{\postmod{m} \adjfuncforward{m} \postmod{m} \ukmod{m} \{ \varphi_m [\retrokmodmodnorm_{m + 1, n} (h_n f_n + \ftd{n})]^2
\}}{(\postmod{m} \ukmod{m, n - 1} \1_{\Xset_n})^2} \\
+ \sum_{m = 0}^{n - 1} \sum_{\ell = 0}^m \frac{\postmod{m} \adjfuncforward{m} \postmod{\ell} \ukmod{\ell} \{\bkmod{\ell}(\tstatmod{\ell} h_{\ell} + \addf{\ell} - \tstatmod{\ell + 1} h_{\ell +1})^2 \ukmod{\ell + 1, m}( \bkmod{m} \varphi_m [\ukmod{m + 1, n - 1} f_n]^2
)\}}{\K^{m - \ell + 1} (\postmod{\ell} \ukmod{\ell, m - 1} \1_{\Xset_m})(\postmod{m} \ukmod{m, n - 1} \1_{\Xset_n})^2}
\end{multline}
for any sequence $(\varphi_\ell)_{\ell \in \nset}$ of measurable functions $\varphi_n : \Xset_n \times \Xset_{n + 1} \to \rsetnn$. 
\end{theorem}

As previously, applying Theorem~\ref{prop:clt:pseudo:marginal:paris} with $f_n \equiv 1$ and $\ftd{n} \equiv 0$ yields immediately Theorem~\ref{cor:hoeffding:tau:marginal}. Before proving Theorem~\ref{prop:clt:pseudo:marginal:paris} we establish some preparatory lemmas, where the first is easily obtained by combining Lemma~\ref{lem:reversibility}(i) (applied to the skew modeled formed by $\chi$ and $(\ukmod{n})_{n \in \nset}$) and \eqref{eq:forward:smoothing}. 

\begin{lemma} \label{eq:critical:identity}
Assume \hypref{assum:biased:estimate}. Then for all $n \in \nset$ and $(f_{n + 1}, \ftd{n + 1}) \in \bmf{\Xfd_{n + 1}}^2$, 
$$
\postmod{n} \ukmod{n}(\tstatmod{n + 1} h_{n + 1} f_{n + 1} + \ftd{n + 1}) = \postmod{n}\{\tstatmod{n} \ukmod{n} f_{n + 1} + \ukmod{n} (\addf{n} f_{n + 1} + \ftd{n + 1})\}. 
$$
\end{lemma}

\begin{lemma}
\label{lem:hoeffding:tau:square}
Assume \hypref[assum:biased:estimate]{assum:bound:filter:pseudomarginal}. Then for all $n \in \nset$, $f \in \bmf{\Xfd_n}$, and $\K \in \nsetpos$, 
\[
\sum_{i = 1}^\N \frac{\ewght{n}{i}}{\sumwght{n}}(\tstat[i]{n})^2 f(\epart{n}{i}) \pplim \postmod{n}([\tstatmod{n} h_n]^2 f) + \eta_n f, 
\]
where the measures $(\eta_m)_{m \in \nsetpos}$ are defined recursively as 
\[
\eta_{m + 1} f = \K^{-1} \frac{\eta_m \ukmod{m} f +  \postmod{m} \ukmod{m} \{ \bkmod{m}(\tstatmod{m} h_m + \addf{m} - \tstatmod{m + 1} h_{m + 1})^2 f \}}{\postmod{m} \ukmod{m} \1_{\Xset_{m + 1}}}, \quad m \in \nset, 
\]
with $\eta_0 \equiv 0$. 
\end{lemma}

Following \cite[Lemma~13]{olsson:westerborn:2017}, the measures $(\eta_m)_{m \in \nsetpos}$ may be expressed non-recursively as
\begin{equation}
\label{eq:def:etak}
\eta_n f = \sum_{m = 0}^{n - 1} \K^{m - n} \frac{\postmod{m} \ukmod{m} \{ \bkmod{m}(\tstatmod{m} h_m + \addf{m} - \tstatmod{m + 1} h_{m + 1})^2 \ukmod{m + 1, n - 1} f \}}{\postmod{m} \ukmod{m, n - 1} \1_{\Xset_n}}.
\end{equation}

\begin{proof}[Proof of Lemma~\ref{lem:hoeffding:tau:square}]
We proceed by induction over $n$. The base case $n = 0$ is a trivial consequence of the fact that $\tstat{0} h_0 = 0$ and $\tstat[i]{0} = 0$ for all $i \in \intvect{1}{\N}$. Thus, we assume that the result holds true for some $n \in \nset$ and write
\[
\sum_{i=1}^\N \frac{\ewght{n+1}{i}}{\sumwght{n+1}}(\tstat[i]{n+1})^2 f(\epart{n+1}{i}) = \frac{a_\N}{b_\N},
\]
where 
$$
a_\N \eqdef \frac{1}{\N}\sum_{i=1}^N\ewght{n+1}{i}(\tstat[i]{n+1})^2 f(\epart{n+1}{i}), \quad b_\N \eqdef \frac{1}{\N} \sum_{i=1}^N \ewght{n+1}{i}
$$

We first establish the convergence in probability of $(a_\N)_{\N \in \nsetpos}$. Using again Hoeffding's inequality for conditional expectations and the fact that the variables $\ewght{n+1}{i}(\tstat[i]{n + 1})^2 f_{n + 1}(\epart{n + 1}{i})$, $i \in \intvect{1}{\N}$, are 
bounded by $\| \wgtfunc{n} \|_\infty \| h_{n + 1} \|^2_\infty \| f_{n + 1}\|_\infty$ and conditionally i.i.d. given $\calF{n}{\N}$ for all $i$, we obtain for all $\epsilon > 0$,  
\[
\pP\left(\left| a_\N - \pE[a_\N \mid \calF{n}{\N}] \right| \geq \epsilon \right) \leq \exp \left( - \frac{2 \N \epsilon^2}{\| \wgtfunc{n} \|_\infty \| h_{n + 1} \|^2_\infty \| f_{n+1} \|_\infty} \right). 
\]
It is hence enough to consider the limit in probability of $\pE[a_\N \mid \calF{n}{\N}]$ as $\N$ tends to infinity. For this purpose, write, using \eqref{eq:expected:weight:given:G}, 
\begin{align*}
\pE[a_\N \mid \calF{n}{\N}] &= \pE \left[ \ewght{n + 1}{1}(\tstat[1]{n + 1})^2 f_{n+1}(\epart{n + 1}{1}) \mid  \calF{n}{\N} \right] \\ 
&= \pE\left[\wgtfuncmod{n}(\epart{n}{\ind{n + 1}{1}}, \epart{n + 1}{1}) \pE \left[(\tstat[1]{n + 1})^2 \cond \calG{n + 1}{\N}\right] f_{n + 1}(\epart{n + 1}{1}) \cond \calF{n}{\N} \right] \\
&= a_\N' + a_\N'',
\end{align*}
where
\begin{align*}
a_\N' &\eqdef \K^{-1} \\ 
&\times \pE\left[ \wgtfuncmod{n}(\epart{n}{\ind{n + 1}{1}}, \epart{n + 1}{1}) f_{n + 1}(\epart{n + 1}{1}) \pE \left[ \left( \tstat[J_{n + 1}^{(1,1)}]{n} + \addf{n}(\epart{n}{J_{n + 1}^{(1,1)}}, \epart{n + 1}{1}) \right)^2 \cond \calG{n + 1}{\N}\right] \cond \calF{n}{\N} \right], \\
a_\N'' &\eqdef (\K - 1)\K^{-1} \\ 
&\times \pE \left[ \wgtfuncmod{n}(\epart{n}{\ind{n + 1}{1}}, \epart{n + 1}{1}) f_{n + 1}(\epart{n + 1}{1}) \pE \left[\tstat[J_{n + 1}^{(1,1)}]{n} + \addf{n}(\epart{n}{J_{n + 1}^{(1,1)}}, \epart{n + 1}{1}) \cond \calG{n + 1}{\N}\right]^2 \cond \calF{n}{\N} \right]. 
\end{align*}
Here the first term is given by
\begin{align}
a_\N' &= \K^{-1} \sum_{i = 1}^\N \frac{\ewght{n}{i}\adjfuncforward{n}(\epart{n}{i})}{\sum_{i' = 1}^{\N} \ewght{n}{i'} \adjfuncforward{n}(\epart{n}{i'})} \int \wgtfuncmod{n}(\epart{n}{i}, x) f_{n + 1}(x) \nonumber \\
&\hspace{3cm}\times \sum_{j = 1}^\N \frac{\ewght{n}{j} \udmod{n}(\epart{n}{j}, x)}{\sum_{j' = 1}^{\N} \ewght{n}{j'} \udmod{n}(\epart{n}{j'}, x)} \left(\tstat[j]{n} + \addf{n}(\epart{n}{j}, x) \right)^2  \kissforward{n}{n}(\epart{n}{i},x) \, \mu(\rmd x), \nonumber \\
&= (\K \post[\N]{n}\adjfuncforward{n})^{-1} \sum_{j = 1}^\N \frac{\ewght{n}{j}}{\sumwght{n}} \int f_{n + 1}(x) \left( \tstat[j]{n} + \addf{n}(\epart{n}{j}, x) \right)^2 \, \ukmod{n}(\epart{n}{j}, \rmd x). \label{eq:a':1}
\end{align}
Thus, by the induction hypothesis and Proposition~\ref{prop:hoeffding:tau:marginal}, $a_\N'$ tends in probability to 
\begin{multline*}
(\K \postmod{n}\adjfuncforward{n})^{-1} \left(\eta_n \ukmod{n} f_{n + 1} + \postmod{n}\{(\tstatmod{n} h_n)^2 \ukmod{n} f_{n+1}\} \right. \\ 
\left. + \postmod{n} \ukmod{n}(f_{n + 1} \addf{n}^2)  + 2 \postmod{n}\{\tstatmod{n} h_n \ukmod{n}(f_{n + 1} \addf{n})\} \right) \\
= (\K \postmod{n} \adjfuncforward{n})^{-1} \left(\eta_n \ukmod{n} f_{n + 1} + \postmod{n} \ukmod{n} \{(\tstatmod{n} h_n +\addf{n})^2 f_{n + 1} \} \right). 
\end{multline*}

We turn to the second term. Along the lines of \eqref{eq:a':1} we may write  
$$
a_\N'' = (\K - 1) \K^{-1} \frac{\post[\N]{n}(\ukmod{n} \varphi_N)}{\post[\N]{n} \adjfuncforward{n}}
$$
with
\[
\varphi_\N(x) \eqdef f_{n+1}(x) \left\{ \sum_{i = 1}^\N \frac{\ewght{n}{i} \udmod{n}(\epart{n}{i}, x)}{\sum_{i' = 1}^{\N} \ewght{n}{i'} \udmod{n} (\epart{n}{i'}, x)} \left( \tstat[i]{n} + \addf{n}(\epart{n}{i}, x) \right) \right\}^2, \quad x \in \Xset_{n + 1}. 
\]
Now, Proposition~\ref{prop:hoeffding:tau:marginal} implies pointwise convergence of $\varphi_\N$ in the sense that for all $x \in \Xset_{n + 1}$, $\pP$-a.s., 
\begin{align*}
\lim_{\N \to \infty} \varphi_\N(x) &= f_{n + 1}(x) \left(\frac{\postmod{n}[\tstatmod{n} h_n \udmod{n}(\cdot, x) + \addf{n}(\cdot, x) \udmod{n}(\cdot, x)]}{\postmod{n}[\udmod{n}(\cdot, x)]} \right)^2 \\
&= f_{n + 1}(x) (\bkmod{n})^2(\tstatmod{n} h_n + \addf{n})(x) \\
&= f_{n + 1}(x) (\tstatmod{n + 1})^2 h_{n + 1}(x). 
\end{align*}
Therefore, since $\| \varphi_N \|_\infty \leq \| f_{n + 1} \|_\infty \| h_{n + 1} \|^2_\infty$ for all $\N$, we may apply Lemma~\ref{lem:generalized:lebesgue} in order to obtain the limit  
\[
a_\N'' \pplim (\K - 1) \K^{-1} \frac{\postmod{n} \ukmod{n}\{f_{n + 1} (\tstatmod{n + 1})^2 h_{n + 1}\}}{\postmod{n}\adjfuncforward{n}}. 
\]
Finally, by Proposition~\ref{prop:hoeffding:tau:marginal}(i), $\pP$-a.s., 
\begin{equation} \label{eq:b:lim}
\lim_{\N \to \infty} b_\N = \frac{\postmod{n} \ukmod{n} \1_{\Xset_{n + 1}}}{\postmod{n} \adjfuncforward{n}}, 
\end{equation}
and using \eqref{eq:recursion:FK:marg} yields the limit
\begin{multline}
\frac{a_\N}{b_\N} \pplim \postmod{n + 1}\{f_{n + 1} (\tstatmod{n + 1})^2 h_{n + 1}\} \\
+ \K^{-1} \frac{\eta_n \ukmod{n} f_{n + 1} + \postmod{n} \ukmod{n}\{(\tstatmod{n} h_n +\addf{n})^2 f_{n + 1}\} - \postmod{n} \ukmod{n}\{f_{n + 1} (\tstatmod{n+1})^2 h_{n + 1}\}}{\postmod{n} \ukmod{n} \1_{\Xset_{n + 1}}}.
\end{multline}
Thus, we may complete the proof by noting, using Lemma~\ref{lem:reversibility}(i), that
\begin{multline*}
\postmod{n} \ukmod{n}\{(\tstatmod{n} h_n +\addf{n})^2 f_{n + 1}\} - \postmod{n} \ukmod{n}\{f_{n + 1} (\tstatmod{n + 1})^2 h_{n + 1}\} \\
= \postmod{n} \ukmod{n} \{ \bkmod{n}(\tstatmod{n} h_n + \addf{n} - \tstat{n + 1} h_{n + 1})^2 f_{n+1}\}.  
\end{multline*}
\end{proof}

\begin{proof}[Proof of Theorem~\ref{prop:clt:pseudo:marginal:paris}]
We proceed by induction over $n$ and assume that the result holds for some $n \in \nset$. In addition, we first assume that $\postmod{n + 1}(\tstatmod{n + 1} h_{n + 1} f_{n + 1} + \ftd{n + 1}) = 0$ (this assumption will be removed in the end of the proof). Write
\begin{equation} \label{eq:main:decomp:CLT}
\sqrt{N} \sum_{i = 1}^\N \frac{\ewght{n + 1}{i}}{\sumwght{n + 1}}\{\tstat[i]{n + 1} f_{n + 1}(\epart{n + 1}{i}) +  \ftd{n + 1}(\epart{n + 1}{i})\} = (\sumwght{n + 1} / \N)^{-1} (\termprime + \termbis) / \N,
\end{equation}
where 
\begin{align*}
\termprime &\eqdef \sqrt{\N} \sum_{i = 1}^\N \pE \left[ \ewght{n + 1}{i}\{\tstat[i]{n + 1}f_{n + 1}(\epart{n + 1}{i}) + \ftd{n + 1}(\epart{n + 1}{i})\} \cond \calF{n}{\N} \right],\\
 \termbis &\eqdef \sqrt{\N} \sum_{i = 1}^N \left( \ewght{n+1}{i}\{\tstat[i]{n+1}f_{n+1}(\epart{n+1}{i}) + \ftd{n+1}(\epart{n+1}{i})\}\right.\\
&\left.\hspace{4cm} - \pE \left[\ewght{n+1}{i} \{ \tstat[i]{n+1} f_{n+1}(\epart{n+1}{i}) + \ftd{n+1}(\epart{n+1}{i})\} \cond \calF{n}{\N} \right] \right).
\end{align*}
In the following we establish the weak limit of $(\termprime + \termbis) / \N$, from which the weak limit of \eqref{eq:main:decomp:CLT} follows by Slutsky's lemma. In order to treat the first term $\termprime / \N$ of the decomposition \eqref{eq:main:decomp:CLT}, write, using Lemma~\ref{lem:expectation:incremental:marginal},
\[
\termprime / \N = \sqrt{\N} (\post[\N]{n} \adjfuncforward{n})^{-1} \sum_{i = 1}^N \frac{\ewght{n}{i}}{\sumwght{n}} \{ \tstat[i]{n} \ukmod{n} f_{n + 1}(\epart{n}{i}) + \ukmod{n}(\addf{n} f_{n + 1} + \ftd{n + 1})(\epart{n}{i}) \}.
\]
In the previous expression, $\lim_{\N \to \infty} \post[\N]{n} \adjfuncforward{n} = \postmod{n} \adjfuncforward{n} > 0$, $\pP$-a.s., by proposition~\ref{prop:hoeffding:tau:marginal}. In addition, by Lemma~\ref{lem:reversibility}(i), 
\begin{align*}
\lefteqn{\postmod{n} \{ \tstatmod{n} h_n \ukmod{n} f_{n + 1} + \ukmod{n}(\addf{n} f_{n + 1} + \ftd{n + 1}) \}} \hspace{10mm} \\
&= \iint \postmod{n} \ukmod{n}(\rmd x_{n + 1}) \, \bkmod{n}(x_{n + 1}, \rmd x_n) \\
&\hspace{10mm} \times \{ \tstatmod{n} h_n(x_n) f_{n + 1}(x_{n + 1}) + \addf{n}(x_n, x_{n + 1}) f_{n + 1}(x_{n + 1}) + \ftd{n + 1}(x_{n + 1}) \} \\
&= \postmod{n} \ukmod{n} \1_{\Xset_{n + 1}} \times \postmod{n + 1}(\tstatmod{n + 1} h_{n + 1} f_{n + 1} + \ftd{n + 1}) = 0,
\end{align*}
where the last equality holds by assumption. Thus, applying the induction hypothesis and Slutsky's lemma yields 
\begin{equation} \label{eq:weak:limit:termprime}
\lim_{N \to \infty} \pE \left[ \exp \left(\im u \termprime / \N \right) \right] = \exp \left( - u^2 \frac{\sigma_n^2 \langle \ukmod{n} f_{n + 1}; \ukmod{n}(\addf{n} f_{n + 1} + \ftd{n + 1}) \rangle}{2 (\postmod{n} \adjfuncforward{n})^2} \right).  
\end{equation}

We turn to the second term of \eqref{eq:main:decomp:CLT}. By Lemma~\ref{lem:expectation:incremental:marginal}, $\termbis / \N = \sum_{i = 1}^\N \upsilon^i_\N$, where
$$
\upsilon^i_\N \eqdef \frac{1}{\K \sqrt{\N}} \sum_{j = 1}^{\K} \tilde{\upsilon}_\N(I_{n + 1}^i, J_{n + 1}^{(i, j)}, \epart{n + 1}{i}, \zpart{n + 1}{i}),
$$
with
\begin{multline*}
\tilde{\upsilon}_\N(i, j, x, z) \eqdef \wgtfunc{n}(\epart{n}{i}, x, z)
\left( \{\tstat[j]{n} + \addf{n}(\epart{n}{j}, x)\} f_{n + 1}(x) + \ftd{n + 1}(x) \right) \\
 -  (\post[\N]{n} \adjfuncforward{n})^{-1} \sum_{\ell = 1}^\N \frac{\ewght{n}{\ell}}{\sumwght{n}} \{ \tstat[\ell]{n} \ukmod{n} f_{n + 1}(\epart{n}{\ell}) + \ukmod{n}(\addf{n} f_{n + 1} + \ftd{n + 1})(\epart{n}{\ell}) \}, \\
 \quad (i, j, x, z) \in \intvect{1}{\N}^2 \times \Xset_{n + 1} \times \Zset_{n + 1}.
\end{multline*}
\begin{remark} \label{rem:non-trivial:extension}
In the framework of a fully dominated HMM and a PaRIS driven by the standard bootstrap particle filter, which was the setting considered in \cite{olsson:westerborn:2014b}, the function $\tilde{\upsilon}_\N$ does not, on the contrary to the general case considered here, depend on $i$. As we will see in the next derivations, this dependence calls for a non-trivial refinement of the proof of \cite[Theorem~3]{olsson:westerborn:2014b}. 
\end{remark}

In the following we establish the weak limit of $\sum_{i = 1}^\N \upsilon^i_\N$ using \cite[Theorem~A.3]{douc:moulines:2008}. By construction, $\pE[\upsilon^i_\N \mid \calF{n}{\N}] = 0$; thus, the condition (31) in the mentioned theorem can be checked by establishing that 
\begin{multline}
\label{eq:decomp:sum:square:upsilon}
\sum_{i = 1}^\N \pE[(\upsilon^i_\N)^2 \mid \calF{n}{\N}] = \K^{-1} \pE \left[  \pE \left[ \tilde{\upsilon}_\N^2(I_{n + 1}^1, J_{n + 1}^{(1,1)}, \epart{n + 1}{1}, \zpart{n + 1}{1}) \mid \calG{n + 1}{\N} \right] \mid \calF{n}{\N} \right] \\
 + (\K - 1) \K^{-1} \pE \left[ \pE^2 \left[ \tilde{\upsilon}_\N(I_{n + 1}^1, J_{n + 1}^{(1,1)}, \epart{n + 1}{1}, \zpart{n + 1}{1}) \mid \calG{n + 1}{\N} \right] \mid \calF{n}{\N} \right] 
\end{multline}
converges in probability as $\N \to \infty$. The first term of \eqref{eq:decomp:sum:square:upsilon} is given by
\begin{align}
\lefteqn{\pE\left[ \pE\left[ \tilde{\upsilon}^2_\N(I_{n + 1}^1,J_{n + 1}^{(1,1)}, \epart{n + 1}{1}, \zpart{n + 1}{1}) \cond \calG{n + 1}{\N} \right] \cond \calF{n}{\N}\right]} \nonumber \\
&= \pE\left[\sum_{j = 1}^\N \frac{\ewght{n}{j} \udmod{n} (\epart{n}{j}, \epart{n + 1}{1})}{\sum_{j' = 1}^\N \ewght{n}{j'} \udmod{n} (\epart{n}{j'}, \epart{n + 1}{1})} \int \tilde{\upsilon}^2_N(I_{n + 1}^1, j, \epart{n + 1}{1}, z) \, \ukdist{n}(\epart{n}{\ind{n + 1}{1}}, \epart{n + 1}{1}, \rmd z) \cond \calF{n}{\N} \right], \nonumber \\
&= \sum_{i = 1}^\N \frac{\ewght{n}{i} \adjfuncforward{n}(\epart{n}{i})}{\sum_{i' = 1}^\N \ewght{n}{i'} \adjfuncforward{n}(\epart{n}{i'})} \nonumber \\ 
&\hspace{10mm} \times \int \kissforward{n}{n}(\epart{n}{i}, x)
\sum_{j = 1}^\N \frac{\ewght{n}{j} \udmod{n} (\epart{n}{j}, x)}{\sum_{j' = 1}^{\N} \ewght{n}{j'} \udmod{n} (\epart{n}{j'}, x)} \int \tilde{\upsilon}^2_\N(i, j,x, z) \, \ukdist{n}(\epart{n}{i}, x, \rmd z) \, \mu(\rmd x) \nonumber \\
&= (\post[\N]{n} \adjfuncforward{n})^{-1} (a_\N + b_\N + c_\N) \nonumber \\
&\hspace{10mm} - (\post[\N]{n} \adjfuncforward{n})^{-2} \left(\sum_{\ell=1}^\N \frac{\ewght{n}{\ell}}{\sumwght{n}} \{ \tstat[\ell]{n} \ukmod{n} f_{n + 1}(\epart{n}{\ell}) + \ukmod{n}(\addf{n} f_{n + 1} + \ftd{n + 1})(\epart{n}{\ell}) \} \right)^2, \label{eq:upsilon:2:decomp} 
\end{align}
where
\begin{align*}
a_\N &\eqdef \sum_{j = 1}^\N \frac{\ewght{n}{j}}{\sumwght{n}}(\tstat[j]{n})^2  \int \udmod{n} (\epart{n}{j}, x) \varphi_\N(x) f_{n + 1}^2(x) \, \mu(\rmd x), \\
b_\N &\eqdef \sum_{i = 1}^\N \frac{\ewght{n}{i}}{\sumwght{n}} \int \lwd{n}(\epart{n}{i}, x) \sum_{j = 1}^\N \frac{\ewght{n}{j} \udmod{n} (\epart{n}{j}, x)}{\sum_{j' = 1}^{\N} \ewght{n}{j'} \udmod{n} (\epart{n}{j'}, x) }\left( \addf{n}(\epart{n}{j}, x) f_{n + 1}(x) + \ftd{n + 1}(x) \right)^2 \, \mu(\rmd x), \\
c_\N &\eqdef  2 \sum_{j = 1}^\N \frac{\ewght{n}{j}}{\sumwght{n}} \tstat[j]{n} \int  \udmod{n}(\epart{n}{j}, x) \varphi_N(x) f_{n + 1}(x) \left( \addf{n}(\epart{n}{j}, x) f_{n + 1}(x) + \ftd{n + 1}(x) \right) \, \mu(\rmd x),
\end{align*}
with 
$$
\varphi_\N(x) \eqdef \frac{\sum_{i = 1}^\N \ewght{n}{i} \lwd{n}(\epart{n}{i}, x)}{\sum_{j' = 1}^\N \ewght{n}{j'} \udmod{n}(\epart{n}{j'}, x)}, \quad x \in \Xset_{n + 1},   
$$
being a random function and 
\begin{multline*}
\lwd{n}(x_n, x_{n + 1}) \eqdef \int \ukest{n}{z}(x_n, x_{n + 1}) \wgtfunc{n}(x_n, x_{n + 1}, z) \, \ukdist{n}(x_n, x_{n + 1}, \rmd z), \\ 
(x_n, x_{n + 1}) \in \Xset_n \times \Xset_{n + 1},  
\end{multline*}
an unnormalised transition density. Define the deterministic function $\varphi(x) \eqdef \post{n}[\lwd{n}(\cdot, x)] / \post{n}[\udmod{n}(\cdot, x)]$, $x \in \Xset_{n + 1}$; then by Proposition~\ref{prop:hoeffding:tau:marginal}, $\pP$-a.s., $\lim_{\N \to \infty} \varphi_\N(x) = \varphi(x)$ for all $x$. Thus, since 
\begin{multline} \label{eq:a:bound}
\left| a_\N -  \sum_{j = 1}^\N \frac{\ewght{n}{j}}{\sumwght{n}}(\tstat[j]{n})^2 \int \udmod{n}(\epart{n}{j}, x) f_{n + 1}^2(x) \varphi(x) \, \mu(\rmd x) \right| \\
\leq \| h_n \|^2_\infty \sum_{j = 1}^\N \frac{\ewght{n}{j}}{\sumwght{n}} \int \udmod{n}(\epart{n}{j}, x) f_{n + 1}^2(x) |\varphi_\N(x) - \varphi(x)| \, \mu(\rmd x)   
\end{multline}
and $f_{n + 1}^2(x) |\varphi_\N(x) - \varphi(x)| \leq 2 \| \wgtfunc{n} \|_\infty \| f_{n + 1} \|_\infty^2$ for all $x$, Lemma~\ref{lem:generalized:lebesgue} implies that the right-hand side of \eqref{eq:a:bound} tends to zero as $\N \to \infty$. Thus, 
by Lemma~\ref{lem:hoeffding:tau:square},
\[
a_\N \pplim \eta_n \ukmod{n}(f^2_{n + 1} \varphi)  + \postmod{n} \{ (\tstatmod{n})^2 h_n \ukmod{n}(f^2_{n + 1} \varphi) \}.
\]

We compute the limit of $b_\N$. For this purpose, define the mapping 
\[
 \psi_\N(x) \eqdef \sum_{j = 1}^\N \frac{\ewght{n}{j} \udmod{n}(\epart{n}{j}, x)}{\sum_{j' = 1}^\N \ewght{n}{j'} \udmod{n} (\epart{n}{j'}, x)} \left( \addf{n}(\epart{n}{j}, x) f_{n + 1}(x) + \ftd{n + 1}(x) \right)^2, \quad x \in \Xset_{n + 1}. 
\]
which tends, for all $x$, by Proposition~\ref{prop:hoeffding:tau:marginal}, $\pP$-a.s. to 
$$
\psi(x) \eqdef \bkmod{n}(\addf{n} f_{n + 1} + \ftd{n + 1})^2(x). 
$$
Moreover, since $\psi_\N(x) \leq (\| \addf{n} \|_\infty \|f_{n + 1} \|_\infty + \| \ftd{n + 1} \|_\infty)^2$ for all $x$, Lemma~\ref{lem:generalized:lebesgue} implies that 
\begin{align*}
\lefteqn{b_\N \pplim \int \post{n} [\lwd{n}(\cdot, x)] \psi(x) \, \mu(\rmd x)} \hspace{15mm} \\
&= \iint \lwd{n}(x_n, x_{n + 1}) \bkmod{n}(\addf{n} f_{n + 1} + \ftd{n + 1})^2(x_{n + 1}) \, \postmod{n}(\rmd x_n) \, \mu(\rmd x_{n + 1}) \\
&= \postmod{n} \ukmod{n}\{ \varphi (\addf{n} f_{n + 1} + \ftd{n + 1})^2 \},  
\end{align*}
where the last equality is obtained using Lemma~\ref{lem:reversibility}(i).
 
We turn to the last term, $c_\N$. Since 
\begin{multline*}
\left| c_\N -  2 \sum_{j = 1}^\N \frac{\ewght{n}{j}}{\sumwght{n}} \tstat[j]{n} \int \udmod{n} (\epart{n}{j}, x) f_{n + 1}(x) \varphi(x) \left( \addf{n}(\epart{n}{j} ,x) f_{n + 1}(x) + \ftd{n + 1}(x)\right) \, \mu(\rmd x) \right| \\
\leq \| h_n \|_{\infty}( \| \addf{n} \|_\infty \| f_{n + 1} \|_\infty + \| \ftd{n + 1} \|_\infty) \\ 
\times \sum_{j = 1}^\N \frac{\ewght{n}{j}}{\sumwght{n}} \int \udmod{n}(\epart{n}{j}, x) f_{n + 1}(x) |\varphi_\N(x) - \varphi(x)| \, \mu(\rmd x),   
\end{multline*}
where the right-hand side tends, by Lemma~\ref{lem:generalized:lebesgue}, to zero in probability as $\N \to \infty$, using again Proposition~\ref{prop:hoeffding:tau:marginal} yields
\[
c_\N \pplim 2 \postmod{n} \{ \tstatmod{n} h_n \ukmod{n} [\varphi f_{n + 1}( \addf{n} f_{n + 1} + \ftd{n +1})] \}. 
\]

Finally, by Proposition~\ref{prop:hoeffding:tau:marginal}, $\pP$-a.s., 
\begin{multline} \label{eq:zero:limit}
\lim_{\N \to \infty} \sum_{\ell = 1}^N \frac{\ewght{n}{\ell}}{\sumwght{n}} \{ \tstat[\ell]{n} \ukmod{n} f_{n + 1}(\epart{n}{\ell}) + \ukmod{n}(\addf{n} f_{n + 1} + \ftd{n + 1})(\epart{n}{\ell}) \} \\
=  \postmod{n}\{ \tstatmod{n} h_n \ukmod{n} f_{n + 1} + \ukmod{n}(\addf{n} f_{n + 1} + \ftd{n + 1})\} = 0,
\end{multline}
where the last equality follows by Lemma~\ref{eq:critical:identity} and assumption, and since also $\lim_{\N \to \infty}\post[\N]{n} \adjfuncforward{n} = \postmod{n} \adjfuncforward{n}$, $\pP$-a.s., the second term of \eqref{eq:upsilon:2:decomp} tends $\pP$-a.s. to zero.   

To sum up, as $\N \to \infty$, the first term of \eqref{eq:decomp:sum:square:upsilon} satisfies the limit 
\begin{multline} \label{eq:limit:first:term:i}
\K^{-1} \pE \left[ \pE \left[\tilde{\upsilon}_\N^2(I_{n + 1}^1, J_{n + 1}^{(1,1)}, \epart{n + 1}{1}, \zpart{n + 1}{1}) \cond \calG{n + 1}{\N} \right] \cond \calF{n}{\N} \right] \\
\pplim (\K \postmod{n} \adjfuncforward{n})^{-1} \left( \eta_n \ukmod{n}(f^2_{n + 1} \varphi) + \postmod{n} \ukmod{n} \{ \varphi [(\tstatmod{n} h_n + \addf{n}) f_{n+1} + \ftd{n + 1}]^2 \} \right). 
\end{multline}

We turn to the second term of \eqref{eq:decomp:sum:square:upsilon} and write 
\begin{align*}
\lefteqn{\pE \left[ \pE^2 \left[ \tilde{\upsilon}_\N(I_{n + 1}^1, J_{n + 1}^{(1,1)}, \epart{n + 1}{1}, \zpart{n + 1}{1}) \cond \calG{n + 1}{\N} \right] \cond \calF{n}{\N} \right]} \\
&= \pE\left[ \left( \sum_{j = 1}^\N \frac{\ewght{n}{j} \udmod{n}(\epart{n}{j}, \epart{n + 1}{1})}{\sum_{j' = 1}^\N \ewght{n}{j'} \udmod{n}(\epart{n}{j'}, \epart{n + 1}{1})} \right. \right. \\ 
& \left. \left. \hspace{40mm} \times \int \tilde{\upsilon}_\N(\ind{n + 1}{1}, j, \epart{n + 1}{1}, z) \, \ukdist{n}(\epart{n}{\ind{n + 1}{1}}, \epart{n + 1}{1}, \rmd z) \right)^2 \cond \calF{n}{\N}\right] \\
&= \sum_{i = 1}^{\N}  \frac{\ewght{n}{i} \adjfuncforward{n}(\epart{n}{i})}{\sum_{i' = 1}^\N \ewght{n}{i'} \adjfuncforward{n}(\epart{n}{i'})} \\ 
& \hspace{10mm} \times \int \left( \sum_{j = 1}^\N \frac{\ewght{n}{j} \udmod{n} (\epart{n}{j}, x)}{\sum_{j' = 1}^{\N} \ewght{n}{j'} \udmod{n}(\epart{n}{j'}, x)} \int \tilde{\upsilon}_\N(i, j, x, z) \, \ukdist{n}(\epart{n}{i}, x, \rmd z)\right)^2 
\prop{n}(\epart{n}{i}, \rmd x) \\
&= (\post[\N]{n} \adjfuncforward{n})^{-1} \sum_{i = 1}^\N \frac{\ewght{n}{i}}{\sumwght{n}} \int \wgtfuncmod{n}(\epart{n}{i}, x) \gamma_\N^2(x) \, \ukmod{n}(\epart{n}{i}, \rmd x) \\ 
&\hspace{10mm} - 2 (\post[\N]{n} \adjfuncforward{n})^{-1} (\post[\N]{n} \ukmod{n} \gamma_\N) \sum_{\ell = 1}^\N \frac{\ewght{n}{\ell}}{\sumwght{n}} \{ \tstat[\ell]{n} \ukmod{n} f_{n + 1}(\epart{n}{\ell}) + \ukmod{n}(\addf{n} f_{n + 1} + \ftd{n + 1})(\epart{n}{\ell}) \} \\
&\hspace{20mm} + (\post[\N]{n} \adjfuncforward{n})^{-1} \left( \sum_{\ell = 1}^\N \frac{\ewght{n}{\ell}}{\sumwght{n}} \{ \tstat[\ell]{n} \ukmod{n} f_{n + 1}(\epart{n}{\ell}) + \ukmod{n}(\addf{n} f_{n + 1} + \ftd{n + 1})(\epart{n}{\ell})\} \right)^2, 
\end{align*}
where 
$$
\gamma_\N(x) \eqdef f_{n + 1}(x) \sum_{j = 1}^\N \frac{\ewght{n}{j} \udmod{n} (\epart{n}{j}, x)}{\sum_{j' = 1}^\N \ewght{n}{j'} \udmod{n}(\epart{n}{j'}, x)} \left(\tstat[j]{n} + \addf{n}(\epart{n}{j}, x) \right) + \ftd{n + 1}(x), \quad x \in \Xset_{n + 1}. 
$$
Since, by Proposition~\ref{prop:hoeffding:tau:marginal}, for all $x \in \Xset_{n + 1}$, $\lim_{\N \to \infty} \gamma_\N(x) = \gamma(x)$ $\pP$-a.s., where 
$$
\gamma(x) \eqdef f_{n + 1}(x) \bkmod{n}(\tstatmod{n} h_n + \addf{n})(x) + \ftd{n + 1}(x) = f_{n + 1}(x) \tstatmod{n + 1} h_{n + 1}(x) + \ftd{n + 1}(x),  
$$
Lemma~\ref{lem:generalized:lebesgue} implies that  
\begin{multline} \label{eq:limit:second:term:i}
 \sum_{i = 1}^\N \frac{\ewght{n}{i}}{\sumwght{n}} \int \wgtfuncmod{n}(\epart{n}{i}, x) \gamma_\N^2(x) \, \ukmod{n}(\epart{n}{i}, \rmd x) \\
\pplim \iint \wgtfuncmod{n}(x_n, x_{n + 1}) \gamma^2(x_{n + 1}) \, \postmod{n}(\rmd x_n) \, \ukmod{n}(x_n, \rmd x_{n + 1})
\end{multline}
and $\post[\N]{n} \ukmod{n} \gamma_\N \pplim \postmod{n} \ukmod{n} \gamma$. 
Thus, by (\ref{eq:zero:limit}--\ref{eq:limit:second:term:i}), 
\begin{multline*}
\sum_{i = 1}^\N \pE[(\upsilon^i_\N)^2 \mid \calF{n}{\N}] \pplim 
\incrementalvar_n^2 \langle f_{n + 1}, \ftd{n + 1} \rangle(h_n) \\ 
\eqdef (\K \postmod{n} \adjfuncforward{n})^{-1} \left( \eta_n \ukmod{n}(f^2_{n + 1} \varphi) + \postmod{n} \ukmod{n} \{ \varphi [(\tstatmod{n} h_n + \addf{n}) f_{n + 1} + \ftd{n + 1}]^2 \} \right.  \\
\left. \vphantom{\eta_n \ukmod{n}(f^2_{n + 1} \varphi + \post{n} \uk{n} \{(\bkw{\post{n}} \wgtfunc{n}) [(\tstat{n} h_n + \addf{n}) f_{n+1} + \ftd{n + 1}]^2 \}} + (\K - 1) \postmod{n} \ukmod{n}\{\wgtfuncmod{n}(f_{n + 1} \tstatmod{n + 1} h_{n + 1} + \ftd{n + 1})^2 \} \right), 
\end{multline*}
which verifies condition (i) in \cite[Theorem~A.3]{douc:moulines:2008}. Now, with the conditional relative weight variance $\ukestvar{n}$ defined as in \eqref{eq:def:ukestvar} it holds that  
$$
\varphi(x) - \bkmod{n} \wgtfuncmod{n}(x) = \frac{\postmod{n}[\udmod{n}(\cdot, x) \ukestvar{n}(\cdot, x)]}{\postmod{n}[\udmod{n}(\cdot, x)]} = \bkmod{n} \ukestvar{n}(x). 
$$
Then, since, applying Lemma~\ref{lem:reversibility}(i) (twice),
\begin{multline*}
\lefteqn{\postmod{n} \ukmod{n} \{\varphi [(\tstatmod{n} h_n + \addf{n}) f_{n + 1} + \ftd{n + 1}]^2 \} -  \postmod{n} \ukmod{n}\{\wgtfuncmod{n}(f_{n + 1} \tstatmod{n + 1} h_{n + 1} + \ftd{n + 1})^2 \}} \\
= \postmod{n} \ukmod{n} \{ \wgtfuncmod{n} f_{n + 1}^2 \bkmod{n} (\tstatmod{n} h_n + \addf{n} - \tstatmod{n + 1} h_{n + 1})^2 \} \\ 
+ \postmod{n} \ukmod{n} \{ \bkmod{n} \ukestvar{n} [(\tstatmod{n} h_n + \addf{n}) f_{n + 1} + \ftd{n + 1}]^2 \}, 
\end{multline*}
we may express the previous limit as 
\begin{multline*}
\incrementalvar_n^2 \langle f_{n + 1}, \ftd{n + 1} \rangle =  (\K \postmod{n} \adjfuncforward{n})^{-1} \left( \vphantom{\postmod{n} \ukmod{n} \{ \wgtfuncmod{n} f_{n + 1}^2 \bkmod{\postmod{n}} (\tstatmod{n} h_n + \addf{n} - \tstatmod{n + 1} h_{n + 1})^2 \} + \K \postmod{n} \ukmod{n} \{\wgtfuncmod{n}(f_{n + 1} \tstat{n + 1} h_{n + 1} + \ftd{n + 1})^2} \eta_n \ukmod{n}(f^2_{n + 1} \bkmod{n} \wgtfuncmod{n}) \right. \\
+ \postmod{n} \ukmod{n} \{ \wgtfuncmod{n} f_{n + 1}^2 \bkmod{n} (\tstatmod{n} h_n + \addf{n} - \tstatmod{n + 1} h_{n + 1})^2 \} \\
+ \K \postmod{n} \ukmod{n} \{ \wgtfuncmod{n}(f_{n + 1} \tstatmod{n + 1} h_{n + 1} + \ftd{n + 1})^2 \} \\
\left. \vphantom{\post{n} \uk{n} \{ \wgtfunc{n} f_{n + 1}^2 \bkw{\post{n}} (\tstat{n} h_n + \addf{n} - \tstat{n + 1} h_{n + 1})^2 \} + \K \post{n} \uk{n} \{\wgtfunc{n}(f_{n + 1} \tstat{n + 1} h_{n + 1} + \ftd{n + 1})^2}
+ \eta_n \ukmod{n}(f^2_{n + 1} \bkmod{n} \ukestvar{n}) + \postmod{n} \ukmod{n} \{ \bkmod{n} \ukestvar{n} [(\tstatmod{n} h_n + \addf{n}) f_{n + 1} + \ftd{n + 1}]^2 \} \right).  
\end{multline*}

Condition~(ii) in \cite[Theorem~A.3]{douc:moulines:2008} is satisfied as an immediate consequence of  \hypref{assum:bound:filter:pseudomarginal}; indeed, letting $d_n \eqdef 2 \| \wgtfunc{n} \|_\infty (\| h_{n + 1} \|_\infty \| f_{n + 1} \|_\infty + \| \ftd{n + 1} \|_\infty)$ it holds that $| \upsilon^i_\N | \leq d_n / \sqrt{\N}$ for all $i \in \intvect{1}{\N}$, and consequently, for all $\epsilon>0$, $\pP$-a.s., 
\[
\sum_{i = 1}^\N \pE \left[ (\upsilon^i_\N)^2 \1_{\{|\upsilon^i_\N| \geq \epsilon\}} \cond \calF{n}{\N} \right] \leq d_n^2 \1_{\{ d_n \geq \epsilon \sqrt{\N} \}}, 
\]
where the right-hand side tends to zero as $\N \to \infty$. Thus, \cite[Theorem~A.3]{douc:moulines:2008} provides the limit 
$$
\pE \left[ \exp \left(\im u \termbis / \N \right) \cond \calF{n}{\N} \right] \pplim \exp \left( - u^2 \incrementalvar_n^2 \langle f_{n + 1}, \ftd{n + 1} \rangle / 2 \right).    
$$
Now, since $\lim_{\N \to \infty} \sumwght{n + 1} / \N = \postmod{n} \ukmod{n} \1_{\Xset_{n + 1}} / \postmod{n} \adjfuncforward{n}$ $\pP$-a.s. by Proposition~\ref{prop:hoeffding:tau:marginal}, we may combine the previous limit with \eqref{eq:weak:limit:termprime}, \cite[Lemma~A.5]{delmoral:moulines:olsson:verge:2016}, and Slutsky's lemma in order to obtain the weak convergence  
\begin{multline} \label{eq:first:weak:limit}
(\termprime + \termbis) / \sumwght{n + 1} \\ 
\dlim N \left( 0, \frac{\sigma_n^2(h_n, \ukmod{n} f_{n + 1}, \ukmod{n}(\addf{n} f_{n + 1} + \ftd{n + 1})) + (\postmod{n} \adjfuncforward{n})^2 \incrementalvar_n^2 \langle f_{n + 1}, \ftd{n + 1} \rangle}{(\postmod{n} \ukmod{n} \1_{\Xset_{n + 1}})^2} \right).    
\end{multline}
In order to treat the general case where $\postmod{n + 1}(\tstatmod{n + 1} h_{n + 1} f_{n + 1} + \ftd{n + 1})$ is non-zero, define 
\[
\bar{f}_{n + 1} \eqdef \ftd{n + 1} - \postmod{n + 1}(\tstatmod{n + 1} h_{n + 1} f_{n + 1} + \ftd{n + 1}). 
\]
Since the functions $f_{n + 1}$ and $\bar{f}_{n + 1}$ satisfy $\postmod{n + 1}(\tstatmod{n + 1} h_{n+1} f_{n + 1} + \bar{f}_{n + 1}) = 0$, \eqref{eq:first:weak:limit} provides    
\begin{multline*}
\sqrt{\N} \left( \sum_{i = 1}^\N \frac{\ewght{n + 1}{i}}{\sumwght{n + 1}}\{\tstat[i]{n + 1} f_{n + 1}(\epart{n + 1}{i}) +  \ftd{n + 1}(\epart{n + 1}{i})\} - \postmod{n + 1}(\tstatmod{n + 1} h_{n + 1} f_{n + 1} + \ftd{n + 1}) \right) \\
= \sqrt{\N} \left( \sum_{i = 1}^\N \frac{\ewght{n + 1}{i}}{\sumwght{n + 1}}\{\tstat[i]{n + 1} f_{n + 1}(\epart{n + 1}{i}) +  \bar{f}_{n + 1}(\epart{n + 1}{i})\} \right) \\ 
\dlim N \left(0, \sigma_{n + 1}^2 (h_{n + 1}, f_{n + 1}, \ftd{n + 1}) \right), 
\end{multline*}
where 
\begin{multline} \label{eq:as:var:first:expression}
\sigma^2_{n + 1}(h_{n + 1}, f_{n + 1}, \ftd{n + 1}) \\
\eqdef \frac{\sigma_n^2(h_n, \ukmod{n} f_{n + 1}, \ukmod{n}(\addf{n} f_{n + 1} + \bar{f}_{n + 1}))
+ (\postmod{n} \adjfuncforward{n})^2 \incrementalvar_n^2 \langle f_{n + 1}, \bar{f}_{n + 1} \rangle}{(\postmod{n} \ukmod{n} \1_{\Xset_{n + 1}})^2} \\
= \frac{\sigma^2_n(h_n, \ukmod{n} f_{n + 1}, \ukmod{n}(\addf{n} f_{n + 1} + \bar{f}_{n + 1}))}{(\postmod{n} \ukmod{n} \1_{\Xset_{n + 1}})^2} +  \frac{\postmod{n} \adjfuncforward{n} \eta_n \ukmod{n}(f^2_{n + 1} \bkmod{n} \wgtfuncmod{n})}{\M (\postmod{n} \ukmod{n} \1_{\Xset_{n + 1}})^2} \\
+ \frac{\postmod{n} \adjfuncforward{n} \postmod{n} \ukmod{n} \{\wgtfuncmod{n}(f_{n + 1} \tstatmod{n + 1} h_{n + 1} + \bar{f}_{n + 1} )^2 \}}{(\postmod{n} \ukmod{n} \1_{\Xset_{n + 1}})^2} \\
+ \frac{\postmod{n} \adjfuncforward{n} \postmod{n} \ukmod{n} \{\wgtfuncmod{n} f^2_{n + 1} \bkmod{n} (\tstatmod{n} h_n + \addf{n} - \tstatmod{n + 1} h_{n + 1})^2\}}{\M (\postmod{n} \ukmod{n} \1_{\Xset_{n + 1}})^2} \\
+ \frac{\postmod{n} \adjfuncforward{n} \eta_n \ukmod{n}(f^2_{n + 1} \bkmod{n} \ukestvar{n})}{\M (\postmod{n} \ukmod{n} \1_{\Xset_{n + 1}})^2} + \frac{\postmod{n} \adjfuncforward{n} \postmod{n} \ukmod{n} \{ \bkmod{n} \ukestvar{n} [(\tstatmod{n} h_n + \addf{n}) f_{n + 1} + \bar{f}_{n + 1}]^2 \}}{\M (\postmod{n} \ukmod{n} \1_{\Xset_{n + 1}})^2}. 
\end{multline}
Since by Lemma~\ref{lem:reversibility}(i),
\begin{multline*}
\postmod{n} \ukmod{n} \{\wgtfuncmod{n} f^2_{n + 1} \bkmod{n} (\tstatmod{n} h_n + \addf{n} - \tstatmod{n + 1} h_{n + 1})^2\} \\
= \postmod{n} \ukmod{n} \{ \bkmod{n} (\tstatmod{n} h_n + \addf{n} - \tstatmod{n + 1} h_{n + 1})^2 f^2_{n + 1} \bkmod{n} \wgtfuncmod{n} \}, 
\end{multline*}
and 
\begin{multline*}
\postmod{n} \ukmod{n} \{ \bkmod{n} \ukestvar{n} [(\tstatmod{n} h_n + \addf{n}) f_{n + 1} + \bar{f}_{n + 1}]^2 \} \\
= \postmod{n} \ukmod{n} \{ \ukestvar{n} (f_{n + 1} \tstatmod{n + 1} h_{n + 1} + \bar{f}_{n + 1})^2 \} \\
+ \postmod{n} \ukmod{n} \{ \bkmod{n} (\tstatmod{n} h_n + \addf{n} - \tstatmod{n + 1} h_{n + 1})^2 f_{n + 1}^2 \bkmod{n} \ukestvar{n} \},  
\end{multline*}
inserting the expression of $\eta_n$ given in \eqref{eq:def:etak} into \eqref{eq:as:var:first:expression} yields,   
\begin{multline} \label{eq:as:var:second:expression}
\sigma^2_{n + 1}(h_{n + 1}, f_{n + 1}, \ftd{n + 1}) \\
= \frac{\sigma^2_n(h_n, \ukmod{n} f_{n + 1}, \ukmod{n}(\addf{n} f_{n + 1} + \bar{f}_{n + 1}))}{(\postmod{n} \ukmod{n} \1_{\Xset_{n + 1}})^2}
+ \frac{\postmod{n} \adjfuncforward{n} \postmod{n} \ukmod{n} \{\wgtfuncmod{n}(f_{n + 1} \tstatmod{n + 1} h_{n + 1} + \bar{f}_{n + 1})^2 \}}{(\postmod{n} \ukmod{n} \1_{\Xset_{n + 1}})^2} \\
+ \frac{\postmod{n} \adjfuncforward{n}}{(\postmod{n} \ukmod{n} \1_{\Xset_{n + 1}})^2} \sum_{m = 0}^n \frac{\postmod{m} \ukmod{m}\{\bkmod{m}(\tstatmod{m} h_m + \addf{m} - \tstatmod{m + 1} h_{m + 1})^2 \ukmod{m + 1, n}(f_{n + 1}^2 \bkmod{n} \wgtfuncmod{n})\}}{\K^{n - m + 1} \postmod{m} \ukmod{m, n - 1} \1_{\Xset_n}} \\
+ \frac{\postmod{n} \adjfuncforward{n} \postmod{n} \ukmod{n} \{\ukestvar{n}(f_{n + 1} \tstatmod{n + 1} h_{n + 1} + \bar{f}_{n + 1})^2 \}}{(\postmod{n} \ukmod{n} \1_{\Xset_{n + 1}})^2} \\
+ \frac{\postmod{n} \adjfuncforward{n}}{(\postmod{n} \ukmod{n} \1_{\Xset_{n + 1}})^2} \sum_{m = 0}^n \frac{\postmod{m} \ukmod{m}\{\bkmod{m}(\tstatmod{m} h_m + \addf{m} - \tstatmod{m + 1} h_{m + 1})^2 \ukmod{m + 1, n}(f_{n + 1}^2 \bkmod{n} \ukestvar{n})\}}{\K^{n - m + 1} \postmod{m} \ukmod{m, n - 1} \1_{\Xset_n}}. 
\end{multline}

It remains to establish the non-recursive expression \eqref{eq:non-recursive:as:var:affine} of the asymptotic variance. We proceed by induction and assume that \eqref{eq:non-recursive:as:var:affine} holds true at time $n$. Recall the definitions of the kernels $\retrokmodmod_{m, n}$ and $\retrokmodmodnorm_{m, n}$ in \eqref{eq:def:retrokmodmod} and \eqref{eq:def:retrokmodmodnorm}, respectively. First, note that since by Lemma~\ref{eq:critical:identity}, 
$$
\postmod{0:n} \{h_n \ukmod{n} f_{n + 1} + \ukmod{n}(\addf{n} f_{n + 1} + \bar{f}_{n + 1})\} = 0, 
$$
it holds for all $m \in \intvect{0}{n - 1}$,  
\begin{align} 
\lefteqn{\retrokmodmodnorm_{m, n} \{h_n \ukmod{n} f_{n + 1} + \ukmod{n}(\addf{n} f_{n + 1} + \bar{f}_{n + 1})\}} \nonumber \hspace{10mm} \\
&= \retrokmodmod_{m, n} \{h_n \ukmod{n} f_{n + 1} + \ukmod{n}(\addf{n} f_{n + 1} + \bar{f}_{n + 1})\} \nonumber \\ 
&= \retrokmodmod_{m, n + 1} \{ h_{n + 1} f_{n + 1} + \ftd{n + 1} - \postmod{n + 1}(\tstatmod{n + 1} h_{n + 1} f_{n + 1} + \ftd{n + 1}) \} \nonumber \\
&= \retrokmodmodnorm_{m, n + 1}(h_{n + 1} f_{n + 1} + \ftd{n + 1}). \label{eq:retrokmodmodnorm:one:step:forward}
\end{align}
Further, by definition, 
\begin{equation} \label{eq:tstatmod:vs:retrokmodmodnorm}
f_{n + 1} \tstatmod{n + 1} h_{n + 1} + \bar{f}_{n + 1} = \retrokmodmodnorm_{n + 1, n + 1} (h_{n + 1} f_{n + 1} + \ftd{n + 1}).
\end{equation}
Since by \eqref{eq:retrokmodmodnorm:one:step:forward}, 
\begin{multline*}
\frac{\chi (\initwgtfunc
\retrokmodmodnorm_{0, n}
\{ h_n \ukmod{n} f_{n + 1} + \ukmod{n}(\addf{n} f_{n + 1} + \bar{f}_{n + 1})\} )^2}{(\chi \ukmod{0, n - 1} \1_{\Xset^n})^2 (\postmod{n} \ukmod{n} \1_{\Xset_{n + 1}})^2} \\
= \frac{\chi \{ \initwgtfunc
\retrokmodmodnorm_{0, n + 1}
( h_{n + 1} f_{n + 1} + \ftd{n + 1}) \}^2}{(\chi \ukmod{0, n} \1_{\Xset^{n + 1}})^2}, 
\end{multline*}
using the induction hypothesis yields   
\begin{multline*}
 \frac{\sigma^2_n(h_n, \ukmod{n} f_{n + 1}, \ukmod{n}(\addf{n} f_{n + 1} + \bar{f}_{n + 1}))}{(\postmod{n} \ukmod{n} \1_{\Xset_{n + 1}})^2} = 
  \frac{\chi \{ \initwgtfunc
  \retrokmodmodnorm_{0, n + 1}
  ( h_{n + 1} f_{n + 1} + \ftd{n + 1}) \}^2}{(\chi \ukmod{0, n} \1_{\Xset^{n + 1}})^2} \\
+ \frac{\sigma^2_n \langle (\wgtfuncmod{\ell})_{\ell = 0}^{n - 1} \rangle(h_n, \ukmod{n} f_{n + 1}, \ukmod{n}(\addf{n} f_{n + 1} + \bar{f}_{n + 1}))}{(\postmod{n} \ukmod{n} \1_{\Xset_{n + 1}})^2} \\
+ \frac{\sigma^2_n \langle (\ukestvar{\ell})_{\ell = 0}^{n - 1} \rangle(h_n, \ukmod{n} f_{n + 1}, \ukmod{n}(\addf{n} f_{n + 1} + \bar{f}_{n + 1}))}{(\postmod{n} \ukmod{n} \1_{\Xset_{n + 1}})^2}. 
\end{multline*}
Now, for every $(\varphi_\ell)_{\ell \in \nset}$, using \eqref{eq:partial:variance}, \eqref{eq:retrokmodmodnorm:one:step:forward}, and \eqref{eq:tstatmod:vs:retrokmodmodnorm}, 
\begin{multline*}
\frac{\sigma^2_n \langle (\varphi_\ell)_{\ell = 0}^{n - 1} \rangle(h_n, \ukmod{n} f_{n + 1}, \ukmod{n}(\addf{n} f_{n + 1} + \bar{f}_{n + 1}))}{(\postmod{n} \ukmod{n} \1_{\Xset_{n + 1}})^2} \\
+ \frac{\postmod{n} \adjfuncforward{n} \postmod{n} \ukmod{n} \{\varphi_n(f_{n + 1} \tstatmod{n + 1} h_{n + 1} + \bar{f}_{n + 1})^2 \}}{(\postmod{n} \ukmod{n} \1_{\Xset_{n + 1}})^2} \\
+ \frac{\postmod{n} \adjfuncforward{n}}{(\postmod{n} \ukmod{n} \1_{\Xset_{n + 1}})^2} \sum_{m = 0}^n \frac{\postmod{m} \ukmod{m}\{\bkmod{m}(\tstatmod{m} h_m + \addf{m} - \tstatmod{m + 1} h_{m + 1})^2 \ukmod{m + 1, n}(f_{n + 1}^2 \bkmod{n} \varphi_n)\}}{\K^{n - m + 1} \postmod{m} \ukmod{m, n - 1} \1_{\Xset_n}} \\
= \sum_{m = 0}^{n - 1} \frac{\postmod{m} \adjfuncforward{m} \postmod{m} \ukmod{m} \{ \varphi_m [\retrokmodmodnorm_{m + 1, n + 1} (h_{n + 1} f_{n + 1} + \ftd{n + 1})]^2\}}{(\postmod{m} \ukmod{m, n} \1_{\Xset_{n + 1}})^2} \\
+ \sum_{m = 0}^{n - 1} \sum_{\ell = 0}^m \frac{\postmod{m} \adjfuncforward{m} \postmod{\ell} \ukmod{\ell} \{\bkmod{\ell}(\tstatmod{\ell} h_{\ell} + \addf{\ell} - \tstatmod{\ell + 1} h_{\ell +1})^2 \ukmod{\ell + 1, m}( \bkmod{m} \varphi_m 
[\ukmod{m + 1, n} f_{n + 1}]^2 
)\}}{\K^{m - \ell + 1} (\postmod{\ell} \ukmod{\ell, m - 1} \1_{\Xset_m})(\postmod{m} \ukmod{m, n} \1_{\Xset_{n + 1}})^2} \\
+ \frac{\postmod{n} \adjfuncforward{n} \postmod{n} \ukmod{n} \{ \varphi_n  
[\retrokmodmodnorm_{n + 1, n + 1} (h_{n + 1} f_{n + 1} + \ftd{n + 1})]^2 
\}}{(\postmod{n} \ukmod{n} \1_{\Xset_{n + 1}})^2} \\
+ \postmod{n} \adjfuncforward{n} \sum_{\ell = 0}^n \frac{\postmod{\ell} \ukmod{\ell} \{\bkmod{\ell}(\tstatmod{\ell} h_\ell + \addf{\ell} - \tstatmod{\ell + 1} h_{\ell + 1})^2 \ukmod{\ell + 1, n}(f_{n + 1}^2 \bkmod{n} \varphi_n)\}}{\K^{n - \ell + 1} (\postmod{\ell} \ukmod{\ell, n - 1} \1_{\Xset_n}) (\postmod{n} \ukmod{n} \1_{\Xset_{n + 1}})^2}, 
\end{multline*}
which indeed equals $\sigma^2_{n + 1} \langle (\varphi_\ell)_{\ell = 0}^n \rangle(h_{n + 1}, f_{n + 1}, \ftd{n + 1})$. Consequently, \eqref{eq:as:var:second:expression} can be expressed as 
\begin{multline*}
\sigma^2_{n + 1}(h_{n + 1}, f_{n + 1}, \ftd{n + 1}) = 
\frac{\chi \{ \initwgtfunc 
\retrokmodmodnorm_{0, n + 1}( h_{n + 1} f_{n + 1} + \ftd{n + 1}) \}^2}{(\chi \ukmod{0, n} \1_{\Xset^{n + 1}})^2} \\
+ \sigma^2_{n + 1} \langle (\wgtfuncmod{\ell})_{\ell = 0}^n \rangle(h_{n + 1}, f_{n + 1}, \ftd{n + 1}) + \sigma^2_{n + 1} \langle (\ukestvar{\ell})_{\ell = 0}^n \rangle(h_{n + 1}, f_{n + 1}, \ftd{n + 1}),  
\end{multline*}
which was to be established. 

Finally, it remains to check the base case $n = 1$. Indeed, since the initial particles are drawn independently from $\init$ (and $h_0 \equiv 0$, $\tstat{0} h_0 \equiv 0$, and $\eta_0 \equiv 0$), $\sigma^2_0(h_0, f_0, \ftd{0}) = \chi \{ \initwgtfunc (\ftd{0} - \chi \ftd{0}) \}^2 / (\chi \1_{\Xset_0})^2$ for all $(f_0, \ftd{0}) \in \bmf{\Xfd_0}^2$; thus, \eqref{eq:as:var:second:expression} provides, since $\chi \ukmod{0}(\addf{0} f_1 + \bar{f}_1) = 0$ and $\ukmod{0} (\addf{0} f_1 + \bar{f}_1) = \retrokmodmodnorm_{0, 1}(h_1 f_1 + \ftd{1})$,  
\begin{multline} 
\sigma^2_1(h_1, f_1, \ftd{1}) \eqdef \frac{\chi \{  \initwgtfunc \retrokmodmodnorm_{0, 1}(h_1 f_1 + \ftd{1}) \}^2}{(\chi \ukmod{0} \1_{\Xset_1})^2} + 
\frac{\post{0} \adjfuncforward{0} \post{0} \ukmod{0} \{\wgtfuncmod{0}(f_1 \tstatmod{1} h_1 + \bar{f}_1)^2 \}}{(\post{0} \ukmod{0} \1_{\Xset_1})^2} \\
+ \K^{-1} \frac{\post{0} \adjfuncforward{0} \post{0} \ukmod{0} \{\bkmod{0}(\addf{0} - \tstatmod{1} h_1)^2 f_1^2 \bkmod{0} \wgtfuncmod{0} \}}{(\post{0} \ukmod{0} \1_{\Xset_1})^2} 
+ \frac{\post{0} \adjfuncforward{0} \post{0} \ukmod{0} \{ \ukestvar{0} (f_1 \tstatmod{1} h_1 + \bar{f}_1)^2 \}}{(\post{0} \ukmod{0} \1_{\Xset_1})^2} \\
+ \K^{-1} \frac{\post{0} \adjfuncforward{0} \post{0} \ukmod{0} \{\bkmod{0}(\addf{0} - \tstatmod{1} h_1)^2 f_1^2 \bkmod{0} \ukestvar{0} \}}{(\post{0} \ukmod{0} \1_{\Xset_1})^2},  
\end{multline}
and by noting that $f_1 \tstatmod{1} h_1 + \bar{f}_1 = \retrokmodmodnorm_{1, 1} (h_1 f_1 + \ftd{1})$ we may conclude that the previous quantity coincides with \eqref{eq:non-recursive:as:var:affine} for $n = 1$. This completes the proof. 
\end{proof}


\section{Proof of Theorem~\ref{thm:bias:bound}}
\label{sec:proofs}
We preface the proof of Theorem~\ref{thm:bias:bound} by a few technical lemmas. For each $n \in \nset$ and $m \in \intvect{0}{n}$, define the kernel  
\begin{equation} \label{eq:def:retrokmod}
    \retrokmod_{m, n}(x_m', \rmd x_{0:n}) \eqdef \delta_{x_m'}(\rmd x_m) \, \tstatmod{m}(x_m, \rmd x_{0:m - 1}) \prod_{\ell = m}^{n - 1} \uk{\ell}(x_\ell, \rmd x_{\ell + 1}), 
\end{equation}
on $\Xset_m \times \Xfd^n$ (where $\tstatmod{m}$ is defined in \eqref{eq:skew:backward:law}). 

\begin{lemma} \label{lem:retro:prospective:id}
For all $m \in \intvect{0}{n}$ and $h \in \bmf{\Xfd^n}$, 
\begin{equation} \label{eq:retro:prospective:id}
\postmod{0:m} \uk{m, n - 1} h = \postmod{m} \retrokmod_{m, n} h.  
\end{equation}
\end{lemma}

\begin{proof}
The result is obtained easily by applying Lemma~\ref{lem:reversibility}(ii) to the skew model and using definition \eqref{eq:def:retrokmod}.
\end{proof}

The following probability measures play a key role in the following: for $h \in \bmf{\Xfd_m}$, 
\begin{align}
\noshift_{m, n} h &\eqdef \frac{\postmod{m} (h \times \retrokmod_{m, n} \1_{\Xset^n})}{\postmod{m} \retrokmod_{m, n} \1_{\Xset^n}}, \nonumber \\
\shiftfwd_{m, n} h &\eqdef \frac{\postmod{m} (h \times \ukmod{m} \retrokmod_{m + 1, n} \1_{\Xset^n})}{\postmod{m} \ukmod{m} \retrokmod_{m + 1, n} \1_{\Xset^n}}, \nonumber \\
\shiftbwd_{m, n} h &\eqdef \frac{\postmod{m - 1} \uk{m - 1}(h \times \retrokmod_{m, n} \1_{\Xset^n})}{\postmod{m - 1} \retrokmod_{m - 1, n} \1_{\Xset^n}}.  \nonumber 
\end{align}
In addition, for each $k \in \intvect{0}{n - 1}$, let 
\begin{equation} \label{eq:def:addf}
\addf[n]{k} : \Xset^n \ni x_{0:n} \mapsto \addf{k}(x_k, x_{k + 1}).  
\end{equation}
denote the extension of $\addf{k}$ to $\Xset^n$. 

\begin{lemma} \label{lemma:three:identities}
Let $n \in \nset$ and $m \in \intvect{1}{n}$. Then the following holds true for all $k \in \intvect{0}{n - 1}$. 
\begin{itemize}
\item[(i)]  
$
\displaystyle \noshift_{m, n} \left( \frac{\retrokmod_{m, n} \addf[n]{k}}{\retrokmod_{m, n} \1_{\Xset^n}} \right) = \frac{\postmod{m} \retrokmod_{m, n} \addf[n]{k}}{\postmod{m} \retrokmod_{m, n} \1_{\Xset^n}}$ \quad for $m \in \intvect{1}{n}$, 
\item[(ii)]  
$
\displaystyle \shiftfwd_{m, n} \left( \frac{\retrokmod_{m, n} \addf[n]{k}}{\retrokmod_{m, n} \1_{\Xset^n}} \right) = \frac{\postmod{m + 1} \retrokmod_{m + 1, n} \addf[n]{k}}{\postmod{m + 1} \retrokmod_{m + 1, n} \1_{\Xset^n}}
$ \quad for $m \in \intvect{k + 1}{n}$, and 
\item[(iii)] 
$
\displaystyle \shiftbwd_{m, n} \left( \frac{\retrokmod_{m, n} \addf[n]{k}}{\retrokmod_{m, n} \1_{\Xset^n}} \right) = \frac{\postmod{m - 1} \retrokmod_{m - 1, n} \addf[n]{k}}{\postmod{m - 1} \retrokmod_{m - 1, n} \1_{\Xset^n}}
$ \quad for all $m \in \intvect{1}{k}$. 
\end{itemize}
\end{lemma}

\begin{proof}
The identity (i) follows straightforwardly by the definition of $\noshift_{m, n}$. 

We hence turn to (ii), which is established by first noting that for all $m \in \intvect{k + 1}{n}$, 
$$
\frac{\retrokmod_{m, n} \addf[n]{k}}{\retrokmod_{m, n} \1_{\Xset^n}} 
= \bkmod[\postmod{m - 1}]{m - 1} \cdots \bkmod[\postmod{k + 1}]{k + 1} (\bkmod[\postmod{k}]{k} \addf{k}).  
$$
Now, since, by applying Lemma~\ref{lem:reversibility}(i) to the skew model,   
\begin{align}
\shiftfwd_{m, n} h 
&= \iint \frac{\postmod{m}(\rmd x_m) \, \ukmod{m}(x_m, \rmd x_{m + 1}) \, h(x_m) \retrokmod_{m + 1, n} \1_{\Xset^n}(x_{m + 1})}{\postmod{m} \ukmod{m} \retrokmod_{m + 1, n} \1_{\Xset^n}} \nonumber \\
&= \iint \frac{\postmod{m} \ukmod{m}(\rmd x_{m + 1}) \, \bkmod[\postmod{m}]{m}(x_{m + 1}, \rmd x_m) \, h(x_m) \retrokmod_{m + 1, n} \1_{\Xset^n}(x_{m + 1})}{\postmod{m} \ukmod{m} \retrokmod_{m + 1, n} \1_{\Xset^n}} \nonumber \\
&= \frac{\postmod{m + 1} (\bkmod[\postmod{m}]{m} h \times \retrokmod_{m + 1, n} \1_{\Xset^n})}{\postmod{m + 1} \retrokmod_{m + 1, n} \1_{\Xset^n}}, \nonumber 
\end{align}
we may establish the identity by proceeding like  
\begin{align*}
\lefteqn{\shiftfwd_{m, n} \left( \frac{\retrokmod_{m, n} \addf[n]{k}}{\retrokmod_{m, n} \1_{\Xset^n}} \right)} \\
&= \iint \frac{\postmod{m + 1}(\rmd x_{m + 1}) \, \bkmod[\postmod{m}]{m}(x_{m + 1}, \rmd x_m) \, \bkmod[\postmod{m - 1}]{m - 1}  
\cdots \bkmod[\postmod{k + 1}]{k + 1}(\bkmod[\postmod{k}]{k} \addf{k})(x_m) 
\retrokmod_{m + 1, n} \1_{\Xset^n}(x_{m + 1})}{\postmod{m + 1} \retrokmod_{m + 1, n} \1_{\Xset^n}} \\
&= \frac{\postmod{m + 1} \retrokmod_{m + 1, n} \addf[n]{k}}{\postmod{m + 1} \retrokmod_{m + 1, n} \1_{\Xset^n}}.  
\end{align*}

Finally, to check (iii), note that for all $m \in \intvect{1}{k}$, 
$$
\uk{m - 1} \retrokmod_{m, n} \addf[n]{k} = \retrokmod_{m - 1, n} \addf[n]{k}. 
$$
Thus, in this case 
$$
\shiftbwd_{m, n} \left( \frac{\retrokmod_{m, n} \addf[n]{k}}{\retrokmod_{m, n} \1_{\Xset^n}} \right) = \frac{\postmod{m - 1} \uk{m - 1} \retrokmod_{m, n} \addf[n]{k}}{\postmod{m - 1} \retrokmod_{m - 1, n} \1_{\Xset^n}} = \frac{\postmod{m - 1} \retrokmod_{m - 1, n} \addf[n]{k}}{\postmod{m - 1} \retrokmod_{m - 1, n} \1_{\Xset^n}}. 
$$
\end{proof}

\begin{lemma} \label{lem:geo:bound}
Assume \hypref{assum:strong:mixing}. Then for all $n \in \nset$, $m \in \intvect{0}{n}$, $k \in \intvect{0}{n - 1}$, and $(\lambda, \lambda') \in \probmeas{\Xfd_m}^2$, 
$$
\left|\frac{\lambda \retrokmod_{m, n} \addf[n]{k}}{\lambda \retrokmod_{m, n} \1_{\Xset^n}} - \frac{\lambda' \retrokmod_{m, n} \addf[n]{k}}{\lambda' \retrokmod_{m, n} \1_{\Xset^n}} \right| \leq \| \addf{k} \|_\infty \rho^{|k - m| - 1}. 
$$
\begin{proof}
First, assume that $m \leq k$; then note that for all $x_m \in \Xset_m$, 
\begin{multline} \label{eq:retrokmodnorm:vs:forward:kernel}
\frac{\retrokmod_{m, n} \addf[n]{k}(x_m)}{\retrokmod_{m, n} \1_{\Xset^n}(x_m)} - \frac{\retrokmod_{m, n} \addf[n]{k}(x_m)}{\retrokmod_{m, n} \1_{\Xset^n}(x_m)}
 = \fk{m}{n} \cdots \fk{k - 1}{n} (\fk{k}{n} \addf{k})(x_m) \\ 
- \fk{m}{n} \cdots \fk{k - 1}{n} (\fk{k}{n} \addf{k})(x_m),    
\end{multline}

where we have introduced the \emph{forward kernels}
$$
\fk{m}{n} h (x_m) \eqdef \frac{\uk{m}(h \times \retrokmod_{m  + 1, n} \1_{\Xset^n})(x_m)}{\retrokmod_{m, n} \1_{\Xset^n} (x_m)}, \quad x_m \in \Xset_m, \quad h \in \bmf{\Xfd_{m + 1}}. 
$$
Under \hypref{assum:strong:mixing}, each forward kernel satisfies a global Doeblin condition in the form of the uniform lower bound 
$$
\fk{m}{n} h (x_m) \geq \frac{\udlow}{\udup} \mu_{m, n} h,  
$$
where we have defined the probability measure 
$$
\mu_{m, n} h \eqdef \frac{\mu_{m + 1}(h \times \retrokmod_{m  + 1, n} \1_{\Xset^n})}{\mu_{m + 1} \retrokmod_{m  + 1, n}  \1_{\Xset^n}}, \quad h \in \bmf{\Xfd_{m + 1}} 
$$
(where $\mu_{m + 1}$ is the reference measure introduced in Section~\ref{sec:model}). Thus, by standard results for uniformly minorised Markov chains (see, \emph{e.g.}, \cite[Lemma~4.3.13]{Cappe:2005:IHM:1088883}), the Dobrushin coefficient of each $\fk{m}{n}$ is bounded by $\rho = 1 - \udlow / \udup$. Thus, \eqref{eq:retrokmodnorm:vs:forward:kernel} implies that 
\begin{align*}
\left| \frac{\lambda \retrokmod_{m, n} \addf[n]{k}}{\lambda \retrokmod_{m, n} \1_{\Xset^n}} - \frac{\lambda' \retrokmod_{m, n} \addf[n]{k}}{\lambda' \retrokmod_{m, n} \1_{\Xset^n}} \right| &= \left| (\lambda_{m, n} - \lambda_{m, n}') \fk{m}{n} \cdots \fk{k - 1}{n} (\fk{k}{n} \addf{k})
\right| \\ 
&\leq \rho^{k - m} \| \fk{k}{n} \addf{k} \|_\infty \leq \rho^{k - m} \| \addf{k} \|_\infty, 
\end{align*}
where for $h \in \bmf{\Xfd_m}$, 
$$
\lambda_{m, n} h \eqdef \frac{\lambda(h \times \retrokmod_{m, n} \1_{\Xset^n})}{\lambda \retrokmod_{m, n} \1_{\Xset^n}}, \quad 
\lambda_{m, n}' h \eqdef \frac{\lambda'(h \times \retrokmod_{m, n} \1_{\Xset^n})}{\lambda' \retrokmod_{m, n} \1_{\Xset^n}}. 
$$
Now, assume that $m > k$; then note that for all $x_m \in \Xset_m$, 
\begin{multline} \label{eq:retrokmodnorm:vs:backward:kernel}
\frac{\retrokmod_{m, n} \addf[n]{k}(x_m)}{\retrokmod_{m, n} \1_{\Xset^n}(x_m)} - \frac{\retrokmod_{m, n} \addf[n]{k}(x_m)}{\retrokmod_{m, n} \1_{\Xset^n}(x_m)} \\
 = \bkmod{m - 1} \cdots \bkmod{k + 1}(\bkmod{k} \addf{k})(x_m) - \bkmod{m - 1} \cdots \bkmod{k + 1}(\bkmod{k} \addf{k})(x_m).     
\end{multline}
Under \hypref{assum:strong:mixing}, also each backward kernel satisfies a Doeblin condition, namely   
$$
\bkmod{m} h (x_{m + 1}) \geq \frac{\udlow}{\udup} \postmod{m} h,  
$$
with the marginal $\postmod{m}$ playing the role as minorising measure. Thus, the backward kernel Dobrushin coefficients are bounded by the same constant $\rho$, implying, via \eqref{eq:retrokmodnorm:vs:backward:kernel}, that 
\begin{align*}
\left| \frac{\lambda \retrokmod_{m, n} \addf[n]{k}}{\lambda \retrokmod_{m, n} \1_{\Xset^n}} - \frac{\lambda' \retrokmod_{m, n} \addf[n]{k}}{\lambda' \retrokmod_{m, n} \1_{\Xset^n}} \right| 
&= \left| (\lambda_{m, n} - \lambda_{m, n}') \bkmod{m - 1} \cdots \bkmod{k + 1}(\bkmod{k} \addf{k}) \right| \\
&\leq \rho^{m - k - 1} \| \addf{k} \|_\infty. 
\end{align*}
This completes the proof. 
\end{proof}
\end{lemma}

\begin{lemma} \label{lem:diff:bound} 
Assume \hypref{assum:strong:mixing} and \hypref{assum:bias:bound}. Then the following holds. 
\begin{itemize}
\item[(i)] For all $n \in \nset$, $m \in \intvect{0}{n}$, $\precpar \in \precparsp$, and $h \in \bmf{\Xfd^n}$,   
$$
\left| \noshift_{m, n} h - \shiftbwd_{m, n} h \right| \vee \left| \shiftfwd_{m, n} h - \noshift_{m, n} h \right| \leq 2 c \precpar \frac{\udup}{\udlow^2} \| h \|_\infty.   
$$
\item[(ii)] For all $n \in \nset$, $m \in \intvect{0}{n - 1}$, and $\precpar \in \precparsp$,
$$
\left| \frac{\postmod{m + 1} \retrokmod_{m + 1, n} \addf[n]{m}}{\postmod{m + 1} \retrokmod_{m + 1, n} \1_{\Xset^n}} - \frac{\postmod{m} \retrokmod_{m, n} \addf[n]{m}}{\postmod{m} \retrokmod_{m, n} \1_{\Xset^n}} \right| \leq 2 c \precpar \frac{\udup}{\udlow^2} \| h_m \|_\infty. 
$$
\end{itemize}
In both cases, $c$ is the constant in \hypref{assum:bias:bound}. 
\end{lemma}

\begin{proof}
We start with (i). Combining the decomposition 
\begin{multline*}
 \noshift_{m, n} h - \shiftbwd_{m, n} h = \varphi_{m, n}^\precpar h \left( \frac{\postmod{m - 1} \uk{m - 1} \retrokmod_{m, n} \1_{\Xset^n} - \postmod{m - 1} \ukmod{m - 1} \retrokmod_{m, n} \1_{\Xset^n}}{\postmod{m - 1} \retrokmod_{m - 1, n} \1_{\Xset^n}} \right) \\
+ \frac{\postmod{m - 1} \ukmod{m - 1}(h \times \retrokmod_{m, n} \1_{\Xset^n}) - \postmod{m - 1} \uk{m - 1}(h \times \retrokmod_{m, n} \1_{\Xset^n})}{\postmod{m - 1} \retrokmod_{m - 1, n} \1_{\Xset^n}} 
\end{multline*}
with \hypref{assum:bias:bound} provides the bound 
$$
\left| \noshift_{m, n} h - \shiftbwd_{m, n} h \right| \leq 2 c \precpar \frac{\| \retrokmod_{m, n} \1_{\Xset^n} \|_\infty  \| h \|_\infty}{\postmod{m - 1} \retrokmod_{m - 1, n} \1_{\Xset^n}}
$$
(where $c$ is the constant in \hypref{assum:bias:bound}). Now, since for all $x_m \in \Xset_m$,  
\begin{align*}
\retrokmod_{m, n} \1_{\Xset^n}(x_m) &= \int \ud{m}(x_m, x_{m + 1}) \retrokmod_{m + 1, n} \1_{\Xset^n}(x_{m + 1}) \, \mu_{m + 1}(\rmd x_{m + 1}) \\
&\leq \udup \mu_{m + 1} \retrokmod_{m + 1, n} \1_{\Xset^n} 
\end{align*}
and 
\begin{align*}
\postmod{m - 1} \retrokmod_{m - 1, n} \1_{\Xset^n} 
&= \iint \postmod{m - 1}(\rmd x_{m - 1}) \, \ud{m - 1}(x_{m - 1}, x_m) \ud{m}(x_m, x_{m + 1}) \\ 
&\hspace{30mm} \times \retrokmod_{m + 1, n} \1_{\Xset^n}(x_{m + 1}) \, \mu_m \tensprod \mu_{m + 1}(\rmd x_{m:m + 1}) \\
&\geq \udlow^2 \mu_{m + 1} \retrokmod_{m + 1, n} \1_{\Xset^n}, 
\end{align*}
implying the bound
\begin{equation} \label{eq:retrokmod:ratio:bd}
 \frac{\| \retrokmod_{m, n} \1_{\Xset^n} \|_\infty}{\postmod{m - 1} \uk{m - 1} \retrokmod_{m, n} \1_{\Xset^n}} \leq \frac{\udup}{\udlow^2}, 
\end{equation}
we may conclude that 
$| \noshift_{m, n} h - \shiftbwd_{m, n} h| \leq 2 c \precpar \udup \| h \|_\infty / \udlow^2
$. 
Along similar lines, the second difference can be bounded by the same quantity by writing   
\begin{multline*}
\shiftfwd_{m, n} h - \noshift_{m, n} h = \shiftfwd_{m, n} h \left( \frac{\postmod{m} \uk{m} \retrokmod_{m + 1, n} \1_{\Xset^n} - \postmod{m} \ukmod{m} \retrokmod_{m + 1, n} \1_{\Xset_n}}{\postmod{m} \uk{m} \retrokmod_{m + 1, n} \1_{\Xset^n}} \right) \\
+ \frac{\postmod{m} (h \times \ukmod{m} \retrokmod_{m + 1, n} \1_{\Xset^n}) - \postmod{m} (h \times \uk{m} \retrokmod_{m + 1, n} \1_{\Xset_n})}{\postmod{m} \uk{m} \retrokmod_{m + 1, n} \1_{\Xset^n}} 
\end{multline*}
and reapplying \hypref{assum:bias:bound} and \eqref{eq:retrokmod:ratio:bd}.  

We turn to (ii). By definition \eqref{eq:def:retrokmod},  
$$
\retrokmod_{m + 1, n} \addf[n]{m}(x_{m + 1}) = \int \addf{m}(x_m, x_{m + 1}) \, \bkmod{m}(x_{m + 1}, \rmd x_m) \, \retrokmod_{m + 1, n} \1_{\Xset^n}(x_{m + 1});
$$
thus, by applying Lemma~\ref{lem:reversibility}(i) to the skew model, 
\begin{align}
\lefteqn{\frac{\postmod{m + 1} \retrokmod_{m + 1, n} \addf[n]{m}}{\postmod{m + 1} \retrokmod_{m + 1, n} \1_{\Xset^n}}} \\
&= \iint \frac{\postmod{m} \ukmod{m}(\rmd x_{m + 1}) \, \addf{m}(x_m, x_{m + 1}) \, \bkmod{m}(x_{m + 1}, \rmd x_m) \, \retrokmod_{m + 1, n} \1_{\Xset^n}(x_{m + 1})}{\postmod{m} \ukmod{m} \retrokmod_{m + 1, n} \1_{\Xset^n}} \nonumber \\
&= \iint \frac{\postmod{m}(\rmd x_m) \, \ukmod{m}(x_m, \rmd x_{m + 1}) \, \addf{m}(x_m, x_{m + 1}) \retrokmod_{m + 1, n} \1_{\Xset^n}(x_{m + 1})}{\postmod{m} \ukmod{m} \retrokmod_{m + 1, n} \1_{\Xset^n}}. \nonumber
\end{align}
We may thus decompose the quantity under consideration as  
\begin{multline*}
\frac{\postmod{m + 1} \retrokmod_{m + 1, n} \addf[n]{m}}{\postmod{m + 1} \retrokmod_{m + 1, n} \1_{\Xset^n}} - \frac{\postmod{m} \retrokmod_{m, n} \addf[n]{m}}{\postmod{m} \retrokmod_{m, n} \1_{\Xset^n}} \\
= \frac{\postmod{m + 1} \retrokmod_{m + 1, n} \addf[n]{m}}{\postmod{m + 1} \retrokmod_{m + 1, n} \1_{\Xset^n}} \left( \frac{\postmod{m} \uk{m} \retrokmod_{m + 1, n} \1_{\Xset^n} - \postmod{m} \ukmod{m} \retrokmod_{m + 1, n} \1_{\Xset^n}}{\postmod{m} \retrokmod_{m, n} \1_{\Xset^n}} \right) \\
+ \int \frac{\postmod{m}(\rmd x_m) \{ \ukmod{m}(\addf{m} \retrokmod_{m + 1, n} \1_{\Xset^n})(x_m) - \uk{m} (\addf{m} \retrokmod_{m + 1, n} \1_{\Xset^n})(x_m)\}}{\postmod{m} \retrokmod_{m, n} \1_{\Xset^n}} ,  
\end{multline*}
from which (ii) follows, as before, by a combination of \hypref{assum:bias:bound} and \eqref{eq:retrokmod:ratio:bd}.  
\end{proof}

\begin{proof}[Proof of Theorem~\ref{thm:bias:bound}]
Write 
$$
\postmod{0:n} h_n - \post{0:n} h_n = \sum_{k = 0}^{n - 1} \left( \postmod{0:n} \addf[n]{k} - \post{0:n} \addf[n]{k} \right), 
$$
where each term can be decomposed according to   
\begin{equation*}
\postmod{0:n} \addf[n]{k} - \post{0:n} \addf[n]{k} = 
\sum_{m = 1}^n \left( \frac{\postmod{0:m} \uk{m, n - 1} \addf[n]{k}}{\postmod{0:m} \uk{m, n - 1} \1_{\Xset^n}} - \frac{\postmod{0:m - 1} \uk{m - 1, n - 1} \addf[n]{k}}{\postmod{0:m - 1} \uk{m - 1, n - 1} \1_{\Xset^n}} \right)  
\end{equation*}
(recall that $\postmod{0} \propto \chi$). In order to bound each term of this decomposition, write, using Lemma~\ref{lem:retro:prospective:id},  
$$
\frac{\postmod{0:m} \uk{m, n - 1} \addf[n]{k}}{\postmod{0:m} \uk{m, n - 1} \1_{\Xset^n}} - \frac{\postmod{0:m - 1} \uk{m - 1, n - 1} \addf[n]{k}}{\postmod{0:m - 1} \uk{m - 1, n - 1} \1_{\Xset^n}} 
= \frac{\postmod{m} \retrokmod_{m, n} \addf[n]{k}}{\postmod{m} \retrokmod_{m, n} \1_{\Xset^n}} - \frac{\postmod{m - 1} \retrokmod_{m - 1, n} \addf[n]{k}}{\postmod{m - 1} \retrokmod_{m - 1, n} \1_{\Xset^n}}. 
$$
Now, for all $m \in \intvect{1}{n}$, pick an arbitrary element $\xarb_m \in \Xset_m$ and define the kernel 
\begin{equation} \label{eq:def:norm:objective:func}
\retrokmodnorm_{m, n} h(x_m) \eqdef \frac{\retrokmod_{m, n} h(x_m)}{\retrokmod_{m, n} \1_{\Xset^n}(x_m)} - \frac{\retrokmod_{m, n} h (\xarb_m)}{\retrokmod_{m, n} \1_{\Xset^n} (\xarb_m)}, \quad x_m \in \Xset_m, \quad h \in \bmf{\Xfd^n}. 
\end{equation}
Combining this definition with Lemma~\ref{lemma:three:identities}, we may express the quantity of interest as 
\begin{multline*}
\postmod{0:n} \addf[n]{k} - \post{0:n} \addf[n]{k} = \sum_{m = 1}^k \left( \noshift_{m, n} \retrokmodnorm_{m, n} \addf[n]{k} - \shiftbwd_{m, n}  \retrokmodnorm_{m, n} \addf[n]{k} \right) \\ 
+ \frac{\postmod{k + 1} \retrokmod_{k + 1, n} \addf[n]{k}}{\postmod{k + 1} \retrokmod_{k + 1, n} \1_{\Xset^n}} - \frac{\postmod{k} \retrokmod_{k, n} \addf[n]{k}}{\postmod{k} \retrokmod_{k, n} \1_{\Xset^n}} + \sum_{m = k + 1}^{n - 1} \left( \shiftfwd_{m, n} \retrokmodnorm_{m, n} \addf[n]{k} - \noshift_{m, n} \retrokmodnorm_{m, n} \addf[n]{k} \right). 
\end{multline*}
Now, applying Lemmas~\ref{lem:geo:bound} and \ref{lem:diff:bound} to the previous decomposition yields
$$
\left| \postmod{0:n} \addf[n]{k} - \post{0:n} \addf[n]{k} \right| \leq 2 c \precpar \frac{\udup }{\udlow^2} \sum_{k = 0}^{n - 1} \| \addf{k} \|_\infty \left( \sum_{m = 1}^{n - 1} \rho^{|k - m| - 1} + 1 \right).  
$$
which was to be established. 

The second inequality follows straightforwardly according to 
\begin{align*}
\sum_{k = 0}^{n - 1} \| \addf{k} \|_\infty \left( \sum_{m = 1}^{n - 1} \rho^{|k - m| - 1} + 1 \right) 
&\leq \left( n + \frac{1}{\rho} \left( n + 2 \sum_{\ell = 1}^{n - 1} (n - \ell) \rho^\ell \right) \right) \sup_{k \in \intvect{0}{n - 1}} \| \addf{k} \|_\infty \\ 
&\leq n \left( 1 + \frac{1}{\rho} + \frac{2}{1 - \rho} \right) \sup_{k \in \intvect{0}{n - 1}} \| \addf{k} \|_\infty.  
\end{align*}

\end{proof}

\section{Proofs of Theorem~\ref{thm:variance:bound}, Corollary~\ref{cor:variance:bound}, and Proposition~\ref{prop:variance:bound:filter}}
\label{sec:variance:bounds}
\begin{proof}[Proof of Theorem~\ref{thm:variance:bound}]
The following is a refinement of the proof of \cite[Proposition~7]{olsson:westerborn:2017}. 

We start bounding the first term of \eqref{eq:non-recursive:as:var}. For this purpose, we first prove that for all $m \in \intvect{0}{n - 1}$, 
\begin{equation} \label{eq:retrokmodmodnorm:bound}
\|\retrokmodmodnorm_{m, n} h_n \|_\infty \leq \| \ukmod{m, n - 1} \1_{\Xset^n} \|_\infty \sum_{k = 0}^{n - 1} \|\addf{k} \|_\infty \rho^{|k - m| - 1}. 
\end{equation}
In order to establish \eqref{eq:retrokmodmodnorm:bound}, first note that  
\begin{equation} \label{eq:post:retrok:id}
\postmod{0:n} h_n = \frac{\postmod{m} \retrokmodmod_{m, n} h_n}{\postmod{m} \retrokmodmod_{m, n} \1_{\Xset^n}}.  
\end{equation}
Thus, using the extensions \eqref{eq:def:addf} we may write, for $x_m \in \Xset_m$, 
$$
\frac{\retrokmodmodnorm_{m, n} h_n(x_m)}{\ukmod{m, n - 1} \1_{\Xset^n}(x_m)} = \sum_{k = 0}^{n - 1} \left( \frac{\retrokmodmod_{m, n} \addf[n]{k}(x_m)}{\retrokmodmod_{m, n} \1_{\Xset^n}(x_m)} - \frac{\postmod{m} \retrokmodmod_{m, n} \addf[n]{k}}{\postmod{m} \retrokmodmod_{m, n} \1_{\Xset^n}} \right). 
$$
Thus, the bound \eqref{eq:retrokmodmodnorm:bound} is obtained by applying Lemma~\ref{lem:geo:bound} (twice: first, with $\precpar = 0$, allowing $\retrokmod_{m, n}$ in Lemma~\ref{lem:geo:bound} to be replaced by $\retrok_{m, n}$; second, with the skew model playing the role of the original model, allowing $\retrok_{m, n}$ to be replaced by $\retrokmodmod_{m, n}$), with $\lambda = \delta_{x_m}$ and $\lambda' = \postmod{m}$, to each term on the right-hand side of the previous identity. 

Now, using \eqref{eq:retrokmodmodnorm:bound} and the bound \eqref{eq:retrokmod:ratio:bd}, 
\begin{multline} \label{eq:first:bound:first:term}
\sum_{m = 0}^{n - 1} \frac{\postmod{m} \adjfuncforward{m} \postmod{m} \ukmod{m}(\varphi_m [\retrokmodmodnorm_{m + 1, n} h_n ]^2)}{(\postmod{m} \ukmod{m, n - 1} \1_{\Xset^n})^2} \\ 
\leq \sup_{\ell \in \nset} \| \adjfuncforward{\ell} \|_\infty \sup_{\ell \in \nset} \| \varphi_\ell \|_\infty \frac{\udup^3}{\udlow^4}
\sum_{m = 0}^{n - 1} \left( \sum_{k = 0}^{n - 1} \| \addf{k} \|_\infty \rho^{|k - m - 1| - 1} \right)^2. 
\end{multline}
Then, since 
$$
\sum_{m = 0}^{n - 1} \left( \sum_{k = 0}^{n - 1} \rho^{|k - m - 1| - 1} \right)^2 \leq \frac{4 n}{\rho^2 (1 - \rho)^2}, 
$$
we obtain the linear (in $n$) bound
\begin{multline} \label{eq:linear:bound:first:bound}
\sum_{m = 0}^{n - 1} \frac{\postmod{m} \adjfuncforward{m} \postmod{m} \ukmod{m}(\varphi_m [\retrokmodmodnorm_{m + 1, n} h_n ]^2)}{(\postmod{m} \ukmod{m, n - 1} \1_{\Xset^n})^2} \\ 
\leq 
\sup_{\ell \in \nset} \| \addf{\ell} \|_\infty^2 \sup_{\ell \in \nset} \| \adjfuncforward{\ell} \|_\infty \sup_{\ell \in \nset} \| \varphi_\ell \|_\infty  \frac{4 n}{\rho^2 (1 - \rho)^5 \udlow }. 
\end{multline}

We turn to the second term of \eqref{eq:non-recursive:as:var} and reusing \eqref{eq:retrokmod:ratio:bd} yields 
\begin{align} 
\frac{\ukmod{\ell + 1, m}( \bkmod{m} \varphi_m [\ukmod{m + 1, n - 1} \1_{\Xset^n}]^2)}{(\postmod{\ell} \ukmod{\ell, m - 1} \1_{\Xset^m})(\postmod{m} \ukmod{m, n - 1} \1_{\Xset^n})^2} 
&\leq \frac{\| \ukmod{\ell + 1, m - 1} \1_{\Xset^m} \|_\infty \|\ukmod{m + 1, n - 1} \1_{\Xset^n} \|_\infty^2 \|\varphi_m \|_\infty \udup}{(\postmod{\ell} \ukmod{\ell, m - 1} \1_{\Xset^m})(\postmod{m} \ukmod{m, n - 1} \1_{\Xset^n})^2} \nonumber \\ 
&\leq \frac{1}{(1 - \rho)^4 \udlow^2} 
\|\varphi_m \|_\infty. \label{eq:linear:bound:second:bound}
\end{align}
By \eqref{eq:forward:smoothing} it holds that  
\begin{align*}
\lefteqn{\tstatmod{\ell} h_{\ell}(x_\ell) + \addf{\ell}(x_\ell, x_{\ell + 1}) - \tstatmod{\ell + 1} h_{\ell +1}(x_{\ell + 1})} \hspace{10mm} \\ 
&= \tstatmod{\ell} h_\ell(x_\ell) - \bkmod{\ell}\tstatmod{\ell} h_\ell(x_{\ell + 1}) + \addf{\ell}(x_\ell, x_{\ell + 1}) - \bkmod{\ell} \addf{\ell}(x_{\ell + 1}) \\
&= \sum_{k = 0}^{\ell - 1} \left( \retrokmodmod_{\ell, \ell} \addf[\ell]{k}(x_\ell) - \bkmod{\ell} \retrokmodmod_{\ell, \ell} \addf[\ell]{k}(x_{\ell + 1}) \right) + \addf{\ell}(x_\ell, x_{\ell + 1}) - \bkmod{\ell} \addf{\ell}(x_{\ell + 1}), 
\end{align*}
and by applying Lemma~\ref{lem:geo:bound} (with $\lambda = \delta_{x_\ell}$ and $\lambda' = \bkmod{\ell}(x_{\ell + 1}, \cdot)$) to each term in the sum yields the bound 
\begin{equation} \label{eq:linear:bound:third:bound}
\| \tstatmod{\ell} h_{\ell} + \addf{\ell} - \tstatmod{\ell + 1} h_{\ell +1} \|_\infty \leq \left( \frac{1}{1 - \rho} + 2 \right) \sup_{\ell \in \nset} \|\addf{\ell} \|_\infty. 
\end{equation}
By combining \eqref{eq:linear:bound:second:bound} and \eqref{eq:linear:bound:third:bound} we obtain 
\begin{multline} 
 \sum_{m = 0}^{n - 1} \sum_{\ell = 0}^m \frac{\postmod{m} \adjfuncforward{m} \postmod{\ell} \ukmod{\ell} \{\bkmod{\ell}(\tstatmod{\ell} h_{\ell} + \addf{\ell} - \tstatmod{\ell + 1} h_{\ell +1})^2 \ukmod{\ell + 1, m}( \bkmod{m} \varphi_m [\ukmod{m + 1, n - 1} \1_{\Xset^n}]^2
)\}}{\K^{m - \ell + 1} (\postmod{\ell} \ukmod{\ell, m - 1} \1_{\Xset^m})(\postmod{m} \ukmod{m, n - 1} \1_{\Xset^n})^2} \\
\leq \frac{1}{(1 - \rho)^4 \udlow^2}  \left( \frac{1}{1 - \rho} + 2 \right)^2 \sup_{\ell \in \nset} \| \addf{\ell} \|_\infty^2 \sup_{\ell \in \nset} \| \adjfuncforward{\ell} \|_\infty \sup_{\ell \in \nset} \| \varphi_\ell \|_\infty \sum_{m = 0}^{n - 1} \sum_{\ell = 0}^m \K^{\ell - (m + 1)}. 
\end{multline}
Finally, since, for $\K \geq 2$, 
$$
\lim_{m \to \infty} \sum_{\ell = 0}^m \K^{\ell - (m + 1)} = \frac{1}{\K - 1}, 
$$
taking the Ces\'{a}ro mean provides  
\begin{multline} \label{eq:eq:linear:bound:fourth:bound}
\hspace{-10mm} \limsup_{n \to \infty} \frac{1}{n}  \sum_{m = 0}^{n - 1} \sum_{\ell = 0}^m \frac{\postmod{m} \adjfuncforward{m} \postmod{\ell} \ukmod{\ell} \{\bkmod{\ell}(\tstatmod{\ell} h_{\ell} + \addf{\ell} - \tstatmod{\ell + 1} h_{\ell +1})^2 \ukmod{\ell + 1, m}( \bkmod{m} \varphi_m [\ukmod{m + 1, n - 1} \1_{\Xset^n}]^2
)\}}{\K^{m - \ell + 1} (\postmod{\ell} \ukmod{\ell, m - 1} \1_{\Xset^m})(\postmod{m} \ukmod{m, n - 1} \1_{\Xset^n})^2} \\ 
\leq \frac{1}{(\K - 1) (1 - \rho)^4 \udlow^2}  \left( \frac{1}{1 - \rho} + 2 \right)^2 \sup_{\ell \in \nset} \| \addf{\ell} \|_\infty^2 \sup_{\ell \in \nset} \| \adjfuncforward{\ell} \|_\infty \sup_{\ell \in \nset} \| \varphi_\ell \|_\infty. 
\end{multline}
Finally, we complete the proof by combining \eqref{eq:linear:bound:first:bound} and \eqref{eq:eq:linear:bound:fourth:bound}. 
\end{proof}

\begin{proof}[Proof of Corollary~\ref{cor:variance:bound}]
First, note that by \eqref{eq:retrokmodmodnorm:bound}, 
$$
\|\retrokmodmodnorm_{0, n} h_n \|_\infty \leq \| \ukmod{0, n - 1} \1_{\Xset^n} \|_\infty \sup_{\ell \in \nset} \|\addf{\ell} \|_\infty  \frac{1}{\rho(1 - \rho)}. 
$$
Moreover, since for all $x_0 \in \Xset_0$, 
\begin{equation} \label{eq:ukmod:zero:upper:lower:bound}
\udlow \mu_1 \ukmod{1, n - 1} \1_{\Xset^n} \leq \ukmod{0, n - 1} \1_{\Xset^n}(x_0) \leq \udup \mu_1 \ukmod{1, n - 1} \1_{\Xset^n}, 
\end{equation}
it holds that 
\begin{equation}\label{eq:bd:inital:term}
\frac{\|\retrokmodmodnorm_{0, n} h_n \|_\infty^2}{(\chi \ukmod{0, n - 1} \1_{\Xset^n})^2} \leq \sup_{\ell \in \nset} \|\addf{\ell} \|_\infty^2  \frac{1}{\rho^2(1 - \rho)^4 (\chi \1_{\Xset_0})^2}. 
\end{equation}
The desired bound is now completed by applying Theorem~\ref{thm:variance:bound} and \eqref{eq:bd:inital:term} to the two first and the last terms of $\sigma_n^2(h_n)$, respectively.
\end{proof}

\begin{proof}[Proof of Proposition~\ref{prop:variance:bound:filter}]
The result follows straightforwardly from the bound \eqref{eq:first:bound:first:term}; indeed, as the mapping $\Xset^n \ni x_{0:n} \mapsto f(x_n)$ belongs to $\bmaf{\Xfd^n}$, \eqref{eq:first:bound:first:term} implies that 
\begin{multline*}
\sum_{m = 0}^{n - 1} \frac{\postmod{m} \adjfuncforward{m} \postmod{m} \ukmod{m}(\varphi_m [\ukmod{m + 1} \cdots \ukmod{n - 1} (f - \postmod{n} f)
]^2)}{(\postmod{m} \ukmod{m} \cdots \ukmod{n - 1} \1_{\Xset_n})^2} \\ 
\leq \sup_{\ell \in \nset} \| \adjfuncforward{\ell} \|_\infty \sup_{\ell \in \nset} \| \varphi_\ell \|_\infty \|f \|_\infty^2 \frac{\udup^3}{\udlow^4} \rho^{-2}
\sum_{m = 0}^{n - 1} \rho^{2 |n - m - 2|},  
\end{multline*}
where the sum on the right-hand side can be bounded by $1 / (1 - \rho^2) + \rho^2$. In addition, by \eqref{eq:retrokmodmodnorm:bound}, 
$$
\| \ukmod{0} \cdots \ukmod{n - 1} (f - \postmod{n} f) \|_\infty \leq \| \ukmod{0} \cdots \ukmod{n - 1} \1_{\Xset^n} \|_\infty \|f \|_\infty \rho^{n - 2}, 
$$
and reusing \eqref{eq:ukmod:zero:upper:lower:bound} yields 
$$
\frac{\| \ukmod{0} \cdots \ukmod{n - 1} (f - \postmod{n} f) \|_\infty^2}{(\chi \ukmod{0} \cdots \ukmod{n - 1} \1_{\Xset_n})^2} \leq \| f \|_\infty^2 \frac{\rho^{2n}}{\rho^4(1 - \rho)^2 (\chi \1_{\Xset_0})^2}. 
$$
The proof is completed by combining these bounds. 
\end{proof}

\section{A technical lemma}
\label{sec:tech:results}
The following technical lemma is a straightforward adaption of \cite[Lemma~14]{olsson:westerborn:2014b} to the framework of Section~\ref{sec:preliminaries}. 

\begin{lemma}
\label{lem:generalized:lebesgue}
Assume \hypref[assum:biased:estimate]{assum:bound:filter:pseudomarginal} and let $\kernel{\Psi}$ be some possibly unnormalised transition kernel on $\Xset_n \times \Xfd_{n + 1}$ having transition density in $\bmf{\Xfd_n \tensprod \Xfd_{n + 1}}$ with respect to some reference measure on $\Xfd_{n + 1}$. Moreover, let $(\varphi_\N)_{\N \in \nset}$ be a sequence of functions in $\bmf{\Xfd_{n + 1}}$  for which 
\begin{itemize}
\item[(i)] there exists $\varphi \in \bmf{\Xfd_{n + 1}}$ such that for all $x \in \Xset_{n + 1}$, $\lim_{\N \to \infty} \varphi_\N(x) = \varphi(x)$ $\pP$-a.s. and 
\item[(ii)] there exists $c \in \rsetpos$ such that $\|\varphi \|_\infty \leq c$ for all $\N \in \nset$. 
\end{itemize} 
Then $\post[\N]{n} \kernel{\Psi} \varphi_\N \pplim \post{n} \kernel{\Psi} \varphi$ as $\N \to \infty$. 
\end{lemma}

\end{document}